\documentclass[12pt]{article}
\usepackage{amsmath,amssymb,url,mathrsfs, nicefrac, color}
\usepackage[capposition=bottom]{floatrow}
\usepackage{mwe}
\usepackage{enumerate}
\usepackage{subcaption}

\usepackage[skins]{tcolorbox}
\usepackage{lipsum}
\definecolor{myred}{RGB}{213,94,0}
\definecolor{mygreen}{RGB}{0,158,115}
\definecolor{myblue}{rgb}{0,0,0.75}
\definecolor{bcblue}{RGB}{0,40,72}
\usepackage{xcolor}

\usepackage{setspace}
\usepackage[plainpages=false,hypertexnames=false]{hyperref}
\hypersetup{
  colorlinks   = true,    %
	urlcolor     = myblue,    %
	  linkcolor    = myblue,    %
		citecolor    = myblue %
	}
\usepackage{tikz}
\usetikzlibrary{calc,shapes,decorations.pathreplacing, arrows, knots}
\usepackage{pgfplots}
\pgfplotsset{compat=1.16} %
\usepgfplotslibrary{fillbetween}

\usetikzlibrary{patterns}
\usetikzlibrary{arrows.meta} %
\usepackage{amsthm}
\makeatletter
\def\th@plain{%
\thm@notefont{}%
  \itshape %
}
\def\th@definition{%
  \thm@notefont{}%
	\normalfont %
}
\makeatother
\usepackage{mathtools}
\usepackage[english]{babel}
\usepackage{color}
\usepackage{bm}
\usepackage{accents}	
\usepackage{graphicx}
\usepackage[semicolon]{natbib}
\theoremstyle{plain}
\newtheorem{theorem}{Theorem}
\newtheorem{lemma}{Lemma}
\newtheorem{proposition}{Proposition}
\newtheorem{corollary}{Corollary}

\theoremstyle{definition}
\newtheorem{definition}{Definition}
\newtheorem*{definition*}{Definition}
\newtheorem*{assumption*}{Assumption}
\newtheorem{assumption}{Assumption}
\newtheorem{conjecture*}{Conjecture}
\newtheorem{example}{Example}
\newtheorem{remark}{Remark}

\usepackage{caption}

\DeclareMathOperator*{\argmin}{arg\,min}

\usepackage{url}
\usepackage[top=1in, bottom=1.5in, left=1.2in, right=1.2in]{geometry}
\usepackage{sectsty}
\usepackage{xurl}
\usepackage{dsfont}
\usepackage[bottom,multiple]{footmisc}
\usepackage{xspace}
\usepackage{siunitx}
\usetikzlibrary{decorations.pathreplacing, calligraphy}

\usepackage{enumitem}

\newcommand{\contents}{items\xspace}
\newcommand{\content}{item\xspace}
\newcommand{\ads}{ads\xspace}
\newcommand{\ad}{ad\xspace}
\newcommand{\ipair}{item-ad pair\xspace}
\newcommand{\ipairs}{item-ad pairs\xspace}

\newcommand{\type}{\ensuremath{\theta}}
\newcommand{\Type}{\ensuremath{[0,1]}}
\newcommand{\typemeas}{\ensuremath{F}}
\newcommand{\density}{\ensuremath{f}}
\newcommand{\types}{\ensuremath{\Theta}}
\newcommand{\typemax}{\ensuremath{1}}
\newcommand{\typemin}{\ensuremath{0}}
\newcommand{\quality}{\ensuremath{q}}

\newcommand{\rev}{\ensuremath{r}}

\newcommand{\matching}{\ensuremath{\pi}}

\newcommand{\bundle}{\ensuremath{B}}
\newcommand{\matchingset}{\ensuremath{\mathcal{M}}}
\newcommand{\grossq}{\ensuremath{Q}}
\newcommand{\grossr}{\ensuremath{R}}
\newcommand{\grossd}{\ensuremath{D}}
\newcommand{\size}{\ensuremath{N}}
\newcommand{\transfer}{\ensuremath{T}}
\newcommand{\attention}{\ensuremath{a}}
\newcommand{\dint}{\ensuremath{\,\mathrm{d}}}

\newcommand{\virtual}{\ensuremath{v}}
\newcommand{\zerotype}{\ensuremath{\theta^0}}
\newcommand{\vs}{\ensuremath{\hat{V}}} %
\newcommand{\R}{\ensuremath{\mathbb{R}}}

\newcommand{\maxquality}{\ensuremath{\overline{q}}}
\newcommand{\minquality}{\ensuremath{\underline{q}}}

\newcommand{\profit}{\ensuremath{\Pi}}

\newcommand{\disutil}{\ensuremath{d}}
\newcommand{\noad}{\ensuremath{\emptyset}}
\newcommand{\qset}{\ensuremath{\mathcal{I}}}
\newcommand{\adset}{\ensuremath{\mathcal{J}}}
\newcommand{\qvec}{\ensuremath{\mathbf{q}}}

\newcommand{\rhovec}{\ensuremath{\boldsymbol{\rho}}}
\newcommand{\itemnumber}{\ensuremath{I}}
\newcommand{\adnumber}{\ensuremath{J}}
\newcommand{\cost}{\ensuremath{c}}

\newcommand{\qvecset}{\ensuremath{\mathcal{Q}}}

\newcommand{\qvecH}{\ensuremath{\overline{\mathbf{q}}}}
\newcommand{\qvecL}{\ensuremath{\underline{\mathbf{q}}}}
\newcommand{\rhovecH}{\ensuremath{\overline{\boldsymbol{\rho}}}}
\newcommand{\rhovecL}{\ensuremath{\underline{\boldsymbol{\rho}}}}

\usepgfplotslibrary{fillbetween}
\usetikzlibrary{patterns}

\usepackage{color}
\makeatletter
\let\comment\@undefined
\let\endcomment\@undefined
\makeatother
\usepackage[deletedmarkup=sout,authormarkup=none]{changes}
\newcommand{\stkout}[1]{\ifmmode\text{\sout{\ensuremath{#1}}}\else \sout{#1}\fi}
\setdeletedmarkup{\stkout{#1}}

\definechangesauthor[color=red]{DS}

\definechangesauthor[color=blue]{SI}

\definechangesauthor[color=red]{BC}

\usepackage{booktabs}
\usepackage{array}
\usepackage{ragged2e}
\usepackage{xltabular}
\newcolumntype{Y}{>{\RaggedRight\arraybackslash}X}

\usepackage{bibunits}
\defaultbibliographystyle{plainnat}
\defaultbibliography{freemium}

\begin{document}

\begin{bibunit}

\newgeometry{top=1in, bottom=0.8in, left=1.2in, right=1.2in}
	\title{Mechanism Design for Ad-Supported Platforms\thanks{
\protect\setstretch{1}\protect\selectfont\footnotesize
    Ichihashi: Queen’s University, Department of Economics, \texttt{shotaichihashi@gmail.com}; Jeon: University of Toulouse Capitole, Toulouse School of Economics, \texttt{dohshin.jeon@tse-fr.eu}; Kim: University of Alabama, Culverhouse College of Business, \texttt{bkim34@ua.edu}.
    We are grateful to Gaurab Aryal, Jacques Crémer, Sungha Hwang, Bruno Jullien, Michihiro Kandori, Heiko Karle, Jeong-Yoo Kim, Jinwoo Kim, Adrien Raizonville, Markus Reisinger, Patrick Rey, Masayuki Sawada, Andrew Rhodes, Robert Somogyi, Tat-How Teh, Yiqing Xing, Yuichi Yamamoto, and Jidong Zhou, and audiences at the ACE 2024 (Rome), IIOC 2025, Postal Economics Conference on E-commerce, Digital Economics and Delivery Services 2024 at TSE, the Workshop on Search and Platforms (Kyoto), and the KER International Conference (2024), as well as seminar participants at Hitotsubashi University, KAIST, Keio University, Kyung Hee University, Nanyang Technological University, Peking University, Renmin University, Seoul National University, the University of Tokyo, and Yonsei University. Jeon  acknowledges the funding from ANR under grant
ANR-17-EUR-0010 (Investissements d'Avenir program).
This paper supersedes our previous working paper, \citet*{ichihashi2024mechanism}.
An earlier version was presented at the Paris Conference on Digital Economics 2022 and the MaCCI Annual Conference 2022 at Mannheim.
}
}
	\author{Shota Ichihashi \and Doh-Shin Jeon \and Byung-Cheol Kim}
	
	\date{September 11, 2026}
	\maketitle
	\thispagestyle{empty}

	\begin{abstract}

		Many digital platforms earn revenue by selling content access and displaying ads. We study monopoly screening by a platform allocating heterogeneous content and ads to heterogeneous consumers. Advertising generates revenue and affects information rents through nuisance, which alters the platform's incentives to allocate content. The optimal mechanism balances a trade-off between advertising and rent extraction: Expanding content access increases advertising revenue but reduces sales revenue by raising information rents. The mechanism features novel distortions in content access and ad exposure, rationalizes contracts across ad-supported platforms, and explains why greater ad profitability may weaken incentives to invest in content quality.



\end{abstract}
\clearpage
\restoregeometry

\setcounter{page}{1}
\section{Introduction }
Digital platforms adopt diverse business models, and evaluating consumer harm from their market power requires accounting for these differences \citep*{scottmorton}.
Existing reports on digital platforms\footnote{Examples include \citet*{cma2020online}, \citet*{cremer2019competition}, and \citet*{stigler}.} focus especially on ad-supported platforms and agree that
their market power can harm consumers through lower quality and reduced innovation.
However, these reports do not distinguish purely ad-funded platforms from ad-supported hybrid ones, which earn revenue from both selling access to content and advertising. 
This paper addresses the question of how hybrid platforms optimally balance their multiple revenue sources and whether we need to be concerned about their incentives to invest in quality.

There are many ad-supported hybrid platforms---such as Netflix, YouTube, Facebook, and The New York Times. They typically offer consumers a menu of plans, which vary in how heavily they rely on each source of revenue.  For example, the New York Times offers free and paid plans, both of which are ad-supported and grant differential access to articles. In contrast, X (formerly Twitter) grants free access to all content and charges users to remove ads. Streaming services---such as Netflix and Prime Video---offer plans that differ in both content access and advertising intensity.

In this paper, we study a mechanism design problem for an ad-supported hybrid platform.
The optimal mechanism we derive rationalizes various monetization strategies based on a unified economic force: the \emph{trade-off between advertising and rent extraction}. 
To monetize attention through advertising, the platform must expand content access.
However, expanding access increases information rents of consumers and limits the platform's ability to extract their surplus through content pricing \citep*{mussa1978, myerson1981optimal}.
This is a classical rent extraction-efficiency trade-off (cf. \citeauthor*{laffont2009theory}, \citeyear{laffont2009theory}), adapted to digital platforms that control content access and advertising.
The trade-off also yields a general intuition for why better advertising technologies or lax regulations surrounding advertising may curtail platforms' incentives to invest in content quality.

Our model is a version of a monopoly screening problem.
The platform hosts items and ads.
Items represent content, such as news articles, videos, or social media posts.
Items are vertically differentiated.
Ads are differentiated along two dimensions: (i) the revenue they generate for the platform and (ii) their disutility level.
Consumers have a one-dimensional private type that represents their taste for item quality and distaste for ads.
Their utility depends on monetary transfers, allocated items, and ads.

The platform designs a menu of contracts (or a mechanism) that specifies, for each consumer type, a set of items to be allocated, an advertising policy, and a monetary transfer to the platform.
An advertising policy determines whether to match each item with an ad, and if so, which one.  
An ad reduces the net quality of the matched item according to its disutility level.  
The platform maximizes total revenue---i.e., the sum of (i) advertising revenue generated by ads allocated to consumers, and (ii) monetary transfers from consumers, which we call sales revenue.

{%
A key difference between our model and standard screening models is the presence of ads.\space
Advertising serves both as a revenue generator and as a rent-extraction device:\space
Including ads in the contract designed for a given type generates ad revenue, lowers the contract’s price, and allows the platform to \textit{raise} the prices of contracts designed for higher types.\space
The relative effectiveness of these two roles varies across ads and is captured by their two-dimensional characteristics.\space
Moreover, because ads must be displayed alongside content, advertising affects the platform's incentive to provide consumers with access to content.\space
As a result, the model yields predictions about content access, advertising policy, and content-quality choices that models of screening with goods alone would not capture.
}

Our first main result characterizes the optimal mechanism. To balance the rent extraction-advertising trade-off, the platform tailors content access and advertising policy to consumer types.
For consumers with negative virtual types---who would be excluded in standard screening models---the platform allocates only items below a type-dependent quality threshold and matches every item with an ad. By doing so, the platform earns ad revenue while reducing information rents, which decrease because of ad nuisance and lower-quality allocations. For consumers with positive virtual types, the platform grants access to all items because the baseline model assumes zero marginal cost of providing items, which are typically digital goods. At the same time, the platform screens them through type-dependent advertising.

Second, the optimal mechanism bundles items with ads based on type-dependent \emph{virtual advertising profits.}
In general, ads affect the platform's total revenue in three ways: They generate advertising revenue, reduce consumers' willingness to pay for the associated items, and reduce information rents of higher types.
The magnitudes of these effects depend on the consumer's type and the characteristics of ads.
The virtual advertising profit captures their net effect as a single value for each ad, and the platform uses it to determine how to bundle items and ads for each consumer type.

{
In practice, ad-supported platforms engage in different forms of discrimination in terms of content access and ad exposure.
Some platforms differentiate mainly through content access.
For example, the New York Times and the Financial Times restrict article access for non-subscribers but display ads to all readers.
Other platforms differentiate mainly through ad exposure: YouTube and X allow broad access to content but charge users for reducing or removing ads. Still other platforms, such as Netflix, Spotify, and Peacock, discriminate along both dimensions.
}

{
In \autoref{sectionApplications}, we rationalize these menu forms as the optimal mechanism.
The exact menu form depends on advertising environments.
When ad nuisance is low relative to its revenue, advertising is not an effective screening margin, and the platform differentiates through content access. 
When ads exhibit both high nuisance and high revenue, the platform screens consumers only through advertising exposure.
If neither of these conditions holds, the platform differentiates jointly through content access and ad exposure. }

The optimal mechanism also shapes the platform's incentives to invest in content quality.
In the baseline model, the platform takes the quality of each item as given.
In \autoref{sectionInnovation}, we allow the platform to choose item quality at a cost.
 We first show that if virtual advertising profits uniformly increase across all ads and consumer types, the platform invests less in content quality.
We then show that such a change occurs if the platform hosts a larger set of ads or faces a higher ad revenue associated with each ad.
The intuition is linked to the rent extraction-advertising trade-off.
If advertising becomes more profitable, the platform allocates more items to types with negative virtual valuation, and the resulting increase in information rents is larger when the allocated items have higher quality.
Thus, greater advertising profitability makes it less attractive for the platform to invest in content quality.

The results on the platform's investment incentives have policy implications.
First, the results allow us to view policy and technological changes that surround advertising in terms of their impact on platforms' content quality.
In \autoref{sectionPolicy}, we illustrate this point by examining a recent bill proposed in California that aims to limit the loudness of ads on streaming platforms.
Second, our results provide a rationale for the existing concerns that ad-supported business models may curtail platforms' incentives to invest in content quality \citep*{cma2020online, stigler}.
We show that the concerns are relevant for both purely ad-funded platforms and hybrid platforms.
Our result is particularly relevant for the media industry because most news organizations use hybrid business models by charging for content and displaying ads.
In fact, there have been considerable concerns about the impact of advertising on the quality of journalism \citep*{OECD, latham}.

Our contribution is twofold.
First, we provide a general framework that rationalizes a spectrum of contracts offered by ad-supported digital platforms, from social media and streaming services to news organizations.
Second, we uncover a determinant of the platforms' investment incentives and link it to policy or technological changes surrounding advertising. The paper contributes to policy and academic discussions that recognize business models and monetization strategies as fundamental for understanding digital platforms and formulating relevant policies \citep*{caffarra2019follow, scottmorton}.

\section{Related Literature}

\paragraph{Monopoly Screening.}
Our work is related to monopoly screening and, more broadly, to mechanism design for selling goods \citep*{mussa1978, myerson1981optimal}.
The departure from standard screening models is the presence of advertising, which plays a dual role in rent extraction and revenue generation.
Consequently, the optimal mechanism may serve consumers with negative virtual values.
This property also arises in two-sided screening in which serving negative virtual values generates positive externalities for other agents (e.g., \citet*{damiano2007price},  \citet*{johnson2013matching}, \citet*{choi2015}, \citet*{gomes2016many, gomes2019price}, \citet*{jeon2022second},  \citet*{corrao2023nonlinear}, \citet*{bar2025selling}).

Compared with these studies, our model is simpler in that we do not model strategic advertisers.\footnote{The way we capture an advertising market is closer to the approach taken by \citet*{corrao2023nonlinear}, who, as one of their various applications, capture advertising in a reduced-form way with an external revenue function for the seller.}
However, we allow richer heterogeneity in content and ads and give the platform multiple screening instruments.
This richness enables us to understand digital platforms from an angle that is relatively underexplored.
First, we rationalize various ad-supported contracts observed in practice, including those that discriminate among consumers in both content access and ad exposure, which are common for streaming services.
Second, we study how advertising affects the platform's investment in content quality, without restricting the kind of contracts the platform can offer.

Our work also differs from the literature on damaged goods \citep*{deneckere1996damaged}.
In our model, the platform can ``damage" its offering in several ways---e.g., granting access to fewer or lower-quality content items and displaying annoying ads---and we characterize how the platform combines these instruments.

\paragraph{Advertising-Funded Platform.}
In the platform-economics literature, several papers study how ad-funded platforms discriminate among users by advertising exposure  \citep*{sato2019freemium, lin2020two, zennyo2020freemium, cai2023freemium}.
This body of work typically abstracts from the platform's content-allocation problem and costly investment in quality.
Our model accommodates these features, which allow broader applications.

Our result on content quality (\autoref{propositionInvest0}) is related to the platform-design literature, which studies platforms' incentives to design policies such as service quality, the quality of hosted sellers, the degree of competition within the platform and allocation of attention to content (e.g., \citet*{casner2020seller}, \citet*{liu2022implications}, \citet*{teh2022platform}, \citet*{johnen2024deceptive}, \citet*{madio2024content}, \citet*{Chen2026attention}).
In particular, \citet*{etro2021device} shows that a purely ad-funded platform underinvests in quality compared with a device-funded platform. \citet*{choi2023platform} compare a purely ad-funded business model with a hybrid one in terms of design bias and find that the former is biased toward the advertising side whereas the latter is biased toward the consumer side. Although their first result is aligned with that of \citet*{etro2021device} and ours, their second result is opposite to ours.\footnote{If we think that the opportunity cost of investment in content quality is investment in advertising technology, then their result implies that a hybrid platform invests too much in quality while a purely ad-funded  one invests too little in quality.}  The main difference between their paper and ours is that they consider homogeneous consumers and hence do not consider menus. Therefore, they cannot capture our main trade-off between advertising and rent extraction  and this is why our prediction about the effect of the hybrid business model on innovation incentive is opposite to theirs.

Finally, the literature on media markets studies how a platform's advertising decision can deviate from the social optimum and affect markets for content and products \citep*{gabszewicz2004programming, anderson2005market, peitz2008content, bergemann2011targeting, prat2022attention}.
We complement this literature by providing an intuition---that consumers' information rents can incentivize the platform to provide a socially excessive level of ``bads," i.e., advertising.

\section{Model}\label{sectionModel}
A monopoly platform allocates items and ads to consumers. Items vary in quality, while ads differ in both the revenue they generate for the platform and the disutility they impose on consumers. Each consumer has a one-dimensional  private type that determines their preferences for item quality and aversion to ads. The platform offers a menu of contracts, each of which specifies the set of items, matching between items and ads, and monetary transfer from consumers to the platform.
The platform maximizes the total revenue from consumer payments and advertising.

\paragraph{Items and Ads.}
The platform hosts a fixed set of items, denoted by $\qset := \{1,...,\itemnumber\}$, and a fixed set of ads, denoted by $\adset := \{1,..., \adnumber\}$, where $\itemnumber, \adnumber \in \mathbb{N}$.\footnote{{Each index $i$ or $j$ may instead denote a category of identical items or ads, with the corresponding primitives defined at the category level; the analysis is unchanged.
}}
The quality of each item $i \in \qset$ is denoted by $\quality(i) \ge 0$.
Without loss, assume that item $i$ has the $i$-th lowest quality, i.e., $\quality(1) \le  \quality(2)\le  \cdots \le \quality(\itemnumber)$.
The platform's marginal cost of providing items is zero as we consider digital content (see \autoref{sectionDiscussion}).
Each ad $j \in \adset$ is characterized by its revenue $\rev(j) \ge  0$ and disutility level $\disutil(j) \ge 0$.\footnote{All results hold regardless of (i) the (finite) numbers of items and ads, (ii) whether multiple items share the same quality, and (iii) whether multiple ads share the same revenue or disutility level.
Cases (ii) and (iii) imply only that there may be multiple optimal mechanisms; the analysis is unchanged.
\autoref{fig:combined} and \autoref{exampleItemCostInvestment} (Supplemental Appendix) study the case with $\itemnumber=\adnumber=1$.}\textsuperscript{,\kern0.12em}\footnote{{In reality, the platform or consumers may not be fully informed about item quality or ad characteristics.
	In such a case, we require all players to be symmetrically informed about these primitives, and we view $q(i)$, $r(j)$, and $d(j)$ as expectations conditional on public information.
Private information leads to an informed-principal problem or multidimensional types, which are beyond our scope.}}
{Without loss, assume that no ad $j$ has $\rev(j) = \disutil(j) = 0$.}\footnote{{Displaying an ad with $\rev(j) = \disutil(j) = 0$ is equivalent to displaying no ad, so we can omit such ads from $\adset$ without affecting our results.}}

We assume that $\disutil(j) \ge 0$ for every $j \in \adset$, i.e., ads reduce consumer utility.
The assumption---that ads are utility-decreasing bads---is also adopted by various papers and consistent with the idea that advertising is an implicit price consumers pay for using a digital service (e.g., \citealt*{stigler}, p.~63).\footnote{For papers that adopt a similar assumption, see, e.g., \citet*{anderson2005market, anderson2011platform, johnson2013targeted, gomes2016many, sato2019freemium}.
Our assumption does not exclude cases in which some ads bring (unmodeled) benefits to consumers, such as informing them about new products.
The assumption of nonnegative disutility only implies that the nuisance from ads exceeds other possible benefits of ads, the net effect of which is reflected in each $\disutil(j) \ge 0$.}
However, many of our results---including those on the characterization of the optimal mechanism (\autoref{lemma0}) and the platform's quality choice (\autoref{propositionInvest0})---hold even if some ads have negative disutility levels. That said, the key intuitions become clearer and our model's predictions have stronger applied relevance when consumers dislike ads.
For this reason, we assume upfront that all ads entail nonnegative disutility.

We also assume that the net quality of any pair of an item and an ad is nonnegative:
\begin{align}\label{equationNON}
\quality(i) - \disutil(j) \ge 0 \quad \text{for every } (i, j) \in \qset \times \adset.
\end{align}
The assumption mirrors the one in standard one-dimensional screening models, in which quality or quantity is nonnegative (see \autoref{sectionDiscussion} for detailed discussion).

\paragraph{Contracts.}
The platform offers consumers a menu of contracts.
A contract is a tuple $( \transfer, \bundle,\matching)$, where $\transfer \in \R$ denotes the monetary transfer from a consumer to the platform, and $\bundle \subseteq \qset$ denotes the set of allocated items, which we call the \emph{item bundle}.
The last component of the contract, $\matching$, denotes an \emph{advertising policy}, which specifies a one-to-one matching between items and ads with the possibility of matching some items with no ads.
Formally, an advertising policy $\matching$ is a map
\begin{equation}\label{equationAdpolicy}
\matching: \qset \to \adset\cup \{\noad\}
\end{equation}
such that there are no distinct items $i, i' \in \qset$ with $\matching(i) = \matching(i') \in \adset$.
Given an advertising policy $\matching$, each item $i \in \qset$ is bundled with ad $\matching(i)$, where $\matching(i) =\noad$ means that item $i$ is not bundled with any ad.
We call any element of $\qset \times ( \adset\cup \{\emptyset\})$, including $(i, \noad)$, an item-ad pair.
Let $\matchingset$ denote the set of all advertising policies.

One-to-one matching implies that (i) each item can be tied with at most one ad, and (ii) each ad can be shown to each consumer at most once.
In \autoref{remark1}, we discuss how we may relax the first restriction (i).
The second restriction (ii) is without loss of generality, because ads in $\adset$ can represent copies of the same ad.
For example, if the platform hosts a single ad but can show it $\adnumber$ times to a given consumer, we can interpret each ad $j \in \adset$ as the $j$-th opportunity to display the ad.\footnote{By taking $\rev(j)$ to be decreasing in $j$, we can capture a situation in which the effectiveness of advertising declines over time as the platform shows a given ad repeatedly to the same consumer.}

According to \eqref{equationAdpolicy}, the domain of an advertising policy $\matching$ is equal to the set of all items, $\qset$, regardless of the underlying set $\bundle$ of items specified in the contract.
We might think this is redundant because, as will become clear later, how $\matching$ varies outside of $\bundle$ does not affect a consumer's utility from a contract.
However, we take the domain of $\matching$ to be $\qset$ instead of $\bundle$ to streamline the exposition and simplify proofs.

\paragraph{Consumers.}
A unit mass of consumers interacts with the platform.
Each consumer is privately informed of her type $\type$, which is distributed according to $\typemeas \in \Delta \Type$ that has a positive density $\density$ on $\Type$.\footnote{{The same analysis holds for any discrete type space.
The platform's problem reduces to the maximization of a virtual surplus that has the same form as in our model, so the characterization of the optimal mechanism carries over. See Supplemental Appendix~\ref{sectionAppendixDiscreteType}.}}
The platform knows the type distribution.
Assume that the virtual value, $\virtual(\type):= \type - \frac{1 - \typemeas(\type)}{\density(\type)}$, is strictly increasing and continuous in $\type$.\footnote{Throughout, ``increasing,'' ``decreasing,'' ``greater than,'' ``less than,'' and related terms are understood in the weak sense unless explicitly qualified by ``strictly.''}
Define $\zerotype \in (\typemin, \typemax)$ as the unique type that has zero virtual value, $\virtual(\zerotype) =0$.\footnote{The assumption---that $\virtual(\type)$ is strictly increasing and continuous---ensures that $\zerotype$ uniquely exists.
If $\virtual(\cdot)$ is nonmonotone, we use the ironed virtual valuation \citep*{myerson1981optimal}.}
We refer to consumers with $\type<\zerotype$ as \emph{negative virtual types} and to those with $\type>\zerotype$ as \emph{positive virtual types}.

A consumer's utility from a contract $(\transfer, \bundle,\matching)$ is defined as follows:
\begin{equation}\label{equationUtil0}
\type \sum_{i \in \bundle} \left[ \quality (i) -  \disutil(\matching(i)) \right]  - \transfer
\end{equation}
with the notational convention that $\disutil (\noad) =0$.\footnote{Recall that $\disutil(\matching(i))$ is the disutility level of ad $\matching(i)$, which is matched with item $i \in \bundle$.
If item $i$ is not matched with any ad (i.e., if $\matching (i) =\noad$), we have $\disutil(\matching(i))=0$.}
Thus, consumers who place greater value on content quality (i.e., those with higher types) also experience greater disutility from ads.\footnote{\citet*{Bisceglia2025privacy} consider the same multiplicative payoff structure $\type (q-L)$ where $q$ is quality and $L$ is loss from data collection.}
In \autoref{sectionDiscussion}, we motivate our specification and discuss alternatives.
Also, we assume that consumers cannot ignore ads while consuming items.\footnote{See, e.g., \citet*{anderson2011platform}, \citet*{johnson2013targeted}, and \citet*{peitz2023adblocking} for implications of costly ad-avoidance technologies on market outcomes.}

\paragraph{Mechanism Design.} 
The platform chooses a direct mechanism, represented by 
$$\{( \transfer(\type), \bundle(\type),\matching(\cdot|\type))\}_{\type \in \Type},$$ 
which for each type $\type$ specifies a contract, i.e., monetary transfer $\transfer(\type)$ to the platform, item bundle $\bundle(\type)$, and advertising policy $\matching (\cdot|\type)$.
Aside from incentive compatibility and individual rationality constraints, we impose no restrictions on feasible mechanisms.\footnote{Nonetheless, a more realistic model might impose some restrictions that prevent platforms from differentiating content access, such as contracts between the platform and content producers.}

{A mechanism involves two types of matching (\autoref{fig:matching}): a one-to-one matching between items and ads, specified by the advertising policy $\matching(\cdot|\type)$ for each type $\type$, and a one-to-many matching between consumer types and item-ad pairs, because each type receives the item-ad pairs formed by their bundle and advertising policy.
Generally, consumers with different types may receive different sets of items, see different ads even for the same item, and pay different prices.}

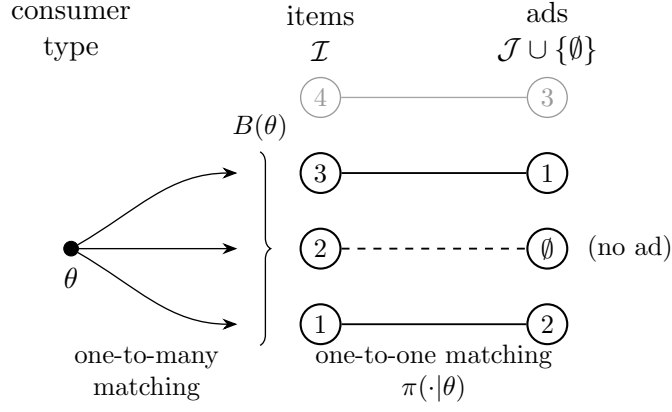
\begin{figure}[H]
\centering
\begin{tikzpicture}[
    x=1cm, y=1cm,
    >=Stealth,
    every node/.style={font=\small},
    allocNode/.style={circle, draw=black, fill=white, line width=0.8pt, inner sep=0pt, minimum size=15pt, font=\footnotesize},
    grayNode/.style={circle, draw=black!40, fill=white, line width=0.5pt, inner sep=0pt, minimum size=15pt, font=\footnotesize, text=black!40},
    matchEdge/.style={draw=black, line width=0.7pt},
    noAdEdge/.style={draw=black, line width=0.7pt, dashed},
    grayEdge/.style={draw=black!35, line width=0.5pt},
    assignArrow/.style={draw=black, line width=0.5pt, -{Stealth[length=2mm]}},
    colHead/.style={font=\small, align=center},
    tagLabel/.style={font=\footnotesize, align=center}
]
\node[circle, draw=black, fill=black, inner sep=0pt, minimum size=5.5pt, label={below:$\type$}] (theta) at (0,-0.5) {};

\node[grayNode]  (i4) at (3.3,1.5)  {$4$};
\node[allocNode] (i3) at (3.3,0.5)  {$3$};
\node[allocNode] (i2) at (3.3,-0.5) {$2$};
\node[allocNode] (i1) at (3.3,-1.5) {$1$};

\node[grayNode]  (a3) at (6.3,1.5)  {$3$};
\node[allocNode] (a1) at (6.3,0.5)  {$1$};
\node[allocNode] (a0) at (6.3,-0.5) {$\noad$};
\node[allocNode] (a2) at (6.3,-1.5) {$2$};
\node[tagLabel, right=3pt of a0] {(no ad)};

\node[colHead] at (0,2.35)   {consumer\\type};
\node[colHead] at (3.3,2.35) {items\\$\qset$};
\node[colHead] at (6.3,2.35) {ads\\$\adset\cup\{\noad\}$};

\draw[grayEdge]  (i4) -- (a3);
\draw[matchEdge] (i3) -- (a1);
\draw[noAdEdge]  (i2) -- (a0);
\draw[matchEdge] (i1) -- (a2);
\node[tagLabel] at (4.8,-2.15) {one-to-one matching\\$\matching(\cdot|\type)$};

\draw[decorate, decoration={brace, amplitude=4pt}, line width=0.5pt]
    (2.5,0.78) -- (2.5,-1.78);
\node[tagLabel] at (2.5,1.1) {$\bundle(\type)$};

\draw[assignArrow] (theta) to[out=30,in=180]  (2.2,0.5);
\draw[assignArrow] (theta) -- (2.2,-0.5);
\draw[assignArrow] (theta) to[out=-30,in=180] (2.2,-1.5);
\node[tagLabel] at (1.0,-2.15) {one-to-many\\matching};

\end{tikzpicture}
\caption{{A one-to-one matching between items and ads (horizontal lines), and a one-to-many matching between consumer types and item-ad pairs (arrows).
Type $\type$ receives bundle $\bundle(\type) = \{1,2,3\}$; gray indicates the item and the ad that $\type$ does not receive.}}
\label{fig:matching}
\end{figure}

Given any direct mechanism, if a consumer with type $\type$  reports to be $\type'$, she obtains a payoff of 
\begin{equation}\label{equationUtil}
\type \sum_{i \in \bundle(\type')} \left[ \quality (i) -  \disutil(\matching(i|\type')) \right]  - \transfer(\type').
\end{equation}
If all consumers report their type truthfully, the platform's revenue is
\begin{align*}
\int^1_0  \transfer(\type) + \sum_{i \in \bundle(\type)} \rev(\matching(i|\type))
 \dint \typemeas(\type).
\end{align*}
The first term $\transfer(\type)$ of the integrand is the monetary transfer from type $\type$. This is the revenue from selling access to content and we call it the sales revenue.
The second term $\sum_{i \in \bundle(\type)} \rev(\matching(i|\type))$ is the total advertising revenue generated by type $\type$.

The platform chooses a mechanism to maximize total revenue subject to incentive compatibility (IC) and individual rationality (IR) constraints. The resulting problem is as follows.
\begin{align}
&\max_{\{( \transfer(\type), \bundle(\type), \matching(\cdot|\type))\}_{\type \in \Type}} \int^1_0  \transfer(\type) + \sum_{i \in \bundle(\type)} \rev(\matching(i|\type))
 \dint \typemeas(\type) \tag{M}\label{M} \\[12pt]
\text{subject to} \quad  \quad&\type \sum_{i \in \bundle(\type)} \left[ \quality (i) -  \disutil(\matching(i|\type)) \right] - \transfer(\type) 
\ge \type \sum_{i \in \bundle(\type')} \left[ \quality (i) -  \disutil(\matching(i|\type')) \right]  - \transfer(\type'), \forall \type, \type' \in \Type \quad \tag{IC}\label{IC} \\[12pt]
 \quad& \type \sum_{i \in \bundle(\type)} \left[ \quality (i) -  \disutil(\matching(i|\type)) \right]  - \transfer(\type) \ge 0 , \forall \type \in \Type. \quad \tag{IR}\label{IR}
\end{align}

\subsection{Discussion of Assumptions}\label{sectionDiscussion}
\paragraph{Consumer Utility Specification.}
Our utility specification \eqref{equationUtil0} follows the standard one-dimensional screening formulation of \citet*{mussa1978}, in which a buyer's gross utility is the product of her type and the quality of the good. In our model, the net quality of an item-ad pair is the item's quality minus the disutility from the associated ad. Thus, consumers with higher $\type$ have both a greater willingness to pay for item quality and a stronger aversion to advertising.

This type-dependent utility specification is important for our results. The model predicts that the platform may screen consumers through content access and advertising exposure, in addition to price, as observed in many ad-supported platforms. Some nearby specifications would shut down one of these screening margins.

First, suppose that a consumer's utility is
\[
\type \sum_{i\in \bundle}\quality(i)
-
\sum_{i\in \bundle}\disutil(\matching(i))
-
\transfer,
\]
so that advertising disutility is independent of $\type$. 
In this case, advertising exposure does not affect information rents. 
Then, which ads to display depends only on each $\rev(j)-\disutil(j)$ and becomes independent of consumer types.
This specification is subsumed by the type-independent ad disutility extension (see \autoref{sectionIndep}) and does not rationalize paid ad-free contracts, such as those observed on YouTube and Netflix.

Second, suppose that item utility is type-independent:
\[
\sum_{i\in \bundle}\quality(i)
-
\type \sum_{i\in \bundle}\disutil(\matching(i))
-
\transfer.
\]
Then, all consumers have the same value for item quality, so the platform does not use content access as a screening device.
Indeed, with zero marginal cost of providing items, the platform allocates all items to all consumers.
Such a model fails to explain partial-access plans or premium-content tiers, such as those offered by Peacock, Netflix, and The New York Times.

Finally, one could allow multidimensional heterogeneity:
\[
\type_1 \sum_{i\in \bundle}\quality(i)
-
\type_2 \sum_{i\in \bundle}\disutil(\matching(i))
-
\transfer,
\]
where $(\type_1,\type_2)$ is drawn from a joint distribution.\footnote{This specification brings our model closer to that of \citet*{yang2021costly}. However, his Theorem 1 does not apply to our baseline model or to the specification discussed here, for at least two reasons: (i) \citet*{yang2021costly} assumes that the ``productive component'' of screening instruments is one-dimensional, whereas in our model the space of productive allocations, namely advertising policies, is multidimensional; and (ii) his Theorem 1 requires $\type_1$ and $\type_2$ to be negatively correlated, whereas in our baseline model they are perfectly positively correlated (as $\type_1=\type_2$).}
This specification may generate richer menus, but one-dimensional types allow us to keep the analysis simple and render the model amenable to richer analysis, such as endogenous quality investment.

\paragraph{Cost of Producing Content.}
We assume that the platform incurs no marginal cost of providing consumers with access to items.
This is natural if (i) items are freely replicable digital goods and (ii) there is no expense tied to content dissemination on the platform's side, such as payments to content producers.
The assumption of zero fixed cost is relaxed in \autoref{sectionInnovation}, \autoref{sectionItemCost}, and \autoref{sectionNoncont}.

\paragraph{Nonnegative Net Quality and Contractible Consumption.}
The condition \eqref{equationNON} of nonnegative net quality is crucial, but its exact role depends on whether consumption is contractible.
If consumption is contractible (i.e., the platform can enforce the consumption of any allocated bundle), condition \eqref{equationNON} allows us to apply the standard technique for monopoly screening.\footnote{Without condition \eqref{equationNON}, the binding IR constraint can depend on the allocation.
For example, if the platform allocates only \ipairs with $\quality(i) - \disutil(j) < 0$, the IR constraint binds for $\type=1$.}
If consumption is not contractible, consumers would freely dispose of any allocated \ipair with negative net quality.
As a result, every consumed pair satisfies condition \eqref{equationNON}, even if some allocated pairs violate it.\footnote{The platform can, without loss, allocate only \ipairs with $\quality(i) - \disutil(j) \ge 0$, because consumers would choose not to consume any pair with negative net quality.}
The condition still simplifies our analysis, because it ensures that the matching between items and ads does not affect which item bundles the platform can assign to each consumer (see Supplemental Appendix~\ref{sectionAppendixPAM} for an example in which relaxing the condition changes the optimal advertising policy; see also Supplemental Appendix~\ref{sectionAppendixFreeDisposal} for a related analysis).
In sum, under condition \eqref{equationNON}, consumers have no ex post incentive to discard any allocated item-ad pair, so whether consumption is contractible does not affect our baseline analysis, but the way in which our analysis is affected by the absence of the condition depends on whether consumption is contractible.

\section{Optimal Mechanism}\label{sectionSubscription2}
We first define a few concepts that help us describe the platform's optimal mechanism.
We then characterize the optimal mechanism and present its key properties.

\subsection{Virtual Advertising Profits}
We begin by defining \emph{virtual advertising profits} (hereafter, \emph{virtual ad profits}), which measure an ad's contribution to the platform's total revenue when we take into account the impact of advertising on the sales revenue in any incentive-compatible mechanism.

Formally, for any type $\type \in [0,1]$ and ad $j \in \adset$, the \emph{$\type$-virtual ad profit} of ad $j$, denoted by $\rho_\type(j)$, is defined as
\begin{align}\label{equationVirtualAd}
\rho_\type(j) := \rev(j) - \virtual(\type) \disutil(j).
\end{align}
Unpacking the virtual value as $\virtual(\type)= \type - \frac{1 - \typemeas(\type)}{\density(\type)}$, we can express $\rho_\type(j)$ as
\begin{align}\label{equationVirtualAd2}
\rho_\type(j) = \underbrace{\rev(j)}_{\text{ad revenue}} \underbrace{- \type \disutil(j)}_{\substack{\text{lower}\\\text{price for $\type$ }}} \underbrace{+\frac{1 - \typemeas(\type)}{\density(\type)} \disutil(j)}_{\substack{\text{reduction in}\\\text{information rents}}}.
\end{align}
The first term $\rev(j)$ is the direct contribution of ad $j$ to the platform revenue through the advertising revenue.
The rest of the expression \eqref{equationVirtualAd2} captures the impact of advertising on the platform's sales revenue.
Specifically, the second term, $- \type \disutil(j)$, captures the impact of allocating ad $j$ on the sales revenue $\transfer(\type)$ from type $\type$, i.e., the nuisance from ad $j$ reduces type $\type$'s willingness to pay for an item by $\type \disutil(j)$. The last term, $\frac{1 - \typemeas(\type)}{\density(\type)} \disutil(j)$, captures the impact of advertising on information rents for types higher than $\type$.
Displaying ad $j$ to type $\type$ reduces the incentives of higher-type consumers to misreport and take the contract intended for type $\type$.
As a result, the platform can increase prices for higher types while keeping the mechanism incentive compatible.

Virtual ad profits illustrate the idea that 
the true contribution of an advertisement to the platform's revenue typically differs from either (i) its contribution to social welfare (i.e., ad revenue minus nuisance, $\rev(j) - \type \disutil(j)$) or (ii) its ad revenue (i.e., $\rev(j)$).
Below we provide an intuition for each comparison.

{First, an ad's contribution $\rho_\type(j)$ to total revenue exceeds the contribution to social welfare (i.e., $\rho_\type(j) > \rev(j) - \type \disutil(j)$ whenever $\type< 1$ and $\disutil(j)>0$), because nuisance enables the platform to reduce information rents, which is captured by the last term in \eqref{equationVirtualAd2}.}
Thus, for any fixed set of items, the platform has a socially excessive incentive to display ads (except for the highest type).
For example, suppose that ad $j$ satisfies $\rev(j)=\hat\theta \disutil(j)$ for some $\hat\theta \in (0,1)$. 
Then, for any type $\type > \hat\type$, the ad revenue $\rev(j)$ falls below the disutility $\type \disutil(j)$; thus, in the first best, the platform never allocates ad $j$ to type $\type$.
However, in the second best, the platform may strictly prefer allocating ad $j$ to allocating no ad in order to reduce information rents.
In other words, the same trade-off between efficiency and rent extraction, which generates downward distortions in ``goods,'' may generate upward distortions in ``bads.''

Second, the virtual ad profit $\rho_\type(j)$ is larger (smaller) than the ad revenue $\rev(j)$ for consumers with negative (positive) virtual types. 
For example, for negative virtual types, advertising is more profitable than suggested by their ad revenue---that is, $\rho_\type(j) > \rev(j)$ when $\virtual(\type) < 0$ {and $\disutil(j) > 0$}---because the reduction in information rents from ad nuisance is larger than the nuisance itself.
In fact, the virtual ad profit increases (decreases) with the disutility level for negative (positive) virtual types. Therefore, if all ads generate the same revenue, the platform prefers showing the ads with the highest (lowest) disutility levels for negative (positive) virtual types.
Generally, the platform's ranking of ads by virtual ad profits depends on both ad characteristics, $(\rev(j), \disutil(j))$, and a consumer's type, $\type$.

\subsection{Assortative Advertising Policies}
{The optimal mechanism can be implemented by matching items and ads in a negatively assortative way, based on item quality and virtual ad profits.}
To formalize this notion, we first define a class of advertising policies.
Take any map $\rho : \adset \to \mathbb{R}$, which assigns each ad $j \in \adset$ a score, $\rho(j)$.
Call a permutation $\mu:\adset \to \adset$ a $\rho$-permutation if $\mu$ labels ads in such a way that an ad with a larger index has a lower score---i.e., 
\begin{equation*}
\rho(\mu(1)) \ge \rho(\mu(2)) \ge \cdots \ge \rho(\mu(\adnumber)).
\end{equation*}
If multiple ads have the same score under $\rho(\cdot)$, the $\rho$-permutation may not be unique.

\begin{definition}\label{definitionAssortative}

An advertising policy is a \emph{$\rho$-negative assortative policy} if the following holds for some $\rho$-permutation, $\mu$:
For each $k=1,\ldots,\itemnumber$, item $k$, which has the $k$-th lowest quality, is matched with ad $\mu(k)$ if $k\leq\adnumber$ and $\rho(\mu(k))\geq0$, and not matched with any ad otherwise.
\end{definition}

As shown below, an optimal mechanism uses a $\rho_\type$-negative assortative policy for every type $\type$, where $\rho_\type$ is the virtual ad profit function for type $\type$ (see \eqref{equationVirtualAd}).
According to \autoref{definitionAssortative}, we can construct {this} advertising policy for each type $\type$ as follows.
First, we assign each ad $j$ its $\type$-virtual ad profit, $\rho_\type(j)$.  
Then, starting from the lowest-quality item, we match items and ads in a negatively assortative manner, according to the quality $\quality(i)$ of each item and $\rho_\type(j)$ of each ad.
Once we exhaust ads that have nonnegative $\rho_\type(j)$, we stop and leave the remaining higher-quality items unmatched with any ads.

\subsection{Characterization of the Optimal Mechanism}
We now solve the platform's problem and describe the optimal mechanism.
First, we rewrite the platform's problem \eqref{M} as maximization of virtual surplus subject to the monotonicity constraint, i.e., $\sum_{i \in \bundle(\type)} \left[ \quality (i) -  \disutil(\matching(i|\type)) \right]$ is increasing in $\type$ (e.g., \citealt*{myerson1981optimal}).
We then define the relaxed problem as the one in which we maximize the virtual surplus without the monotonicity constraint:
\begin{align}
\max_{\{( \bundle(\type), \matching(\cdot|\type))\}_{\type \in \Type}}  \int^1_0    
\virtual(\type)  \sum_{i \in \bundle(\type)} \left[ \quality (i) -  \disutil(\matching(i|\type)) \right]   
+\sum_{i \in \bundle(\type)} \rev(\matching(i|\type))  \dint \typemeas(\type) \tag{P}\label{P}.
\end{align}

Solving \hyperref[P]{Problem \eqref{P}} and applying the standard technique to recover monetary transfers from local IC constraints, we obtain the optimal mechanism.
\begin{lemma}\label{lemma0}
The platform's problem \eqref{M} has the following solution: For each type $\type \in \Type$, the platform adopts a $\rho_\type $-negative assortative matching, $\matching^*(\cdot|\type)$, and allocates all items that generate non-negative virtual surplus, i.e., 
 $\bundle(\type):= \{ i  \in \qset: \virtual(\type) \quality(i) + \rho_\type ( \matching^*(i|\type) )  \ge 0 \}$.
 The monetary transfer from type $\type$ equals
 $$\transfer(\type) = \type \displaystyle\sum_{i \in \bundle(\type)} \left[ \quality (i) -  \disutil(\matching^*(i|\type)) \right]  - 
 \int^\type_0 \displaystyle\sum_{i \in \bundle(t)} \left[ \quality (i) -  \disutil(\matching^*(i|t)) \right]  \dint t.$$
\end{lemma}

We sketch the proof of \autoref{lemma0} by describing how to solve the relaxed problem, \hyperref[P]{Problem \eqref{P}} (see \autoref{sectionAppendixA} for omitted details).
First, the virtual surplus is separable across types, so we can reduce \hyperref[P]{Problem \eqref{P}} to pointwise maximization of each type's virtual surplus, $\virtual(\type)  \sum_{i \in \bundle(\type)} \left[ \quality (i) -  \disutil(\matching(i|\type)) \right]   
+\sum_{i \in \bundle(\type)} \rev(\matching(i|\type))$.

Second, we calculate the contribution of each item-ad pair to type $\type$'s virtual surplus.
Suppose that the platform matches item $i \in \qset$ with ad $j\in \adset \cup \{\noad\}$.
If the platform allocates an item-ad pair $(i,j)$, the virtual surplus from type $\type$ increases by
\begin{align}\label{equationVirtual2}
\virtual(\type) [ \quality(i) - \disutil(j)]  + \rev(j),
\end{align}
where $\disutil(\noad) = \rev(\noad) =0$.

The platform allocates an item-ad pair $(i, j)$ if and only if its contribution to the virtual surplus is nonnegative. 
As a result, the contribution of the pair $(i, j)$ to the virtual surplus, given the platform's optimal decision for whether to allocate the pair, is equal to $\max\left(0, \virtual(\type) [ \quality(i) - \disutil(j)] + \rev(j)\right)$, or equivalently, 
\begin{align}
\max\left(0, \virtual(\type) \quality(i) + \rho_\type (j) \right),\label{equationVirtual1}
\end{align}
with $\rho_\type (\noad) =0$.
Thus, the platform's problem of maximizing type $\type$'s virtual surplus is 
\begin{align}
\max_{\matching \in \matchingset} \sum_{i \in \qset} \max\left(0, \virtual(\type) \quality(i) + \rho_\type(\matching(i)) \right).\tag{R-$\theta$}\label{equationVirtual3}
\end{align}
The outside max operator captures the choice of an advertising policy, and the inside max operator reflects the problem of whether to allocate each item-ad pair.

The problem \eqref{equationVirtual3} admits the following solution.
First, consider the outside max operator for a negative virtual type, $\virtual(\type)<0$.
The platform never displays ad $j$ when $\rho_\type(j)<0$.
Also, the virtual surplus \eqref{equationVirtual1} from an item-ad pair $(i,j)$  is submodular in $(\quality(i), \rho_\type(j))$.
Consequently, {one solution is to match} items in $\qset$ and ads in $\{j\in \adset: \rho_\type(j)>0\}$ in a negatively assortative way, which reduces to a $\rho_\type$-negative assortative matching.
The case of positive virtual types  follows the same logic, except that 
the submodularity of the objective in \eqref{equationVirtual3} is trivial:
If $\virtual(\type) > 0$, we have $\virtual(\type) \quality(i) + \rho_\type (j) =\virtual(\type) [ \quality(i) - \disutil(j)]  + \rev(j)\geq0$.
Finally, the inside max operator of  \eqref{equationVirtual3} determines whether to allocate each matched item-ad pair to type $\type$, the solution to which reduces to the item bundle $\bundle(\type)$ described in \autoref{lemma0}.

\begin{remark}\label{remarkMatching}
{The optimal mechanism in \autoref{lemma0} has two properties: If a consumer of a given type $\type$ receives $K$ items, then (i) the items are $K$ lowest-quality ones and are matched with at most $K$ ads with the highest nonnegative virtual ad profits, and (ii) within the allocated sets, items and ads are matched according to a $\rho_\type$-negative assortative policy.
We note that only (i) is a robust feature of the optimal mechanism, whereas (ii) is not.
Indeed, because the revenue and disutility of each ad do not depend on the matched item, the virtual surplus of an item-ad pair takes the additively separable form in \eqref{equationVirtual1}; hence, given the sets of allocated items and ads, the exact matching between them is payoff-irrelevant.
As a result, we view the negative assortative policy as a concise way to represent an optimal mechanism.
Item costs can yield an optimal mechanism with a different matching pattern (see \autoref{sectionItemCost}), while a richer structure of advertising revenue, such as complementarity between ad revenue and item quality, can make the exact matching payoff-relevant (see \autoref{sectionItemAdRevenue}).}
\end{remark}

We build on \autoref{lemma0} and establish key properties of the optimal mechanism.
To streamline exposition, assume that no two items have the same quality level.
We then define a set of items whose quality levels fall below some threshold.
\begin{definition}
An item bundle $B \subseteq \qset$ is a \emph{lower-contour bundle} if it can be written as
$B = \{i \in \qset : \quality(i) \leq \bar\quality\}$
for some quality cutoff $\bar\quality$, where either $\bar\quality < \quality(1)$ or
$\bar\quality \in \{\quality(1), \ldots, \quality(\itemnumber)\}$.
\end{definition}

For a quality cutoff
$\bar\quality < \quality(1) = \min_{i \in \qset} \quality(i)$,
the corresponding lower-contour bundle is empty and represents exclusion from
access to content. To streamline the exposition, we normalize the quality cutoff
for the empty bundle to a fixed constant, say $\bar\quality = -1$.
\begin{theorem}\label{theorem}
The optimal mechanism has the following properties:
\begin{enumerate}
\item Consumers with negative virtual types (i.e., $\type <\zerotype$) receive lower-contour bundles with quality cutoff $\maxquality(\type)$ that is increasing in $\type$.
Every allocated item is matched with an ad.
Specifically, if consumers with type $\type$ receive $K^-(\type)$ items, they are matched with the $K^-(\type)$ ads with the highest nonnegative virtual ad profits, where $K^-(\type)$ increases in $\type$.

\item Consumers with positive virtual types (i.e., $\type >\zerotype$) receive all items.
Consumers with type $\type$ will see $K^+(\type)\le \itemnumber$ ads with the highest nonnegative virtual ad profits, where $K^+(\type)$ decreases in $\type$.
\end{enumerate}
\end{theorem}

{\autoref{theorem} highlights the departure from standard models.
In standard screening, the seller finds it optimal to exclude consumers with negative virtual values, because serving such consumers raises information rents by more than the increment of total surplus.
In our model, the platform can earn ad revenue by serving consumers with negative virtual values.
However, whenever the platform does so, it trades off ad revenue against information rents: 
Allocating items bundled with ads increases ad revenue but raises information rents and reduces the platform's sales revenue, even though nuisance from those ads reduces information rents.}

{Part 1 illustrates how this trade-off shapes the allocation of content and ads for negative virtual types.
	The platform may serve these types, but allocate only lower-contour bundles, with every allocated item matched with an ad.
This is because low-quality items create smaller rents, and ad nuisance further limits rents by reducing net quality. Under condition \eqref{equationNON}, absent ad revenue the platform has no incentive to serve negative virtual types.
The type-dependent threshold $\maxquality(\type)$ is increasing in $\type$.
Thus, higher types gain access to higher-quality items, in addition to lower-quality items, and correspondingly pay higher prices.} %
{Such a contract might arise as a plan that has a low price but excludes access to some premium content.}

{%
The extent to which the platform serves negative virtual types depends on ad characteristics.\space
Part 1 of \autoref{theorem} includes two extreme cases: If all ads generate zero revenue and disutility, the threshold $\maxquality(\type)$ is so low for every $\type < \zerotype$ that the platform excludes all negative virtual types, as in standard screening.\space
If all ads generate high revenues, the threshold is so high that consumers receive all items regardless of their types.\space
The next result connects these extremes: The platform expands content allocation to negative virtual types if ads are more profitable or annoying.
}

\begin{corollary}\label{corollaryThreshold}
{Consider two sets of ads, $\{(\rev(j), \disutil(j))\}_{j \in \adset}$ and $\{(\rev'(j), \disutil'(j))\}_{j \in \adset}$, such that $\rev'(j) \ge \rev(j)$ and $\disutil'(j) \ge \disutil(j)$ for every $j \in \adset$.
Then, for every $\type < \zerotype$, the quality threshold $\maxquality(\type)$ in the optimal mechanism is greater under $\{(\rev'(j), \disutil'(j))\}_{j \in \adset}$ than under $\{(\rev(j), \disutil(j))\}_{j \in \adset}$.}
\end{corollary}

{%
Higher ad revenues raise virtual ad profits directly.
Higher disutility levels also raise virtual ad profits for negative virtual types, because the nuisance of an ad reduces the information rents of higher types by more than it reduces type $\type$'s willingness to pay.
In both cases, each item-ad pair becomes more profitable to allocate, and thus the platform expands the set of items offered to negative virtual types.}

The platform's problem for positive virtual types also differs from standard screening (Part 2).
The platform allocates all items to consumers with positive virtual types, because the marginal cost of production is zero.\footnote{Another reason the platform does not screen consumers with positive virtual types by quality is that consumers incur no cost of consuming items. \autoref{sectionItemCost} relaxes both cost assumptions; together with \autoref{sectionNoncont}, it shows that attention costs can induce quality screening for positive virtual types whether or not consumption is contractible.}
This property also allows us to select an optimal mechanism that uses negative assortative matching for every positive virtual type.
 Yet advertising exposure varies with type: As a consumer's type increases, the reduction in their willingness to pay due to ad nuisance becomes more important than the reduction in the information rents of higher types.
As a result, consumers with higher types are exposed to fewer ads.
Thus, the screening of positive virtual types operates through advertising exposure. 

Our results also predict how the platform distorts advertising allocation relative to the first best.
{For consumers with positive virtual types, the virtual ad profit of each ad is greater than its contribution to total surplus (i.e., $\rho_\type(j) \ge \rev(j) - \type \disutil(j)$); thus, the platform shows more ads than the efficient level to each type, except for the highest type, which receives the first-best set of ads.}

For consumers with negative virtual types, Part 1 shows that they receive too few items. 
This implies that, contrary to the case of positive virtual types, the platform may serve fewer ads than the efficient level, because ads are bundled with items.
Moreover, the platform is biased toward high-nuisance ads. To see this, suppose that all ads generate the same revenue.
Then, given the number of ads to be displayed, the platform prefers displaying ads with the highest nuisance, even though total surplus is maximized by allocating ads with the lowest nuisance.
This is because for negative virtual types, virtual ad profits increase with disutility levels.
Such a distortion is absent for positive virtual types, for which virtual ad profits are decreasing in disutility levels.

The following example illustrates the typical structure and the distortion in the allocation of items and ads induced by the platform's optimal mechanism.
\begin{example}
Consumer types are distributed on $[0,1]$ according to $\typemeas(\type)=\type^3$, so the virtual type is given by $\virtual(\type)=\frac{4\type^3-1}{3\type^2}$ and $\zerotype\approx 0.63$.
The platform has the same number of items and ads, i.e., $I=J\geq 2$; each item $i=1,\ldots,\itemnumber$ has quality $\quality(i)=1+\frac{i-1}{\itemnumber-1}$; and each ad $j=1,\ldots,\itemnumber$ generates the same revenue $\rev>0$ and has disutility level $\disutil(j)=\frac{j-1}{\itemnumber-1}$.
The example is a discrete version of a situation in which item quality and ad disutility are uniformly distributed on $[1,2]$ and $[0,1]$, respectively.
Hereafter, we identify each item and ad by its quality level and disutility level, respectively; the figure below sets $\rev=0.1$ and depicts the continuum limit, $I=J\to\infty$.

We use \autoref{lemma0} to derive the optimal allocation policy.
For consumers with negative virtual types, the platform matches every item $\quality$ with ad $\disutil=2-\quality$, and allocates an item-ad pair $(\quality,\disutil)$ if and only if $\virtual(\type)\,(\quality-\disutil)+\rev\ge 0$, i.e.,
$\quality\le \overline{\quality}^{P}(\type):=1+\min\left\{1,\frac{\rev}{2|\virtual(\type)|}\right\}$.
Correspondingly, the platform allocates to each type $\type$ all ads whose disutility levels exceed $\underline{\disutil}^{P}(\type)=2-\overline{\quality}^{P}(\type)$.
For consumers with positive virtual types, the platform allocates all items and displays only ads that yield a {nonnegative} virtual ad profit, i.e., $\rev-\disutil\,\virtual(\type)\ge 0$, or equivalently
$\disutil\le \overline{\disutil}^{P}(\type):=\min\left\{1,\frac{\rev}{\virtual(\type)}\right\}$.

\autoref{example1} depicts the resulting policy.
The left panel shows that in the optimal mechanism, each consumer receives all items whose quality levels lie in $[1,\overline{\quality}^{P}(\type)]$.
Because the first best is to allocate all items to all types, the second-best allocation to negative virtual types reflects the standard underprovision of goods. 
But in this example, unlike standard screening, the platform allocates some items to every type.

The right panel presents the set of ads allocated to each type $\type$ under the first best and the second best. Under the first best (i.e., the orange dashed area), the platform allocates all ads whose disutility levels are below $\overline{\disutil}^{FB}(\type)=\min\left\{1,\frac{\rev}{\type}\right\}$.
Hence, the number of ads decreases with $\type$: As all ads generate the same revenue in this example, the first-best policy allocates fewer annoying ads to minimize welfare loss.

The right panel also presents the second-best ad policy, where each type $\type$ receives all ads whose disutility levels lie in $[\underline{d}^{P}(\theta),\overline{d}^{P}(\theta)]$ (i.e., the light blue area).\footnote{{Here, $\overline{d}^{P}(\theta)=1$ for negative virtual types and $\underline{d}^{P}(\theta)=0$ for positive virtual types.}}
For low types, the platform distorts the advertising policy in terms of both  the
 disutility levels and volume of advertising.
Given a volume of advertising, the platform allocates ads with the highest disutility levels, although total surplus would be maximized if it allocated ads with the lowest disutility levels.
At the same time, these consumers are exposed to fewer ads (i.e., $1-\underline{\disutil}^{P}(\type)<\overline{\disutil}^{FB}(\type)$), because they view fewer items.

For positive virtual types, the distortion is in advertising volume: 
The platform prioritizes allocating ads with low disutility levels, but it now allocates more ads than under the first best.
Negative virtual types close to $\zerotype$ may also be exposed to more ads than under the first best, but for them the platform continues to prioritize ads with high disutility levels.
Overall, the example illustrates that our model can capture rich ways in which the platform distorts its advertising policy away from the first best.

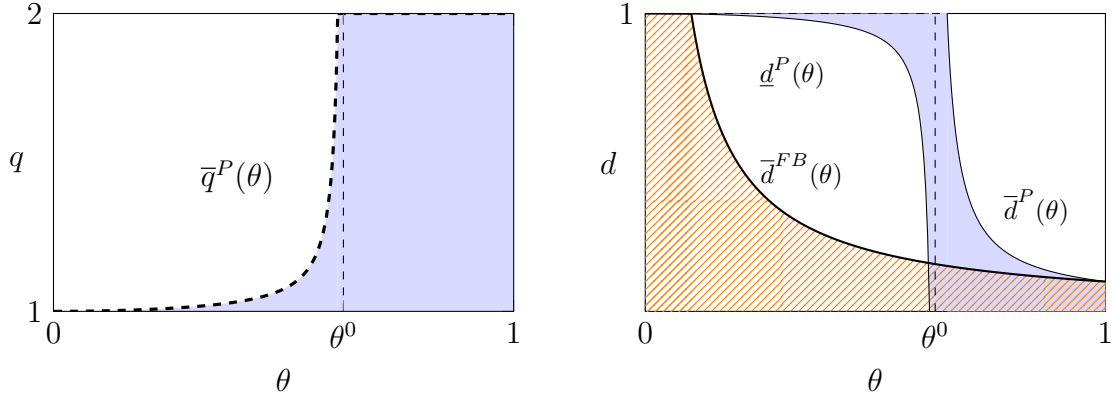
\begin{figure}[H]
\centering
\def\r{0.1}
\begin{minipage}{0.495\textwidth}
\centering
\begin{tikzpicture}
\begin{axis}[
  width=\linewidth, height=0.72\linewidth,
  xmin=0, xmax=1,
  ymin=1, ymax=2,
  xlabel={$\theta$},
  ylabel={$q$},
  ylabel style={rotate=-90, anchor=center},
  xtick={0,0.63,1},
  xticklabels={$0$,$\zerotype$,$1$},
  ytick={1,2},
  clip=true,
  axis on top
]
  \addplot[very thick, dashed, name path=curve,
           domain=0.001:0.61771, samples=200]
           {1 + 1.5*\r*x^2/(1 - 4*x^3)};

  \addplot[name path=baseline1, domain=0:0.61771, samples=2, draw=none] {1};
  \addplot[name path=baseline2, domain=0.61771:1, samples=2, draw=none] {1};
  \addplot[name path=topline, domain=0.61771:1, samples=2, draw=none] {2};
  \addplot[fill=blue!15, draw=none] fill between[of=curve and baseline1];
  \addplot[fill=blue!15, draw=none] fill between[of=topline and baseline2];

  \addplot[very thick, dashed] coordinates {(0.61771,2) (1,2)};

  \addplot[domain=0:1, samples=2, mark=none, transform canvas={yshift=0.8pt}] {2};
  \node[anchor=south east] at (axis cs:0.99,2) {$q=2$};

  \node at (axis cs:0.4,1.45) {$\overline{q}^{P}(\theta)$};
  \addplot[dashed] coordinates {(0.62996,1) (0.62996,2)};
\end{axis}
\end{tikzpicture}
\end{minipage}
\hfill
\begin{minipage}{0.495\textwidth}
\centering
\begin{tikzpicture}
\begin{axis}[
  width=\linewidth, height=0.72\linewidth,
  xmin=0, xmax=1,
  ymin=0, ymax=1,
  xlabel={$\theta$},
  ylabel={$d$},
  ylabel style={rotate=-90, anchor=center},
  xtick={0,0.63,1},
  xticklabels={$0$,$\zerotype$,$1$},
  ytick={1}
]
  \addplot[thick, domain=0:\r, samples=2] {1};
  \addplot[thick, domain=\r:1, samples=200] {\r/x};

  \addplot[domain=0.001:0.61771, samples=150] {1 - 1.5*\r*x^2/(1 - 4*x^3)};
  \addplot[domain=0.65598:1, samples=150] {3*\r*x^2/(4*x^3 - 1)};
  \addplot[dashed, domain=0:1, samples=2] {1};

  \addplot[name path=leftC,  domain=0.001:0.61771, samples=150, draw=none] {1 - 1.5*\r*x^2/(1 - 4*x^3)};
  \addplot[name path=rightC, domain=0.65598:1, samples=150, draw=none] {3*\r*x^2/(4*x^3 - 1)};
  \addplot[name path=topL, domain=0:0.61771,  samples=2, draw=none] {1};
  \addplot[name path=botR, domain=0.65598:1,  samples=2, draw=none] {0};
  \addplot[fill=blue!15, draw=none] fill between[of=topL and leftC];
  \addplot[fill=blue!15, draw=none] fill between[of=rightC and botR];
  \addplot[name path=topMid, domain=0.61771:0.65598, samples=2, draw=none] {1};
  \addplot[name path=botMid, domain=0.61771:0.65598, samples=2, draw=none] {0};
  \addplot[fill=blue!15, draw=none] fill between[of=topMid and botMid];

  \addplot[name path=bottomPartial, domain=0:\r,   samples=2,   draw=none] {0};
  \addplot[name path=bottomAll,    domain=0:1,     samples=2,   draw=none] {0};
  \addplot[name path=topCap,       domain=0:\r,    samples=2,   draw=none] {1};
  \addplot[name path=fbCurve,      domain=\r:1,    samples=200, draw=none] {\r/x};
  \addplot[pattern=north east lines, pattern color=orange, draw=none]
    fill between[of=topCap and bottomPartial];
  \addplot[pattern=north east lines, pattern color=orange, draw=none]
    fill between[of=fbCurve and bottomAll];

  \node[] at (axis cs:0.85, 0.35) {\footnotesize $\overline{d}^{P}(\theta)$};
  \node[] at (axis cs:0.32,0.8) {\footnotesize $\underline{d}^{P}(\theta)$};
  \node[] at (axis cs:0.34,0.47) {\footnotesize $\overline{d}^{FB}(\theta)$};
  \addplot[dashed] coordinates {(0.62996,0) (0.62996,1)};
\end{axis}
\end{tikzpicture}
\end{minipage}
\caption{The left panel depicts the set of allocated items in terms of their quality for each type under the second best.
The right panel depicts the set of allocated ads in terms of their disutility levels for each type under the first best (dashed orange area) and the second best (light blue area).
The vertical dashed line indicates $\zerotype$.
}
\label{example1}
\end{figure}

The example also illustrates how advertising affects consumers.
To this end, we compare the above platform with a \emph{subscription platform}, which hosts the same set of items but no ads.
Its optimal mechanism is a posted price of $\frac{3}{2}\zerotype$ for the set of all items, which consumers with types above $\zerotype$ purchase.
Let $CS_S \approx 0.239$ denote the resulting aggregate consumer surplus, and let $CS_A(\rev)$ denote the aggregate consumer surplus under the optimal mechanism of the ad-supported platform, viewed as a function of the per-ad revenue $\rev$.
\autoref{figureCScomparison} depicts the comparison: There is a threshold $\rev^* \approx 0.23$ such that advertising lowers consumer surplus for $\rev<\rev^*$ and raises it for $\rev>\rev^*$.
In particular, at the depicted value $\rev=0.1$, consumers are on average worse off than under the subscription platform (the right panel of \autoref{figureCScomparison}).

{Advertising affects consumers in two opposing ways: Consumers benefit when the platform uses ads to expand the item allocation to negative virtual types, but for a fixed item allocation, advertising harms consumers because of its nuisance.
	In this example, the negative effect dominates when $\rev$ is small, whereas the positive effect dominates when $\rev$ is large, which generates the single crossing between $CS_A(\rev)$ and $CS_S$.
	However, such a single crossing is not general.\footnote{For example, suppose that $\type \sim U[0,1]$ and the platform hosts two items with $(\quality(1),\quality(2))=(1.01,4)$ and two ads, each with disutility $1$ and revenue $\rev$.
In this case, as $\rev$ increases, $CS_A(\rev)$ crosses $CS_S$ from above at $\rev_1\approx0.02$ and from below at $\rev_2\approx0.84$.}}
Finally, the comparison of total surplus is also ambiguous; depending on the primitives, either platform can attain higher total surplus.

\medskip

\begin{figure}[H]
\centering
\begin{minipage}{0.495\textwidth}
\centering
\begin{tikzpicture}
\begin{axis}[
  width=\linewidth, height=0.8\linewidth,
  xmin=0, xmax=3,
  ymin=0.22, ymax=0.43,
  xlabel={$r$},
  ylabel={$CS$},
  ylabel style={rotate=-90, anchor=center},
  xtick={0, 1, 2, 3},
  ytick={0.239118, 0.3, 0.35, 0.4},
  yticklabels={$CS_S$, $0.3$, $0.35$, $0.4$},
  legend pos=north west,
  legend style={draw=none, font=\footnotesize},
  axis on top
]
  \addplot[very thick, myblue] coordinates {(0.0,0.239118) (0.01,0.239101) (0.02,0.239065) (0.03,0.239019) (0.04,0.238968) (0.05,0.238914) (0.06,0.238860) (0.07,0.238808) (0.08,0.238760) (0.09,0.238716) (0.1,0.238679) (0.11,0.238649) (0.12,0.238627) (0.13,0.238614) (0.14,0.238611) (0.15,0.238618) (0.16,0.238637) (0.17,0.238667) (0.18,0.238709) (0.19,0.238763) (0.2,0.238831) (0.21,0.238912) (0.22,0.239007) (0.23,0.239117) (0.24,0.239240) (0.25,0.239379) (0.26,0.239532) (0.27,0.239701) (0.28,0.239885) (0.29,0.240085) (0.3,0.240301) (0.31,0.240532) (0.32,0.240780) (0.33,0.241044) (0.34,0.241325) (0.35,0.241622) (0.36,0.241935) (0.37,0.242265) (0.38,0.242611) (0.39,0.242975) (0.4,0.243355) (0.42,0.244164) (0.44,0.245041) (0.46,0.245983) (0.48,0.246990) (0.5,0.248063) (0.52,0.249199) (0.54,0.250398) (0.56,0.251658) (0.58,0.252977) (0.6,0.254355) (0.65,0.258040) (0.7,0.262040) (0.75,0.266314) (0.8,0.270813) (0.85,0.275480) (0.9,0.280254) (0.95,0.285064) (1.0,0.289831) (1.05,0.294488) (1.1,0.299025) (1.15,0.303445) (1.2,0.307754) (1.25,0.311955) (1.3,0.316054) (1.35,0.320053) (1.4,0.323957) (1.45,0.327768) (1.5,0.331491) (1.55,0.335128) (1.6,0.338683) (1.65,0.342158) (1.7,0.345556) (1.75,0.348880) (1.8,0.352131) (1.85,0.355314) (1.9,0.358429) (1.95,0.361479) (2.0,0.364465) (2.1,0.370258) (2.2,0.375822) (2.3,0.381170) (2.4,0.386316) (2.5,0.391271) (2.6,0.396046) (2.7,0.400650) (2.8,0.405094) (2.9,0.409386) (3.0,0.413533)};
  \addlegendentry{$CS_A(\rev)$}
  \addplot[very thick, myred, dashed, domain=0:3, samples=2] {0.239118};
  \addlegendentry{$CS_S$}
\end{axis}
\end{tikzpicture}
\end{minipage}
\hfill
\begin{minipage}{0.495\textwidth}
\centering
\begin{tikzpicture}
\begin{axis}[
  width=\linewidth, height=0.8\linewidth,
  xmin=0, xmax=0.5,
  ymin=0.2380, ymax=0.2485,
  xlabel={$r$},
  xtick={0, 0.1, 0.23, 0.4},
  xticklabels={$0$, $0.1$, $\rev^*$, $0.4$},
  ytick={0.239118, 0.244, 0.248},
  yticklabels={$CS_S$, $0.244$, $0.248$},
  axis on top
]
  \addplot[very thick, myblue] coordinates {(0.0,0.239118) (0.01,0.239101) (0.02,0.239065) (0.03,0.239019) (0.04,0.238968) (0.05,0.238914) (0.06,0.238860) (0.07,0.238808) (0.08,0.238760) (0.09,0.238716) (0.1,0.238679) (0.11,0.238649) (0.12,0.238627) (0.13,0.238614) (0.14,0.238611) (0.15,0.238618) (0.16,0.238637) (0.17,0.238667) (0.18,0.238709) (0.19,0.238763) (0.2,0.238831) (0.21,0.238912) (0.22,0.239007) (0.23,0.239117) (0.24,0.239240) (0.25,0.239379) (0.26,0.239532) (0.27,0.239701) (0.28,0.239885) (0.29,0.240085) (0.3,0.240301) (0.31,0.240532) (0.32,0.240780) (0.33,0.241044) (0.34,0.241325) (0.35,0.241622) (0.36,0.241935) (0.37,0.242265) (0.38,0.242611) (0.39,0.242975) (0.4,0.243355) (0.42,0.244164) (0.44,0.245041) (0.46,0.245983) (0.48,0.246990) (0.5,0.248063)};
  \addplot[very thick, myred, dashed, domain=0:0.5, samples=2] {0.239118};
  \addplot[only marks, mark=*, mark size=1.8pt, black]
    coordinates {(0.230, 0.239118)};
  \addplot[dashed, gray] coordinates {(0.230, 0.2380) (0.230, 0.239118)};
\end{axis}
\end{tikzpicture}
\end{minipage}
\caption{Aggregate consumer surplus under the ad-supported platform, $CS_A(\rev)$ (solid blue), and under the subscription platform, $CS_S \approx 0.239$ (dashed orange), as functions of the per-ad revenue $\rev$.
The right panel zooms in on $\rev \in [0, 0.5]$: Advertising lowers consumer surplus for $\rev < \rev^* \approx 0.23$ and raises it for $\rev > \rev^*$.
}
\label{figureCScomparison}
\end{figure}
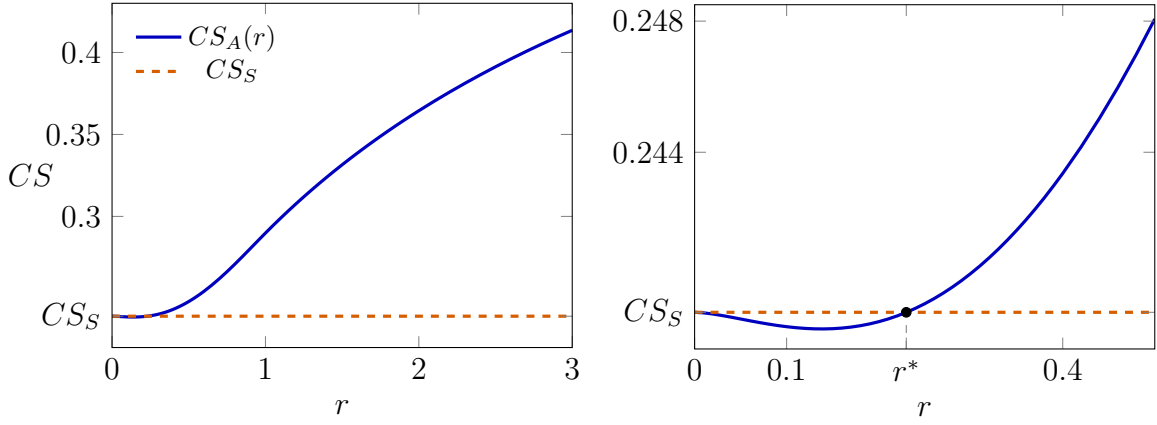

\end{example}

\section{Application: Rationalizing Platform Menus}\label{sectionApplications}

We now use the characterization above to rationalize contracts used by ad-supported platforms in practice.
The optimal mechanism can screen consumers through content access, advertising exposure, or both.\footnote{To be precise, the platform price discriminates in all three cases, so ``content access differentiation'' should be read as ``content access and price differentiation.''}
For each screening pattern, we first present the condition under which the optimal mechanism exhibits it.
We then point to platforms whose menus exhibit the corresponding pattern; Appendix~\ref{app:platform_menus} documents these menus in more detail.

Throughout the section, we impose two assumptions. 
First, the number of ads exceeds the number of items, $\adnumber\ge\itemnumber$; let $\rev(\itemnumber)$ denote the $\itemnumber$-th highest ad revenue in $\adset$.
Second, there is a constant $\cost>0$ such that $\disutil(j)=\cost\rev(j)$ for every $j\in\adset$. 
The parameter $\cost$ is the nuisance cost per unit of advertising revenue: Higher $\cost$ means that ads are more annoying relative to the revenue they generate. The specification captures the idea that more intrusive ads---longer, louder, less skippable, or more privacy-invasive---tend to generate more revenue but also impose greater nuisance.

Under this assumption, the $\type$-virtual ad profit of ad $j$ is $\rho_\type(j)=\rev(j)(1-\cost\virtual(\type))$, whose sign depends on $\type$ but not on $j$.
When a cutoff exists, let $\type^*>\zerotype$ solve $1-\cost\virtual(\type^*)=0$.
Then, all ads have positive virtual ad profit for types below $\type^*$ and negative for types above it. 
The two parameters, $\cost$ and $\rev(\itemnumber)$, generate three regions.

\paragraph{Case 1: Differentiated Content Access.}
Suppose $\cost<1$. 
Then, ad nuisance is low enough that every ad has positive virtual ad profit for every type, so whenever an item is allocated it is allocated with an ad. 
The platform can still differentiate by content access and price, with low types receiving no items or only items below a quality threshold and high types receiving all items.
However, the menu does not contain an ad-free contract. 
This fits menus in which a consumer can pay to access more content but cannot pay to remove ads.
For example, on news platforms such as the New York Times and the Financial Times, subscribers gain access to articles, archives, newsletters, or related products, but advertising is never eliminated.

\paragraph{Case 2: Differentiated Ad Exposure.}
Suppose instead that ads are annoying, $\cost>1$, but even the $\itemnumber$-th most profitable ad is lucrative,
\begin{equation*}
\rev(\itemnumber) > \frac{-\virtual(0)\quality(\itemnumber)}{1-\cost\virtual(0)},
\quad\text{equivalently}\quad
\virtual(0)\quality(\itemnumber)+\rev(\itemnumber)\bigl(1-\cost\virtual(0)\bigr)>0,
\end{equation*}
where $\virtual(0)<0$ is the virtual value of the lowest type and $\quality(\itemnumber)$ is the highest item quality. 
The condition means that even the lowest-type consumers generate a positive virtual surplus when they receive the highest-quality item matched with the ad that has the $\itemnumber$-th highest revenue.
The platform then finds it optimal to allocate all items to all consumers. 
At the same time, because $\cost>1$, the cutoff $\type^*$ is interior, and ads have negative virtual ad profit for types above it, so those types pay to remove ads. The menu differentiates through advertising exposure rather than content access, as on platforms such as YouTube and X, whose free or low-priced tiers offer broad access while premium tiers reduce or eliminate advertising.

\paragraph{Case 3: Differentiation in Both Dimensions.}
The remaining region is
\begin{equation}\label{equationNetflix}
\cost>1 \quad\text{and}\quad \rev(\itemnumber) < \frac{-\virtual(0)\quality(\itemnumber)}{1-\cost\virtual(0)}.
\end{equation}
In this case, the platform no longer finds it profitable to allocate all items to all consumers. 
As a result, the optimal menu differentiates along both content access and advertising (Figure~\ref{figureMenu}). 
Let $\type^{**} \in (0, \zerotype)$ denote the type at which the highest-quality item, matched with the ad with the $\itemnumber$-th highest revenue, first generates nonnegative virtual surplus, i.e., $\virtual(\type^{**})[\quality(\itemnumber) - \cost\rev(\itemnumber)] + \rev(\itemnumber) = 0$ (\autoref{sectionAppendixAB}).
Types above $\type^*$ receive all items without ads; types in $[\type^{**}, \type^*)$, including the negative virtual types in $[\type^{**}, \zerotype)$, receive all items with ads; and types below $\type^{**}$ receive a lower-contour bundle that excludes the highest-quality item, or are excluded, with every allocated item bundled with an ad.

This pattern appears on streaming services such as Netflix, Spotify, and Peacock, whose lower tiers combine limited access with advertising and whose higher tiers provide broader access without ads.
What distinguishes this region is that both cutoffs are interior: ads are lucrative enough that the platform still serves negative virtual types, but not lucrative enough to grant full access to the lowest types, while $\cost>1$ keeps $\type^*$ below $\typemax$.
Content access is the screening margin below $\type^{**}$, ad exposure is the screening margin at $\type^*$, and the middle tier pools consumers on both sides of $\zerotype$ at one price.

\begin{figure}[H]
  \centering
  \begin{tikzpicture}[x=13.5cm, y=2cm, thick]
    \def\thetaLow{0.4}
    \def\thetaZero{0.5}
    \def\thetaStar{0.7}
    \draw[thick] (0,0) -- (1,0);
    \foreach \x in {0, \thetaLow, \thetaStar, 1} {
      \draw[thick] (\x,0.1) -- (\x,-0.1);
    }
    \draw[thin, densely dashed] (\thetaZero,0.1) -- (\thetaZero,-0.1);
    \node[above=0.7cm] at (0,-0.25) {$0$};
    \node[above=0.7cm] at (\thetaLow+0.005,-0.25) {$\theta^{**}$};
    \node[above=0.7cm] at (\thetaZero+0.005,-0.25) {$\theta^0$};
    \node[above=0.7cm] at (\thetaStar+0.005,-0.25) {$\theta^*$};
    \node[above=0.7cm] at (1,-0.25) {$1$};
    \node[below=0.3cm, align=center] at (0.2,0) {$\bundle(\type)= \{i : \quality(i) \le \maxquality(\type)\} \subsetneq \qset$\\with ads};
    \node[below=0.3cm, align=center] at (0.55,0) {All items\\with ads};
    \node[below=0.3cm, align=center] at (0.85,0) {All items\\without ads};
  \end{tikzpicture}
  \caption{Optimal mechanism under condition~\eqref{equationNetflix}. The cutoff $\type^{**} \in (0, \zerotype)$ is defined in \autoref{sectionAppendixAB}; the tier of all items with ads contains $\zerotype$ in its interior.}
  \label{figureMenu}
\end{figure}
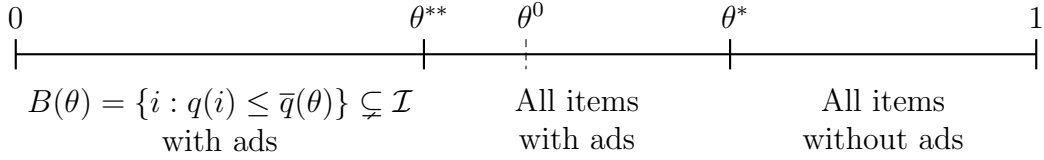

\paragraph{Comparison with Alternative Explanations.}
Each menu form can also arise from other mechanisms. First, differentiation in content access alone can arise from standard second-degree price discrimination:
	A platform may offer vertically differentiated tiers whose lower-priced tiers restrict content access \citep*{mussa1978,deneckere1996damaged}.
Second, differentiation in advertising exposure alone is studied in the freemium literature, in which a platform offers an ad-supported free tier alongside an ad-free paid tier.
This literature analyzes when the platform offers both tiers rather than relying exclusively on either advertising or subscriptions \citep*{sato2019freemium,zennyo2020freemium,cai2023freemium}.
Third, differentiation in both content access and advertising exposure can arise from two-sided incentives in media versioning \citep*{lin2020two}; an application of multidimensional screening in which consumers differ separately in content valuation and ad aversion (see \citealt*{rochetstole2003}); or behavioral heterogeneity.
For example, \citet*{acemoglu2024online} obtain subscription, ad-based, and mixed plans when users differ in their sophistication about digital advertising.\footnote{A related but more distant mechanism is provided by \citet*{corrao2023nonlinear}, who study goods whose usage generates revenue for the seller and that buyers can freely underutilize. Their mechanism rationalizes multipart tariffs and is relevant to freemium and usage-based pricing, but does not feature advertising nuisance or ad-content assignment.}

Our mechanism differs in two respects.
First, it generates all three menu forms from a one-dimensional type without separate ad aversion or heterogeneous sophistication.
The dual role of ads in generating revenue and imposing nuisance is crucial: Without nuisance, consumers would not pay to avoid ads; without ad revenue, the platform would not show ads.
The platform uses ad exposure as a screening device only when ads have both features.
Second, the platform chooses which items and ads to assign within a tier, and, when ad revenue depends on the matched item (\autoref{sectionItemAdRevenue}), our model additionally pins down which ads are paired with which items.
The other models discussed above abstract from this concern.

\section{Platform's Innovation Incentives}\label{sectionInnovation}
The rent extraction-advertising trade-off has an implication for the platform's incentive to improve content quality. To formally investigate this, we extend the model by incorporating the platform's choice of quality profile.

Throughout this section, we fix the number of items, $\itemnumber$.
We also take as given an arbitrary set of ads, $\adset^*$, along with their characteristics, ${(\rev(j), \disutil(j))}_{j \in \adset^*}$.
The set of ads hosted by the platform, which we call the \emph{advertising set}, is a subset of $\adset^*$ and is denoted by $\adset$.
Given $\itemnumber$ and $\adset$, the platform chooses the quality of each item, then adopts the optimal mechanism characterized in \autoref{lemma0}.
Our goal is to examine how the platform's optimal quality choice depends on the advertising set.

We introduce several pieces of notation and define the platform's problem.
First, let $\qvecset$ denote the set of feasible quality profiles:
\begin{align*}
\qvecset  := \{ (\quality(1),..., \quality(\itemnumber)) \in \mathbb{R}^\itemnumber_+: \quality(\itemnumber) \ge \cdots \ge \quality(1) \ge \max_{j \in \adset^*} \disutil(j)\}.
\end{align*}
The inequality $\quality(1) \ge \max_{j \in \adset^*} \disutil(j)$ implies that for any feasible quality profile, the net quality of any item-ad pair is nonnegative.

For any quality profile $\qvec \in \qvecset$ and advertising set $\adset\subseteq \adset^*$, let $\profit(\qvec, \adset)$ denote the platform's revenue from the optimal mechanism.
Let $C(\qvec) \in \R$ denote the cost of choosing a quality profile $\qvec$.
Assume that $C(\cdot)$ is submodular, which holds, e.g., if $C(\cdot)$ is additively separable, i.e., $C(\qvec) = \sum^{\itemnumber}_{i=1} C_i(\quality(i))$. Finally, for any $x, y \in \R^K_+$, we write $x  \ge y$ to mean $x_i \ge y_i$ for all $i=1,..., K$. 

Our specification implies that the cost of increasing the quality of a given item does not depend on the number of consumers who receive the item.
This assumption is natural for digital goods, which feature free replicability.

Define $\qvecset^*(\adset)$ as the set of solutions to the platform's quality choice problem:
\begin{equation}
\qvecset^*(\adset):= \arg\max_{\qvec \in \qvecset} \profit(\qvec, \adset) - C(\qvec).\tag{Q}\label{Q}
\end{equation}
Assume that the problem has a solution, i.e., $\qvecset^*(\adset)\not=\emptyset$ for any $\adset \subseteq \adset^*$.

The following lemma describes how a change in the advertising set affects the platform’s quality choice through a shift in virtual ad profits.
We say that advertising set $\adset_H$ \emph{attains uniformly higher virtual ad profits than} $\adset_L$ if for every type $\type$ and $i=1,..., \itemnumber$, $\max(0,\rho_\type(j^H_{\type, i})) \ge \max(0,\rho_\type(j^L_{\type, i}))$, where for each $x \in \{L, H\}$, $j^x_{\type, i}$ is the ad that generates the $i$-th highest $\type$-virtual ad profit in $\adset_x$ if it exists; otherwise (i.e., if $i>|\adset_x|$), $j^x_{\type, i}=\noad$, so that $\rho_\type(j^x_{\type, i})=0$.

\begin{lemma}\label{lemmaInvest}
If the advertising set $\adset_H$ attains uniformly higher virtual ad profits than $\adset_L$, the platform has a lower incentive to improve content quality under $\adset_H$ than under $\adset_L$, i.e., for any $\qvecH, \qvecL \in \qvecset$ such that $\qvecH \ge \qvecL$,
 we have $\profit(\qvecH, \adset_H)- \profit(\qvecL, \adset_H)\le \profit(\qvecH, \adset_L)- \profit(\qvecL, \adset_L)$.
 Consequently, $\qvecset^*(\adset_H)$ is smaller than $\qvecset^*(\adset_L)$ in the strong set order. 
\end{lemma}

The intuition is as follows.
According to \autoref{lemma0}, {the platform allocates} the item with the $i$-th lowest quality {together} with the ad with the $i$-th highest virtual ad profit, provided that the latter is nonnegative.
Thus, if $\adset_H$ attains uniformly higher virtual ad profits than $\adset_L$, the platform facing $\adset_H$ finds it more profitable to allocate each matched item-ad pair.
As a result, the platform allocates more items to consumers with negative virtual types.
However, the expanded allocation makes it more costly for the platform to improve item quality, because the loss due to information rents is proportional to item quality.
Therefore, the uniform increase in virtual ad profits discourages the platform from investing in content quality.

\autoref{fig:combined} illustrates \autoref{lemmaInvest} for a platform that hosts one item and one ad.
The ad imposes zero disutility, and thus the virtual ad profit is equal to its ad revenue, $\rev$.
The platform's cost of investment is given by $0.1\quality^2$, and consumer types are uniformly distributed on $[0,1]$.
As ad revenue $\rev$ increases, the platform invests less in item quality (the left panel) and allocates the item to more consumers (the right panel).
Once the ad revenue exceeds a threshold ($\rev = 10/27 \approx 0.37$), the platform sets the item quality to $0$ and allocates the item to all consumers for free.\footnote{{The downward jump arises because the platform's revenue from the optimal mechanism is convex in $\quality$, so the profit, which is the revenue minus a convex investment cost, need not be concave and may have multiple maximizers. In this example, the profit is bimodal in $\quality$, and at the cutoff ad revenue, the platform is indifferent between the two optimal quality levels; as $\rev$ crosses the cutoff, the global maximum switches from the higher level to $\quality = 0$. The jump is not a general feature: For other type distributions, the profit can have a unique maximizer and the optimal quality can decrease continuously in $\rev$.}}
In \autoref{exampleItemCostInvestment}, we allow the platform to incur a positive item-provision cost and show that the optimal quality choice first increases and then decreases with $\rev$.

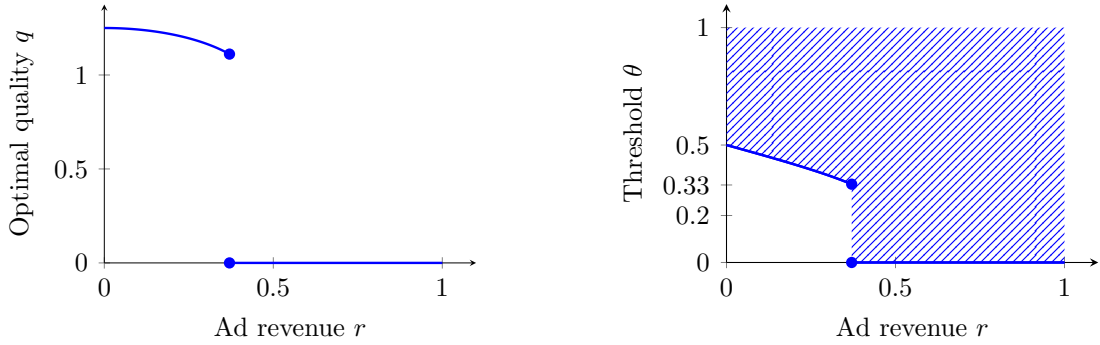
\begin{figure}[htbp]
    \centering
    \begin{subfigure}[t]{0.46\textwidth}
        \centering
        \begin{tikzpicture}
            \begin{axis}[
                clip=false,
                width=6.5cm,
                height=5cm,
                xlabel={Ad revenue $r$},
                ylabel={Optimal quality $q$},
                grid=none,
                axis lines=left,
                line width=0.08pt,
                tick label style={font=\footnotesize},
                label style={font=\footnotesize},
                enlargelimits=upper
            ]
                \addplot [line width=0.9pt, blue] table {data_points3_global.txt};
                \node[fill=black, circle, inner sep=1.5pt, blue] at (axis cs:0.370370, 1.111111) {};
                \node[fill=black, circle, inner sep=1.5pt, blue] at (axis cs:0.370370, 0) {};
            \end{axis}
        \end{tikzpicture}
        \label{fig:a}
    \end{subfigure}
    \hspace{0.05\textwidth}
    \begin{subfigure}[t]{0.46\textwidth}
        \centering
        \begin{tikzpicture}
            \begin{axis}[
                clip=false,
                width=6.5cm,
                height=5cm,
                xlabel={Ad revenue $r$},
                ylabel={Threshold $\type$},
                grid=none,
                axis lines=left,
                line width=0.08pt,
                tick label style={font=\footnotesize},
                label style={font=\footnotesize},
                enlargelimits=upper,
                ymin=0, ymax=1,
                ytick={0,0.2,0.33,0.5,1}
            ]
                \addplot [name path=phi, line width=0.9pt, blue]
                    table {data_points_phi_global.txt};

                \path [name path=top] (axis cs:0,1) -- (axis cs:1,1);

                \addplot [
                    pattern=north east lines,
                    pattern color=blue,
                    draw=none
                ] fill between[
                    of=top and phi,
                    soft clip={domain=0:1}  %
                ];

                \addplot [line width=0.9pt, blue]
                    table {data_points_phi_global.txt};
                \node[fill=black, circle, inner sep=1.5pt, blue] at (axis cs:0.370370, 0) {};
                \node[fill=black, circle, inner sep=1.5pt, blue] at (axis cs:0.370370, 0.333333) {};
            \end{axis}
        \end{tikzpicture}
        \label{fig:b}
    \end{subfigure}

    \caption{The left panel depicts the optimal quality as a function of $r \in [0,1]$ when $C(\quality) = 0.1\quality^2$ and $\theta \sim U[0,1]$.
At a unique cutoff $r = 10/27 \approx 0.37$, where the optimal quality drops to $0$, there are two optimal quality levels.
The right panel depicts the set of types that receive the item at each $r$, where the blue line captures the lowest type that receives the item.}
    \label{fig:combined}
\end{figure}

\autoref{lemmaInvest} implies that the platform invests less in content quality if it hosts a larger set of ads or faces a higher ad revenue for each ad.
\begin{proposition}\label{propositionInvest0}
Take two advertising sets, $\adset_H$ and $\adset_L$.
Suppose we have (i) $\adset_H \supset \adset_L$ or (ii) we obtain $\adset_H$ by increasing the ad revenue $\rev(j)$ of each ad $j\in \adset_L$ while keeping disutility levels the same.
Then, $\adset_H$ attains uniformly higher virtual ad profits than $\adset_L$, and thus the platform chooses lower quality profiles under $\adset_H$ than under $\adset_L$ in the strong set order.
\end{proposition}
\begin{proof}
If the advertising set expands, the $i$-th highest virtual ad profit increases for every $i$.
Increasing the ad revenue for each ad has the same effect because it increases virtual ad profits \eqref{equationVirtualAd} across all consumer types.
Thus, either change leads to uniformly higher virtual ad profits.
By \autoref{lemmaInvest}, the platform chooses lower quality profiles under $\adset_H$ than under $\adset_L$.
\end{proof}

As we discuss in the next section, \autoref{propositionInvest0} allows us to examine how policy or technological changes surrounding advertising may affect the platform's quality choice.
In addition, \autoref{propositionInvest0} allows us to examine the relation between the platform's business model and its quality choice.
To formalize this idea, we categorize mechanisms into three business models:
The platform adopts a \emph{subscription} model if it displays no ads (which occurs if and only if $\adset = \emptyset$);
a \emph{purely ad-funded} model if the optimal mechanism sets a price of zero for all types;
and a \emph{hybrid} model if the platform adopts neither a subscription nor a purely ad-funded model.

Both the platform's business model and quality choice are endogenous and depend on the advertising set.
Thus, by varying the advertising set, we obtain different pairs of optimal business models and quality profiles.
The following result illustrates the resulting relationship between the platform's business model and quality choice.
\begin{proposition}\label{propositionInvest1}
A platform with a subscription model chooses a higher quality profile than a platform with a hybrid model  (in the strong set order), which, in turn, chooses a higher quality profile than a purely ad-funded platform.
\end{proposition}

Intuitively, a platform whose business model relies more heavily on advertising also allocates more items to negative virtual types. 
Such an allocation rule discourages the platform from investing in content quality. 
We also note that a purely ad-funded model could be interpreted as a platform that cannot use monetary transfers for an exogenous reason, rather than one that optimally sets a zero price for all types. 
With this interpretation, \autoref{propositionInvest1} holds verbatim.

{Finally, we note that our results abstract from the following countervailing force: Higher-quality content may attract more attention and increase the value of advertising.
We capture this force in two extensions. 
First, \autoref{sectionItemCost} assumes costly provision and consumption of items, under which higher ad revenue expands allocation among positive virtual types.
Second, \autoref{sectionItemAdRevenue} allows ad revenue to increase in item quality.
In either extension, the platform's optimal quality choice can increase in ad revenue.}

\section{Policy Implications}\label{sectionPolicy}

Propositions \ref{propositionInvest0} and \ref{propositionInvest1} provide two sets of policy implications.
First, the results clarify the possible impacts of policy and technological changes on platforms' incentives to invest in content quality. We illustrate this first implication with two examples.

The first example concerns a recent law enacted in California that applies the Commercial Advertisement Loudness Mitigation (CALM) Act to streaming services.\footnote{For details of the law, see \url{https://leginfo.legislature.ca.gov/faces/billNavClient.xhtml?bill_id=202520260SB576} (accessed on September 11, 2026).}
The CALM Act is a federal law that requires television commercials to be no louder than the average volume of the program.
The recent law in California applies the same restriction to platforms such as Netflix and Hulu.

To examine the potential impact of the CALM Act in our model, suppose the platform faces a baseline set of ads, $\adset$.
For each ad $j \in \adset$, $(\rev(j), \disutil(j))$ represents the ad revenue and disutility level when the platform plays the ad at the same loudness as the main content.
The platform can also play each ad more loudly.
The louder version of ad $j$ is defined as ad $\hat{j}$ such that $\rev(\hat{j}) > \rev(j)$ and $\disutil(\hat{j}) > \disutil(j)$, i.e., louder ads are more annoying but also more effective at grabbing users' attention. 
Thus, without any regulation, the advertising set is given by $\hat{\adset} = \adset \cup \{\hat{j}\}_{j \in \adset}$.

The CALM Act shrinks the advertising set from $\hat{\adset}$ to $\adset$.
As formalized by \autoref{propositionInvest0}, such a regulation limits the platform's ability to degrade the quality of offerings to negative virtual types, which, in turn, incentivizes the platform to invest more in content quality and allocate it to consumers who have higher valuations.

The second example concerns improvements in advertising technology.
While we do not explicitly model the relevance of each ad to consumers, we may capture better targeting---such as the platform's ability to display relevant ads when users are likely to respond---as an increase in the ad revenue associated with each $j \in \adset$.
\autoref{propositionInvest0} then implies that such improvements, which expand the allocation of items to negative virtual types, may reduce the platform's incentives to invest in content quality.
An increase in each $\rev(j)$ might also arise if the platform can extract more surplus from (unmodeled) advertisers, either through greater market power or through improved information about advertisers' willingness to pay.\footnote{In an earlier version of this paper, we explicitly modeled advertisers as strategic players privately informed of their valuation of impressions. We showed that when the platform learns advertisers' types, which eliminates their private information, this can be captured as an increase in each $\rev(j)$. See  \citet*{ichihashi2024mechanism}.}
In summary, while the policy and technological changes described above may seem unrelated, our model provides a unified way to understand their potential impacts as a shift of virtual ad profits.

The second policy implication is that existing concerns about the negative impact of the ad-supported business model on quality and innovation  \citep*{cma2020online,stigler} apply not only to purely ad-funded platforms but also to hybrid platforms, which combine both advertising and direct pricing.
We provide a general rationale for why  this concern is relevant for hybrid platforms from the mechanism-design perspective: Reliance on advertising makes it costly for the platform to improve content quality because of the associated increase in information rents.

The underinvestment result for hybrid platforms in \autoref{propositionInvest1} is particularly relevant to media platforms.
\citet*{latham} note that ``an ad-funded business model (which) might have knock-on effects for media plurality by, for example, reducing incentives for investment in content generation." Moreover,
if artificial intelligence disproportionately lowers the cost of producing low-quality news articles relative to high-quality ones, this strengthens the negative impact of increased ad profitability on the quality of journalism.\footnote{For instance, AI can generate content automatically by converting data into informative and narrative texts,
leading to the production of thousands of stories with little or no human intervention \citep*{Noain}.}

Finally, while Propositions \ref{propositionInvest0} and \ref{propositionInvest1} identify changes in the advertising environment or business model that lead to lower content quality, the resulting impact on consumers is ambiguous.
Suppose, for example, that advertising becomes more profitable as in \autoref{propositionInvest0}.
On the one hand, the platform reduces its quality profile, which lowers every consumer's utility for a fixed allocation rule.\footnote{In any incentive-compatible mechanism such that the IR  constraint for the lowest type binds, the interim payoff of type $\type$ is given by  $\int^\type_0\sum_{i \in \bundle(t)} \left[ \quality (i) -  \disutil(\matching(i|t)) \right]  \dint t$.
For a fixed allocation rule, this expression decreases as the quality of each item decreases.}
On the other hand, a platform that hosts many profitable ads allocates more items to low-valuation consumers, which raises their utilities. This expansion can also raise the information rents of higher-valuation consumers; \autoref{figureCScomparison} illustrates how this force can increase aggregate consumer surplus for a given set of contents.
Thus, greater reliance on advertising may shift the platform's focus from quality toward the quantity of allocation, which may benefit or harm consumers.

\section{Extensions}\label{sectionExtension}
We discuss several extensions that relax some of the key assumptions.

\subsection{Costly Item Provision and Consumption}\label{sectionItemCost}
In our baseline model, distributing and consuming items are costless for the platform and consumers.
This assumption implies that no quality screening occurs for consumers with positive virtual types; i.e., all such consumers receive all items.

Quality screening for positive virtual types can arise when distributing and consuming items are costly.
Specifically, in Supplemental Appendix~\ref{sectionAppendixItemCost}, we assume that the platform incurs a cost of providing items and consumers incur a type-independent attention cost.
The cost of providing items may capture, for example, royalties that the platform pays to content creators each time their content is played.
{We assume that consumption is contractible, so consumers consume every allocated item.}

In this extension, we first characterize the optimal mechanism.
If higher-quality items have higher total item costs (i.e., the sum of provision and attention costs), the allocation for negative virtual types reduces to the pattern in Part 1 of \autoref{theorem}.
Moreover, if item costs are constant, the allocation to positive virtual types can be the mirror image of that to negative virtual types:
They receive items above a type-dependent quality threshold, which decreases with type, and item-ad matching is positive assortative in item quality and virtual ad profits.

We also show that higher ad profitability may increase the platform's incentive to invest in quality:
In a one-item example, the platform's optimal quality choice can increase in ad revenue, because increasing ad revenue may expand allocation among positive virtual types, which raises the marginal return to quality.

\subsection{{Free Disposal with Attention Costs}}\label{sectionNoncont}
{Attention costs can also induce quality screening for positive virtual types when consumption is not contractible.
In Supplemental Appendix~\ref{sectionAppendixFreeDisposal}, we assume that the platform does not incur the cost of providing items, but consumers incur a common attention cost of $a>0$ per item.
We allow consumers to freely dispose of allocated item-ad pairs.
Thus, a type-$\type$ consumer's gross utility from consuming item $i$ matched with ad $j$ is $\type [\quality(i)-\disutil(j)]-\attention$.
The combination of the attention cost and free disposal limits the space of implementable allocation rules, because a consumer ignores any item-ad pair for which $\type [\quality(i)-\disutil(j)]-\attention<0$.
Under a certain set of assumptions, we characterize the optimal mechanism and show that positive virtual types receive items above a quality threshold that decreases with type (\autoref{propositionICO}).
The main technical challenge is that incentive-compatible mechanisms now need to prevent double deviations, whereby consumers misreport their type and then consume a strict subset of allocated items.
To solve the problem, we adopt an approach inspired by \citet*{corrao2023nonlinear}: We first solve the problem without double deviations and then verify that its solution satisfies all IC constraints.}

\subsection{Item-Ad-Specific Advertising Revenue}\label{sectionItemAdRevenue}
In Supplemental Appendix~\ref{sectionAppendixItemAdRevenue}, we allow the ad revenue from matching ad $j$ with item $i$ to be a function $A(i,j)$.
This extension reveals three insights.
First, for any $A(\cdot, \cdot)$, some of the properties of the optimal mechanism in \autoref{theorem} continue to hold: Every item allocated to a negative virtual type is matched with an ad; and within positive virtual types, every consumer type receives all items, and higher types are exposed to fewer ads.
Second, higher-quality content may attract more attention and raise ad revenue. If $A(i,j)=R(\rev(j),\quality(i))$ with $R$ having strictly increasing differences, then for any consumer type, within the set of allocated item-ad pairs, higher-quality items are matched with higher-revenue ads.
Finally, the standard relaxed-problem approach is valid in this setup, and we present simple algorithms to solve the relaxed problem.

\autoref{remarkPairDependentNuisance} considers an extension in which a given ad becomes more annoying when matched with a higher-quality item.
Formally, the nuisance from an item-ad pair has increasing differences in item quality $\quality(i)$ and the ad's baseline disutility level $\disutil(j)$. To examine the implications of item-dependent disutility for item-ad matching, fix a consumer type and a set of allocated items and ads. For positive virtual types, the platform wants to reduce total nuisance, so it prefers to match higher-quality items with less annoying ads. For negative virtual types, the platform instead benefits from greater nuisance and prefers to match higher-quality items with more annoying ads. To compare this nuisance effect with the revenue effect above, suppose that higher-revenue ads are also more annoying. For positive virtual types, earning more ad revenue favors matching these ads with higher-quality items, while reducing nuisance favors matching them with lower-quality items. For negative virtual types, the two effects lead to the same matching pattern: both favor matching higher-revenue, more annoying ads with higher-quality items.

\subsection{Type-Independent Advertising Disutility} \label{sectionIndep}
Our model subsumes a setup in which ads impose both type-dependent and type-independent nuisance.
Each ad $j$ is characterized by $(\rev(j), \disutil_1(j), \disutil_2(j)) \in \R^3_+$, and a consumer's utility from an item-ad pair $(i,j)$ is $\type [\quality(i) - \disutil_1(j)] - \disutil_2(j)$.
The new component is $\disutil_2(j)$, which captures a type-independent but ad-dependent disutility from ad $j$.
This model is equivalent to ours, because we can redefine ad revenue as $\rev(j) - \disutil_2(j)$ and then normalize each $\disutil_2(j)$ to $0$.
This modification does not affect the optimal mechanism, consumers' interim payoffs, or the platform's revenue.\footnote{The modification might affect the analysis if the platform is constrained to using only nonnegative monetary transfers.
Indeed, the nonnegative price constraint binds whenever the platform tries to induce the participation of the lowest type $\type=0$ and exposes type $0$ to ads with positive $\disutil_2(j)>0$.
In contrast, the constraint never binds if $\disutil_2(j)=0$. However, it seems natural to assume a constant type-independent consumption utility if we introduce a constant type-independent disutility. Then, the constraint may not bind.}

\subsection{Many-to-One Matching} \label{remark1}
Our results hold even when each \content can be matched with at most $K \in \mathbb{N}$ \ads, or alternatively, when $\adnumber = K \itemnumber$ holds and each \content must be matched with exactly $K$ \ads.
For example, in the first case, the platform matches, for each $\type$, the lowest-quality \content with the $K$ \ads that have the highest virtual \ad profits, provided they are nonnegative.
The platform then matches the second-lowest-quality \content with the next $K$ highest nonnegative virtual \ad profits, and so on.
In either model, \autoref{theorem} and the results in \autoref{sectionInnovation} hold.\footnote{Both extensions require the lowest \content quality to exceed the total disutility of any $K$ \ads, which is an analogue of our assumption $\quality(i) - \disutil(j) \geq 0$ for any $(i, j)$. When each \content must be matched with exactly $K$ \ads, the platform similarly matches, for each negative virtual type, \contents in increasing order of quality with successive blocks of $K$ \ads ordered by decreasing virtual \ad profit, where the virtual \ad profit of a block is the sum of the virtual \ad profits of the \ads in that block.
The platform allocates the resulting bundles in this order and stops when the next bundle would generate negative virtual surplus.}

\section{Conclusion}
Digital platforms, such as social media, streaming services, and news organizations, all serve content and ads, but they adopt widely different contracts.
Meanwhile, ad-supported platforms have become a focus of policy debates, amid concerns that ad-funded business models may weaken platforms' incentives to invest in quality 
\citep*{scottmorton, cma2020online, cremer2019competition, stigler}.

In this paper, we study mechanism design for ad-supported platforms.
Our results rationalize various monetization strategies observed in practice and clarify how advertising affects a platform’s incentives to invest in content quality.
The results are derived from unified economic forces that are captured by two key concepts: virtual ad profits and the rent extraction--advertising trade-off.
These concepts will be relevant in richer models, such as those with platform competition or strategic content providers, which we leave for future research.

\begingroup
\catcode`\&=12
\putbib[freemium]
\endgroup

\appendix
\renewcommand{\thesection}{\Alph{section}}
\renewcommand{\theequation}{\thesection.\arabic{equation}}
\setcounter{equation}{0}
\begin{center}
\LARGE \textbf{Appendix}
\end{center}

\section{Proofs of \autoref{lemma0}, \autoref{theorem}, and \autoref{corollaryThreshold}} \label{sectionAppendixA}
\subsection{Proof of \autoref{lemma0}} \label{sectionAppendixA1}
We solve problem \eqref{equationVirtual3} for any given type $\type$.
Hereafter, we refer to the $\type$-virtual ad profit simply as the virtual ad profit.
We begin with the case of negative virtual types.

First, the platform prioritizes displaying ads with higher virtual ad profits and never displays ads with negative ones.
Also, the platform is restricted to displaying at most one ad for each item.
Thus, instead of the full set $\adset$ of ads, we can focus on ads that could potentially be displayed. 
Specifically, let $\adnumber^+$ be the number of ads that have a strictly positive virtual ad profit. 
If $\adnumber^+< \itemnumber$---i.e., if the number of such ads is strictly smaller than the number of items---then let $\adset^*$ be the set that consists of (i) all ads with strictly positive virtual ad profits and (ii) $\itemnumber-\adnumber^+$ fictitious ads, which have zero virtual ad profits.
If $\adnumber^+ \ge \itemnumber$, then let $\adset^*$ be the set that consists of ads that have the first through the $\itemnumber$-th highest virtual ad profits.
The set $\adset^*$ depends on type $\type$.

The platform can now focus on a subclass of advertising policies that specify a one-to-one  exhaustive matching between the items in $\qset$ and the ads in $\adset^*$, where matching with a fictitious ad corresponds to matching with no ad in our original setup.\footnote{By one-to-one  exhaustive matching, we mean that every item is matched with some ad in $\adset^*$.}
Let $\matchingset^*$ be the set of such advertising policies. 
Then, we can write the platform's relaxed problem for type $\type$ as
\begin{align}
\max_{\matching \in \matchingset^*} \sum_{i \in \qset} \max\left(0, \virtual(\type) \quality(i) + \rho_\type(\matching(i)) \right).\label{equationVirtual4}
\end{align}
We show that the function $\vs_\type (\quality, \rho) =  \max\left(0, \virtual(\type) \quality + \rho  \right)$ is submodular in $(\quality, \rho)$.
Take any $\quality_L$ and $\quality_H$ such that $\quality_H> \quality_L$.
We have 
\begin{align*}
\vs_\type (\quality_H,\rho)  - \vs_\type (\quality_L,\rho)
=
\begin{cases}
0 \quad &\text{if}\quad \rho < - \virtual(\type)\quality_L\\
-\virtual(\type) \quality_L - \rho  \quad &\text{if}\quad  - \virtual(\type)\quality_L \le \rho \le  - \virtual(\type)\quality_H \\
\virtual(\type)(\quality_H - \quality_L) \quad &\text{if}\quad \rho >- \virtual(\type)\quality_H,
\end{cases}
\end{align*}
which is overall decreasing in $\rho$.
Thus $\vs_\type (\quality, \rho)$ is submodular in $(\quality, \rho)$.
As a result, a negative assortative matching between items in $\qset$ and ads in $\adset^*$ solves \eqref{equationVirtual4} \citep*[e.g.,][Theorem 3.4 and Theorem 4.3]{galichon2018optimal}.

We now translate the solution of \eqref{equationVirtual4} to that of \eqref{equationVirtual3}.
First, the negative assortative matching between items in $\qset$ and ads in $\adset^*$ is equivalent to 
the $\rho_\type $-negative assortative policy, i.e., they match lower-quality items with ads that have higher values of $\rho_\type(j)$, until we exhaust ads that have positive values of $\rho_\type(j)$.
Let $\matching^*(\cdot| \type): \qset \to \adset\cup \{\emptyset\}$ denote the optimal advertising policy (for type $\type$) in our original formulation.
The platform then allocates each item-ad pair $(i, \matching^*(i|\type))$ if and only if $\virtual(\type) \quality(i) + \rho_\type ( \matching^*(i| \type)  ) \ge 0$, which reduces to the allocation policy described in the lemma.

The case of positive virtual types follows the same logic.
In this case, $\vs_\type (\quality_H,\rho)  - \vs_\type (\quality_L,\rho)  = \virtual(\type)(\quality_H - \quality_L)$ is independent of $\rho$ and thus trivially submodular.
To implement the tie-breaking in \autoref{definitionAssortative}, match any available zero-virtual-profit ads to remaining items before using the no-ad option; this leaves virtual surplus unchanged.

Having solved the relaxed problem, we now verify the monotonicity of allocation.
Given the allocation that solves the relaxed problem, we adopt the following notation:
$\grossq(\type) \triangleq \sum_{i \in \bundle(\type)} \quality(i)$, $\grossd(\type) \triangleq \sum_{i \in \bundle(\type)} \disutil( \matching^*(i|\type))$, and $\grossr(\type) \triangleq \sum_{i \in \bundle(\type)} \rev(\matching^*(i|\type))$.
The monotonicity constraint requires that $\grossq(\type) - \grossd(\type)$ is non-decreasing in $\type$.

Suppose to the contrary that there are some $\type_L, \type_H$ such that $\type_H> \type_L$ but $\grossq(\type_L) - \grossd(\type_L)> \grossq(\type_H) - \grossd(\type_H)$.
At the solution to the relaxed problem, we have 
\begin{equation*}
\virtual(\type_L)  [\grossq(\type_L) - \grossd(\type_L)]  
+ \grossr(\type_L)  \ge \virtual(\type_L)  [\grossq(\type_H) - \grossd(\type_H)]   + \grossr(\type_H).
\end{equation*}
If we replace $\virtual(\type_L)$ by $\virtual(\type_H)$ in both sides, the LHS increases strictly more than the RHS, because $\grossq(\type_L) - \grossd(\type_L)> \grossq(\type_H) - \grossd(\type_H) \ge 0$.
The resulting inequality is
\begin{equation*}
\virtual(\type_H)  [\grossq(\type_L) - \grossd(\type_L)]  + \grossr(\type_L)  > \virtual(\type_H)  [\grossq(\type_H) - \grossd(\type_H)]  
+ \grossr(\type_H).
\end{equation*}
This is a contradiction, because the platform could have increased its virtual surplus by replacing the outcome for type $\type_H$ with that for $\type_L$. 
Because the item and ad allocation satisfies monotonicity, the solution to the relaxed problem constitutes the optimal mechanism.
Finally, the binding IR constraint for type $\type=0$ and local IC constraints pin down the monetary transfer as shown in the lemma (see, e.g., Lemma 2 of \citealt*{myerson1981optimal}). \hfill $\square$

\subsection{Proof of \autoref{theorem}} \label{sectionAppendixA2}
We adopt the following notation: For any item $i \in \qset$, ad $j \in \adset \cup \{ \emptyset\}$, and type $\type \in [0,1]$, let $\vs_\type(i, j) = \virtual(\type)\quality(i) + \rho_\type(j)$ be the contribution to the virtual surplus when the platform allocates an item-ad pair $(i, j)$ to type $\type$.

First, we prove Part 1, the case of negative virtual types.
The platform allocates item $i$ to type $\type$ iff $\virtual(\type) \quality(i) + \rho_\type ( \matching^*(i| \type)  ) \ge 0$.
Under the $\rho_\type$-negative assortative policy, a higher-quality item is matched with a lower virtual ad profit.
Thus, the LHS is strictly decreasing in $\quality(i)$, because $\virtual(\type) \quality(i)$ is decreasing in $\quality(i)$ and a higher $\quality(i)$ is matched with a lower virtual ad profit.
Therefore, the platform allocates a lower-contour bundle:
\[
\bundle(\type)=\{i\in\qset:\quality(i)\le \maxquality(\type)\}.
\]
If the platform allocates no items to type $\type$, choose $\maxquality(\type)<\quality(1)=\min_{i\in\qset}\quality(i)$; otherwise, let $\maxquality(\type)=\quality(i)$ if the highest quality level allocated to type $\type$ equals $\quality(i)$.
For a negative virtual type, an allocated item matched with no ad contributes $\virtual(\type)\quality(i)\le0$ to virtual surplus.
Under the allocation rule, such an item can be allocated only when this contribution is zero.
Because $\virtual(\type)<0$, zero contribution with no ad implies $\quality(i)=0$, so omitting the pair leaves virtual surplus and consumer utility unchanged, and therefore leaves IC unchanged.
Hence every allocated item can be treated as matched with an ad with positive $\rho_\type(j)$.

We show that $\maxquality(\type)$ is increasing in types on $[0, \zerotype)$.
Suppose to the contrary that there are some $\type_L$ and $\type_H$ with $\type_L< \type_H < \zerotype$ such that the lower type $\type_L$ has a strictly higher lower-contour cutoff than type $\type_H$.
Let $i_L$ denote the highest-quality item that type $\type_L$ receives, which implies that type $\type_H$ does not receive item $i_L$.
By the selection argument above, type $\type_L$ receives item $i_L$ along with some ad $j_L\in\adset$ with $\rho_{\type_L}(j_L)>0$.

There are now two cases to consider.
First, suppose that ad $j_L$ is not allocated to type $\type_H$.
Note that $\vs_{\type_L}(i_L, j_L)  \ge 0$ implies $\vs_{\type_H}(i_L, j_L)  \ge 0$.
Let $k_H := |\bundle(\type_H)|$ denote the number of items that type $\type_H$ receives, and let $\mu_H$ denote a $\rho_{\type_H}$-permutation under which the advertising policy for type $\type_H$ is $\rho_{\type_H}$-negative assortative.
Because type $\type_H$ receives a lower-contour bundle and, by the selection argument above, every allocated item is matched with an ad, type $\type_H$ receives items $1, \ldots, k_H$ together with ads $\mu_H(1), \ldots, \mu_H(k_H)$.
Ad $j_L$ is not among these ads, so $k_H < \adnumber$ and $\rho_{\type_H}(j_L) \le \rho_{\type_H}(\mu_H(k_H+1))$ by the definition of a $\rho_{\type_H}$-permutation.
Item $i_L$ is not among these items, so $\quality(i_L) \ge \quality(k_H+1)$ and, because $\virtual(\type_H) < 0$, $\virtual(\type_H)\quality(i_L) \le \virtual(\type_H)\quality(k_H+1)$.
Adding the two inequalities, we obtain $\vs_{\type_H}(k_H+1, \mu_H(k_H+1)) \ge \vs_{\type_H}(i_L, j_L) \ge 0$.
The $\rho_{\type_H}$-negative assortative policy matches item $k_H+1$ with ad $\mu_H(k_H+1)$, and the allocation rule in \autoref{lemma0} includes every item whose pair has a nonnegative contribution (ties are broken in favor of allocation).
Hence $k_H+1 \in \bundle(\type_H)$, which contradicts $|\bundle(\type_H)| = k_H$.

Second, suppose that the platform allocates ad $j_L$ to type $\type_H$ along with some item other than item $i_L$.
The fact that type $\type_L$ receives pair $(i_L, j_L)$ implies that any item allocated to type $\type_L$, whose quality is lower than $i_L$, comes with some ad, because the advertising policy for $\type_L$ is the $\rho_{\type_L}$-negative assortative matching.
Then, it holds that type $\type_L$ receives a strictly greater number of ads than type $\type_H$, so there must be some ad, say $\hat{j}_L$, which is allocated to $\type_L$ but not to $\type_H$.
We have $\rho_{\type_L}(\hat{j}_L) \ge \rho_{\type_L}(j_L)$, because ad $\hat{j}_L$ is matched with a lower-quality item than ad $j_L$.
Thus, we have $\vs_{\type_L}(i_L, \hat{j}_L)  \ge \vs_{\type_L}(i_L, j_L) \ge 0$, which implies $\vs_{\type_H}(i_L, \hat{j}_L)   \ge \vs_{\type_L}(i_L, \hat{j}_L)  \ge 0$.
By the same argument as in the previous paragraph, if the unallocated pair $(i_L, \hat{j}_L)$ has a nonnegative contribution to the virtual surplus created by $\type_H$, then the platform must have allocated type $\type_H$ either $(i_L, \hat{j}_L)$ or another item-ad pair from the set of unallocated item-ad pairs.
This is a contradiction.
Therefore, $\maxquality(\type)$ is increasing on $[0,\zerotype)$.

Define $K^-(\type)=|\bundle(\type)|$.
By the selection argument above, every allocated item is matched with an ad with positive $\rho_\type(j)$.
The $\rho_\type$-negative assortative policy matches lower-quality items with higher values of $\rho_\type(j)$ until positive values are exhausted.
Therefore, type $\type$'s $K^-(\type)$ allocated items are matched with the $K^-(\type)$ ads with the highest positive values of $\rho_\type(j)$.
Because (i) lower-contour bundles expand as $\maxquality(\type)$ increases and (ii) no two items have the same quality level, $K^-(\type)$ is increasing in $\type$.

We now prove Part 2, the case of positive virtual types.
The first statement holds because, for any $i\in\qset$ and $j\in\adset\cup\{\emptyset\}$, $\virtual(\type)[\quality(i)-\disutil(j)]+\rev(j)\ge0$ when $\virtual(\type)>0$.
Equivalently, $\virtual(\type)\quality(i)+\rho_\type(j)\ge0$ for every matched pair, so $\bundle(\type)=\qset$.
Define $\adnumber^+_\type=|\{j\in\adset:\rho_\type(j)\ge0\}|$ and $K^+(\type)=\min(\itemnumber,\adnumber^+_\type)$.
The optimal advertising policy uses exactly $K^+(\type)$ ads, namely the $K^+(\type)$ ads with the highest nonnegative values of $\rho_\type(j)$.
As $\type$ increases, $\rho_\type(j) = \rev(j) - \virtual(\type) \disutil(j)$ is decreasing for any $j$, so $\adnumber^+_\type$ is decreasing in $\type$.
Thus, $K^+(\type)$ is decreasing in $\type$.
\hfill $\square$

\subsection{Proof of \autoref{corollaryThreshold}}\label{sectionAppendixA3}
{Fix any type $\type < \zerotype$.
Let $\rho_\type(j) = \rev(j) - \virtual(\type) \disutil(j)$ and $\rho'_\type(j) = \rev'(j) - \virtual(\type) \disutil'(j)$ denote the virtual ad profits under the two profiles of ad characteristics.
For every ad $j \in \adset$, we have $\rho'_\type(j) \ge \rho_\type(j)$, because $\rev'(j) \ge \rev(j)$, $\disutil'(j) \ge \disutil(j)$, and $-\virtual(\type) > 0$.}

{Let $\rho^{(k)}_\type$ and $\rho^{\prime(k)}_\type$ denote the $k$-th highest values of $\rho_\type(\cdot)$ and $\rho'_\type(\cdot)$ among ads in $\adset$, with $\rho^{(k)}_\type := 0$ and $\rho^{\prime(k)}_\type := 0$ for $k > \adnumber$.
For $k\le\adnumber$, at least $k$ number of ads $j$ satisfy $\rho'_\type(j)\ge\rho_\type(j)\ge\rho^{(k)}_\type$, so $\rho^{\prime(k)}_\type\ge\rho^{(k)}_\type$; for $k>\adnumber$, both values are zero by definition.
For each item $k\le\itemnumber$, the virtual-surplus contribution under the $\rho_\type$-negative assortative policy is $\virtual(\type)\quality(k)+\rho^{(k)}_\type$. If $k>\adnumber$, item $k$ has no ad, so $\rho^{(k)}_\type=0$.
Thus this contribution is greater under $\{(\rev'(j),\disutil'(j))\}_{j\in\adset}$ for every item $k$.
As a result, $\bundle(\type) \subseteq \bundle'(\type)$ and the quality threshold $\maxquality(\type)$ is greater under $\{(\rev'(j), \disutil'(j))\}_{j \in \adset}$.} \hfill $\square$

\section{Details for Case 3 in \autoref{sectionApplications}}\label{sectionAppendixAB}
We provide details of the analysis for Case 3 in \autoref{sectionApplications}, which assumes $c>1$ and $\quad \rev(\itemnumber) < \frac{-\virtual(0)\quality(\itemnumber)}{1 - \cost \virtual(0)}$.
Because $c>1$, there is a unique type $\type^* \in (0,1)$ such that $\rho_{\type^*}(j) = \rev(j)(1 -  c \virtual(\type^*) )=0$ for every ad $j$.
We have $\type^*> \zerotype$, because $\rho_{\zerotype}(j) = \rev(j)(1 -  c \virtual(\zerotype)) = \rev(j)>0$.
Thus, any type $\type> \type^*$ receives all items without ads, and  
any type $\type \in [\type^{**}, \type^*)$, where $\type^{**} \in (0, \zerotype)$ is the cutoff derived below, receives all items and every item is matched with some ad.

Type $\type =0$---and generically, a positive measure of types that includes $\type=0$---receives a contract with zero price, because type $0$'s value of any item-ad pair is $0$.
Also, \autoref{theorem} implies that type $\type=0$ receives a (possibly empty) lower-contour bundle.
This bundle excludes at least item $\itemnumber$, which has the highest quality.\footnote{The bundle is empty if and only if $\rev(1)<-\virtual(0)\quality(1)/[1-\cost\virtual(0)]$, where ad $1$ has the highest revenue; at equality, item $1$ is allocated under the tie-breaking rule in \autoref{lemma0}.}
The reason is as follows.
For any $\type< \zerotype$, the virtual ad profit $\rev(j)(1 -  c \virtual(\type) )$ is proportional to $\rev(j)$.
Thus, the negative assortative matching between items and ads implies that item $\itemnumber$ is matched with ad $\itemnumber$, which (by notation) has the $\itemnumber$-th highest ad revenue.
Then, $\rev(\itemnumber) < \frac{-\virtual(0)\quality(\itemnumber)}{1 - \cost \virtual(0)}$ implies that the contribution of item $\itemnumber$ under the optimal item-ad matching to type $0$'s virtual surplus is negative.

For any $\type < \type^*$, we have $1 - \cost\virtual(\type) > 0$, so the virtual ad profit $\rev(j)(1 - \cost\virtual(\type))$ is a positive multiple of $\rev(j)$; the $\rho_\type$-negative assortative policy matches item $k$ with the ad with the $k$-th highest revenue, which we call ad $k$, and by \eqref{equationVirtual2} this pair contributes $h_k(\type) := \rev(k) + \virtual(\type)[\quality(k) - \cost\rev(k)]$ to type $\type$'s virtual surplus, so \autoref{lemma0} allocates item $k$ if and only if $h_k(\type) \ge 0$.
Condition \eqref{equationNON} gives $\quality(k) - \cost\rev(k) = \quality(k) - \disutil(k) \ge 0$, so $h_k$ is increasing in $\type$; and $h_k(\zerotype) = \rev(k) > 0$, because $\rev(k) = 0$ would imply $\disutil(k) = \cost\rev(k) = 0$.
Condition \eqref{equationNetflix} states that $h_\itemnumber(0) < 0$, which requires $\quality(\itemnumber) - \cost\rev(\itemnumber) > 0$ because $\rev(\itemnumber) > 0$ and $\virtual(0) < 0$.
Hence, $h_\itemnumber$ is strictly increasing, and there is a unique $\type^{**} \in (0, \zerotype)$ with $h_\itemnumber(\type^{**}) = 0$, i.e., $\virtual(\type^{**}) = -\rev(\itemnumber)/[\quality(\itemnumber) - \cost\rev(\itemnumber)]$.
For $\type \in [\type^{**}, \zerotype]$ and $k \le \itemnumber$, we have $h_k(\type) \ge h_\itemnumber(\type) \ge h_\itemnumber(\type^{**}) = 0$, because $\virtual(\type) \le 0$, $\quality(k) \le \quality(\itemnumber)$, $\rev(k) \ge \rev(\itemnumber)$, and $1 - \cost\virtual(\type) > 0$; for $\type \in (\zerotype, \type^*)$, $h_k(\type) \ge h_k(\zerotype) > 0$.
Thus, every type in $[\type^{**}, \type^*)$ receives all items, each matched with an ad, whereas every type $\type < \type^{**}$ has $h_\itemnumber(\type) < 0$ and receives a lower-contour bundle that excludes item $\itemnumber$ (\autoref{theorem}).
Finally, let $X(\type) := \sum_{i \in \bundle(\type)} [\quality(i) - \disutil(\matching^*(i|\type))]$.
For every $\type \in [\type^{**}, \type^*)$, $X(\type) = \sum_{i \in \qset} \quality(i) - \cost \sum_{k=1}^{\itemnumber} \rev(k) = X(\type^{**})$, so the transfer in \autoref{lemma0} equals $\transfer(\type) = \type X(\type^{**}) - \int_0^\type X(t) \dint t = \int_0^{\type^{**}} [X(\type^{**}) - X(t)] \dint t$, the same price for every type in this tier.

\section{Proof of \autoref{lemmaInvest} and \autoref{propositionInvest1}} \label{sectionAppendixB}

To prove \autoref{lemmaInvest} and \autoref{propositionInvest1}, we first establish a lemma.  Recall that in the optimal mechanism, the platform's virtual surplus from each type $\type$ is written as the optimal value of the one-to-one matching problem, \eqref{equationVirtual4}, where (i) the platform matches $\itemnumber$ items and $\itemnumber$ ads, (ii) each ad is expressed by its virtual ad profit, and (iii) some ads could be fictitious ads with zero virtual ad profit.
Therefore, we first establish comparative statics for such a problem.

Abusing notation, let $\matchingset^*$ denote the set of all one-to-one exhaustive matching policies between $\itemnumber$ items and $\itemnumber$ ads, where items are represented by the profile of their quality levels, $\qvec$, and ads are represented by the profile of their virtual ad profits, $\rhovec$.
For any $\matching \in \matchingset^*$ and $(i, j) \in \{1,..., \itemnumber\}^2$, $\matching (i, j) =1$ if item $i$ and ad $j$ are matched and $0$ otherwise.

\begin{lemma}\label{lemma1}

Take any submodular function $\vs : \R^2 \to \R$, and
define $$V(\qvec, \rhovec)  := \max_{\matching \in \matchingset^*}\sum_{(i,j) \in \{1,..., \itemnumber\}^2} \vs(\quality(i),\rho(j)) \matching(i,j).$$ 
Then, $V$ has decreasing differences in $(\qvec, \rhovec)$, i.e., for any  $\rhovecH$, $\rhovecL$, $\qvecH$, and $\qvecL$ such that $\rhovecH \ge \rhovecL$ and $\qvecH \ge \qvecL$, we have
$V(\qvecH, \rhovecH)- V(\qvecH, \rhovecL)\le V(\qvecL, \rhovecH)- V(\qvecL, \rhovecL)$.

\end{lemma}
\begin{proof}
The submodularity of $\vs$ implies that the platform adopts negative assortative matching, i.e., the $k$-th lowest $\quality$ is matched with the $k$-th highest $\rho$.
We write $\quality(k)$ for the $k$-th lowest quality in $\qvec$ and $\rho(k)$ for the  $k$-th highest virtual ad profit in $\rhovec$.
Then, we have 
\begin{equation}\label{equationSum}
V(\qvec, \rhovec)  =  \sum^\itemnumber_{k=1} \vs(\quality(k ), \rho(k)).
\end{equation}
Because $\qvecH \ge \qvecL$, 
we have $\qvecH_k \ge \qvecL_k$ for each $k=1,..., \itemnumber$.
Similarly, 
$\rhovecH \ge \rhovecL$ implies $\rhovecH_k \ge \rhovecL_k$.
The submodularity of $\vs$ implies 
\begin{align*}
\vs(\qvecH_k, \rhovecH_k)  - \vs(\qvecH_k, \rhovecL_k)
\le\, \vs(\qvecL_k, \rhovecH_k) - \vs(\qvecL_k, \rhovecL_k) .
\end{align*}
Adding up both sides with respect to $k=1,..., \itemnumber$, we conclude that $V$ has decreasing differences in $(\qvec, \rhovec)$.
\end{proof}

We are now ready to prove \autoref{lemmaInvest}.

\begin{proof}[Proof of Lemma \ref{lemmaInvest}]
Suppose that the advertising set $\adset_H$ attains uniformly higher virtual ad profits than $\adset_L$. As we have shown in the proof of \hyperref[theorem]{Theorem \ref{theorem}}, for each $\type \in \Type$, type $\type$'s virtual surplus $\vs_\type (\quality, \rho)$ is submodular in $(\quality, \rho)$ (trivially so if $\type> \zerotype$).
Thus for any $\type \in \Type$, Lemma \ref{lemma1} implies that 
\begin{equation}\label{equationVS}
\vs_\type (\qvec, \rhovec_\type)  := \max_{\matching \in \matchingset^*}\sum_{(i,j) \in \{1,..., \itemnumber\}^2} \vs_\type (\quality(i),\rho_\type(j))  \matching(i,j)
\end{equation}
has decreasing differences in $(\qvec, \rhovec_\type)$.
Here, $\rhovec_\type \in \R^\itemnumber_+$ is the profile of virtual ad profits (including those of fictitious ads) constructed in the proof of \hyperref[theorem]{Theorem \ref{theorem}} for each $\type$.
 Let $\vs_\type (\qvec, \adset)$ denote the maximized virtual surplus from type $\type$ when written as a function of the quality vector $\qvec$ and advertising set $\adset$.

If the advertising set changes from $\adset_L$ to $\adset_H$, then for every $\type$, the corresponding vector of the top $\itemnumber$ virtual ad profits (ranked in the descending order) increases from $\rhovecL_\type$ to $\rhovecH_\type$.
By the decreasing-difference property of  $\vs_\type (\qvec, \rhovec_\type)$, $\vs_\type (\qvec, \adset_H) - \vs_\type (\qvec, \adset_L)$ decreases in $\qvec$.
Therefore, $\profit(\qvec, \adset_H) - \profit (\qvec, \adset_L)$ also decreases in $\qvec$.

Finally, for any fixed $\adset$, $\profit(\qvec, \adset)$ is trivially supermodular in $\qvec = (\quality(1),..., \quality(\itemnumber))$.
Also, by assumption, $-C(\qvec)$ is supermodular. Thus, $\profit(\qvec, \adset)-C(\qvec)$ is supermodular in $\qvec$. The standard argument of monotone comparative statics establishes the result.

\end{proof}

\begin{proof}[Proof of Proposition \ref{propositionInvest1}]
To show the first part, note that a subscription model arises if and only if $\adset=\emptyset$, whereas a hybrid model arises only if $\adset\not= \emptyset$.
By Part (i) of \autoref{propositionInvest0}, a platform with a subscription model chooses a higher quality profile.

To show the second part, suppose the platform adopts a hybrid model with an advertising set $\adset$.
We then expand $\adset$ to $\hat{\adset}$ by adding $\itemnumber$ ads, each of which has a sufficiently high ad revenue and zero disutility level so that (a) for every type $\type$, its $\theta$-virtual ad profit exceeds any of the $\theta$-virtual ad profits that could arise from $\adset$, and (b) the platform is willing to allocate every item  to all types by matching it with one of the newly added ads.
Because the lowest type is $0$, the platform can implement such an allocation rule only by setting a price of $0$ to all types. 
In such a case, the platform's total revenue equals advertising revenue, i.e., the platform adopts a purely ad-funded model. By Part (i) of \autoref{propositionInvest0}, such a platform chooses a lower quality profile than a hybrid model.

Finally, any purely ad-funded platform must have the same set of optimal quality profiles as the purely ad-funded platform described in the previous paragraph.
Indeed, if two platforms are purely ad-funded, their total revenues are equal to respective advertising revenues, which depend on their advertising sets and the number of allocated items (which equals $\itemnumber$), but not on their quality levels.
Thus, the set of optimal quality profiles for these platforms is equal to $\argmin_{\qvec} C(\qvec)$; otherwise, we obtain a contradiction, because a platform could strictly increase the platform's objective by choosing a quality profile in $\argmin_{\qvec} C(\qvec)$ while maintaining the same optimal mechanism.
\end{proof}

\section{Illustrative Platform Menus}
\label{app:platform_menus}

We provide illustrative examples of platforms that correspond to the three cases discussed in \autoref{sectionApplications}. 
Table~\ref{tab:platform_menus} emphasizes content access and advertising exposure, rather than prices, which change relatively frequently. 
The case labels refer to the theoretical regions in \autoref{sectionApplications}.

\begingroup
\setstretch{1}
\renewcommand{\arraystretch}{1.1}
\footnotesize

\begin{xltabular}{\linewidth}{>{\RaggedRight\arraybackslash}p{2.0cm} >{\RaggedRight\arraybackslash}p{2.8cm} Y Y}
\caption{Illustrative platform menus.}
\label{tab:platform_menus} \\
 
\toprule
\textbf{Platform} & \textbf{Tiers} & \textbf{Content-access margin} & \textbf{Advertising margin} \\
\midrule
\endfirsthead
 
\multicolumn{4}{c}{\tablename\ \thetable{} -- \textit{Continued from previous page}} \\
\toprule
\textbf{Platform} & \textbf{Tiers} & \textbf{Content-access margin} & \textbf{Advertising margin} \\
\midrule
\endhead
 
\bottomrule
\multicolumn{4}{r}{\textit{Continued on next page}} \\
\endfoot
 
\bottomrule
\endlastfoot
 
\multicolumn{4}{l}{\textbf{Case 1: Differentiated Content Access} ($c < 1$)} \\*
\midrule
 
New York Times
& Free access \newline All Access \newline All Access Family
& Free access is limited to a small number of articles per month; All Access provides unlimited access to news together with Cooking, Games, Wirecutter, and The Athletic. All Access Family extends the subscription to multiple household members.
& Advertising is present on both free and subscriber pages; no ad-free tier is offered. \\
\addlinespace[5pt]
 
Financial Times
& Standard Digital \newline FT Digital Edition \newline Premium Digital
& Standard Digital provides full access to FT journalism on ft.com and the FT app, including markets data and portfolio tools. The FT Digital Edition is a standalone digital replica of the printed newspaper, without ft.com or app access. Premium Digital adds the Lex column, mergers-and-acquisitions coverage, and premium-only newsletters such as Moral Money and Trade Secrets.
& Advertising is present in all tiers; no ad-free tier is offered. \\
\addlinespace[10pt]
 
\multicolumn{4}{l}{\textbf{Case 2: Differentiated Ad Exposure} ($c > 1$, high $\rev(\itemnumber)$)} \\*
\midrule
 
YouTube
& Free \newline Premium Lite \newline YouTube Premium
& Full video-catalog access for all users. YouTube Premium adds offline downloads, background play, and YouTube Music Premium.
& Free users see pre-roll, mid-roll, and overlay ads. Premium Lite removes ads from most videos, while ads remain on music content, Shorts, and browsing. YouTube Premium removes ads across the platform. \\
\addlinespace[5pt]
 
X
& Free \newline Basic \newline Premium \newline Premium+
& Free tier provides standard read and write access. Paid tiers progressively unlock post editing, longer posts, priority ranking in replies, creator monetization tools, and higher usage limits on the Grok AI model.
& Free and Basic show the standard ad load. Premium shows approximately half as many ads in the For You and Following timelines. Premium+ provides an ad-free experience across most areas of X, including timelines, replies, and profiles, with occasional sponsored content. \\
\addlinespace[10pt]
 
\multicolumn{4}{l}{\textbf{Case 3: Differentiation in Both Dimensions} ($c > 1$, low $\rev(\itemnumber)$)} \\*
\midrule
 
Netflix
& Standard with Ads \newline Standard \newline Premium
& All plans provide the catalog of movies, TV shows, and games; on Standard with Ads, a small number of titles are unavailable and show a lock icon. Premium adds 4K (Ultra HD) with HDR, spatial audio, more simultaneous streams, and additional extra-member slots.
& Standard with Ads is ad-supported. Standard and Premium are entirely ad-free. \\
\addlinespace[5pt]
 
Peacock
& Select \newline Premium \newline Premium Plus
& Select provides TV series from NBC, Bravo, and other networks, excluding movies, live sports, and Peacock Originals. Premium provides the full lineup, including movies, live sports, and Peacock Originals. Premium Plus adds offline downloads and access to the subscriber's local NBC channel live, 24/7.
& Select and Premium are ad-supported. Premium Plus removes ads with limited exceptions: ads remain in channels, live sports and events, and a small number of titles. \\
 
\end{xltabular}
 
\vspace{0.4em}
\noindent\textit{Note:} The table reports representative menu dimensions rather than a complete description of each platform's pricing or product design. ``Content-access margin'' records whether tiers differ in access to content, content quality, or related products. ``Advertising margin'' records whether tiers differ in advertising exposure. Platform details were accessed in July 2026 from the official help or pricing pages of the New York Times (\url{https://www.nytimes.com/subscription}), Financial Times (\url{https://help.ft.com}), YouTube (\url{https://www.youtube.com/premium}), X (\url{https://help.x.com/en/using-x/x-premium}), Netflix (\url{https://help.netflix.com/en/node/24926}), and Peacock (\url{https://www.peacocktv.com}).

\normalsize
\endgroup

\end{bibunit}

\clearpage

\pagenumbering{arabic}
\renewcommand{\thepage}{S\arabic{page}}
\begin{center}
{\LARGE \textbf{Supplemental Appendix}}\\[6pt]
{\large \textit{For Online Publication}}\\[10pt]
{\large Mechanism Design for Ad-Supported Platforms}\\[5pt]
{\normalsize Shota Ichihashi \qquad Doh-Shin Jeon \qquad Byung-Cheol Kim}
\end{center}

\setcounter{section}{0}
\renewcommand{\thesection}{S\arabic{section}}
\renewcommand{\theequation}{S\arabic{section}.\arabic{equation}}
\setcounter{equation}{0}
\setcounter{theorem}{0}\renewcommand{\thetheorem}{S\arabic{theorem}}
\setcounter{lemma}{0}\renewcommand{\thelemma}{S\arabic{lemma}}
\setcounter{proposition}{0}\renewcommand{\theproposition}{S\arabic{proposition}}
\setcounter{corollary}{0}\renewcommand{\thecorollary}{S\arabic{corollary}}
\setcounter{claim}{0}\renewcommand{\theclaim}{S\arabic{claim}}
\setcounter{assumption}{0}\renewcommand{\theassumption}{S\arabic{assumption}}
\setcounter{definition}{0}\renewcommand{\thedefinition}{S\arabic{definition}}
\setcounter{example}{0}\renewcommand{\theexample}{S\arabic{example}}
\setcounter{remark}{0}\renewcommand{\theremark}{S\arabic{remark}}
\setcounter{note}{0}\renewcommand{\thenote}{S\arabic{note}}
\setcounter{case}{0}\renewcommand{\thecase}{S\arabic{case}}
\setcounter{figure}{0}\renewcommand{\thefigure}{S\arabic{figure}}
\setcounter{table}{0}\renewcommand{\thetable}{S\arabic{table}}
\setcounter{footnote}{0}

\makeatletter
\renewcommand{\@seccntformat}[1]{%
  \@ifundefined{suppappprefix@#1}{}{Appendix~}\csname the#1\endcsname\quad}
\@namedef{suppappprefix@section}{}
\makeatother

\begin{bibunit}

\section{Omitted Materials for \autoref{sectionItemCost}}\label{sectionAppendixItemCost}

\subsection{Item-Specific Distribution and Consumption Costs}\label{sectionAppendixItemCostModel}
In this appendix, the platform incurs a distribution cost of $\kappa_{\mathrm{D}}(i)\geq 0$ when it allocates item $i$, and the consumer incurs a type-independent attention cost $\kappa_{\mathrm{A}}(i)\geq 0$ from consuming item $i$.
These costs enter players' payoffs additively separably.

After optimizing transfers as in the baseline model, the platform's virtual surplus is written as
\[
\int_0^1
\sum_{i\in \bundle(\type)}
\left\{
\virtual(\type)\left[\quality(i)-\disutil(\matching(i|\type))\right]
+\rev(\matching(i|\type))-\kappa(i)
\right\}
\dint \typemeas(\type),
\]
where
$\kappa(i):= \kappa_{\mathrm{D}}(i)+\kappa_{\mathrm{A}}(i)$ {is the total item cost for item $i$}.
The item and ad contributions to the virtual surplus are
\begin{equation}
\alpha_\type(i):= \virtual(\type)\quality(i)-\kappa(i)
\quad\text{and}\quad
\rho_\type(j):= \rev(j)-\virtual(\type)\disutil(j),
\end{equation}
respectively,
with $\rev(\noad)=\disutil(\noad)=\rho_\type(\noad)=0$.

To concisely describe the optimal mechanism, we introduce positive assortative matching between items and ads.
Take any maps $\rho:\adset\to\R$ and $\alpha:\qset\to\R$, which assign each ad $j\in\adset$ an ad score, $\rho(j)$, and each item $i\in\qset$ an item score, $\alpha(i)$.
Recall that a permutation $\mu:\adset\to\adset$ is a $\rho$-permutation if 
\begin{equation*}
\rho(\mu(1))\geq \rho(\mu(2))\geq \cdots \geq \rho(\mu(\adnumber)).
\end{equation*}
Call a permutation $\lambda:\qset\to\qset$ an $\alpha$-permutation if $\lambda$ labels items in such a way that an item with a larger index has a lower score---i.e.,
\begin{equation*}
\alpha(\lambda(1))\geq \alpha(\lambda(2))\geq \cdots \geq \alpha(\lambda(\itemnumber)).
\end{equation*}
If multiple ads or items have the same score under $\rho(\cdot)$ or $\alpha(\cdot)$, the corresponding permutation may not be unique.
\begin{definition}\label{definitionRhoPositiveAssortative}
An advertising policy is a \emph{$\rho$-positive assortative matching} if the following holds for some $\rho$-permutation, $\mu$:
For each $k=1,\ldots,\itemnumber$, item $\itemnumber-k+1$, which has the $k$-th highest quality, is matched with ad $\mu(k)$ if $k\leq\adnumber$ and $\rho(\mu(k))\geq0$, and not matched with any ad otherwise.
\end{definition}
We also introduce positive assortative matching in terms of item scores and virtual ad profits, which simplifies the exposition of the optimal mechanism.\footnote{In either definition, matching an item with an ad whose score is zero is payoff equivalent, in virtual surplus, to matching it with no ad; the definitions select the zero-profit ad whenever it is available.}
\begin{definition}\label{definitionPositiveAssortative}
An advertising policy is a \emph{$(\alpha,\rho)$-positive assortative matching} if the following holds for some $\alpha$-permutation, $\lambda$, and some $\rho$-permutation, $\mu$:
For each $k=1,\ldots,\itemnumber$, item $\lambda(k)$, which has the $k$-th highest item score, is matched with ad $\mu(k)$ if $k\leq\adnumber$ and $\rho(\mu(k))\geq0$, and not matched with any ad otherwise.
\end{definition}

\begin{proposition}\label{propositionItemCost}
An optimal mechanism has the following properties.
\begin{enumerate}
\item For any consumer type $\type \in [0,1]$, the platform adopts an $(\alpha_\type,\rho_\type)$-positive assortative matching.
{There is a type-dependent threshold $\underline{\alpha}(\type)$ such that every item with $\alpha_\type(i) > \underline{\alpha}(\type)$ is allocated, no item with $\alpha_\type(i)<\underline{\alpha}(\type)$ is allocated, and possibly a subset of items with $\alpha_\type(i)=\underline{\alpha}(\type)$ is allocated.}
\item If higher-quality items have higher total item costs, i.e., for any items $i$ and $i'$, $\quality(i)\geq\quality(i')$ implies $\kappa(i)\geq\kappa(i')$, then for negative virtual types (i.e., $\type<\zerotype$), the allocation in {Part 1} is equivalent to adopting a $\rho_\type$-negative assortative matching and allocating all items whose quality levels are below a type-dependent threshold $\maxquality(\type)$, no items whose quality levels are above it, and possibly a subset of items at the threshold, where $\maxquality(\type)$ is increasing in $\type$.
\item If $\kappa(i)$ is constant across items, then for positive virtual types (i.e., $\type>\zerotype$), the allocation in {Part 1} is equivalent to adopting a $\rho_\type$-positive assortative matching and allocating all items whose quality levels are above a type-dependent threshold $\minquality(\type)$, no items whose quality levels are below it, and possibly a subset of items at the threshold, where $\minquality(\type)$ is decreasing in $\type$.
\end{enumerate}
\end{proposition}

\begin{proof}
For a fixed type $\type$, use the definitions of $\alpha_\type$ and $\rho_\type$ to write the relaxed problem as
\begin{equation}\label{equationItemCostRelaxed}
\max_{\matching\in\matchingset}
\sum_{i\in\qset}
\max\left\{0,\alpha_\type(i)+\rho_\type(\matching(i))\right\}.
\end{equation}
As in the baseline characterization, ads with negative $\rho_\type$ are never used because the no-ad option has score zero.
Let $\beta_1(\type)\geq\cdots\geq\beta_{\itemnumber}(\type)$ be the ordered list obtained by taking all nonnegative ad scores $\rho_\type(j)$, adding enough zeros, if necessary, to reach $\itemnumber$ entries, and then keeping the $\itemnumber$ largest entries.
An entry $\beta_k(\type)=0$ means that the corresponding item is matched with a zero-profit ad if such an ad is available at rank $k$, and with no ad otherwise.

The function $m(\alpha,\rho)=\max\{0,\alpha+\rho\}$ is supermodular in $(\alpha,\rho)$.
Thus, \eqref{equationItemCostRelaxed} is solved by an $(\alpha_\type,\rho_\type)$-positive assortative matching.
Because both $\alpha_\type(\lambda(k))$ and $\beta_k(\type)$ are decreasing in $k$, the raw pairwise contribution $\alpha_\type(\lambda(k))+\beta_k(\type)$ is decreasing in $k$.
Choose the payoff-equivalent tie-breaking that allocates item-ad pairs with zero raw contribution.
Then an item at rank $k$ is allocated if and only if $\alpha_\type(\lambda(k))+\beta_k(\type)\geq0$, so the allocated items form an initial segment in the $\alpha_\type$-ranking.
{Equivalently, there is a type-dependent threshold $\underline{\alpha}(\type)$ such that all items with $\alpha_\type(i) > \underline{\alpha}(\type)$ are allocated, no item with $\alpha_\type(i)<\underline{\alpha}(\type)$ is allocated, and possibly a subset of items with $\alpha_\type(i)=\underline{\alpha}(\type)$ is allocated.}

If $\type<\zerotype$ and higher-quality items have higher total item costs, then $\alpha_\type(i)$ is decreasing in $q(i)$: Higher quality lowers $\virtual(\type)\quality(i)$ and raises $\kappa(i)$.
Thus the $(\alpha_\type,\rho_\type)$-positive assortative matching is equivalent to a $\rho_\type$-negative assortative matching in terms of $q(i)$.
The initial segment in the $\alpha_\type$-ranking is a lower set in quality.
Thus, up to possible ties at the cutoff, it can be written as all items with $\quality(i)\leq\maxquality(\type)$, with the convention that this set is empty if $\maxquality(\type)<\quality(1)$.
Index items by increasing quality, and let $\quality_{[k]}$ and $\kappa_{[k]}$ denote the quality and cost of the $k$-th lowest-quality item.
For negative virtual types, the $k$-th lowest-quality item is paired with $\beta_k(\type)$ and has contribution $\virtual(\type)\quality_{[k]}-\kappa_{[k]}+\beta_k(\type)$.
The order statistic $\beta_k(\type)$ is piecewise affine in $\virtual(\type)$, with slope either $-\disutil(j)$ for some ad $j$ or zero.
Wherever this contribution is differentiable in $\virtual(\type)$, its slope is therefore either $\quality_{[k]}-\disutil(j)$ or $\quality_{[k]}$.
The baseline net-quality assumption makes both slopes nonnegative.
Because the contribution is continuous and piecewise affine, it is increasing in $\virtual(\type)$.
Because $\virtual(\type)$ is increasing in $\type$, once the $k$-th lowest-quality item is allocated to some negative virtual type, it remains allocated to all higher negative virtual types.
Therefore $\maxquality(\type)$ can be chosen increasing on $[0,\zerotype)$.

Finally, suppose $\type>\zerotype$ and $\kappa(i)=\kappa$ for every item $i$.
Then $\alpha_\type$ is increasing in quality.
Therefore the $(\alpha_\type,\rho_\type)$-positive assortative matching from Part 1 pairs items and ads in descending order of $\quality(i)$ and $\rho_\type(j)$.
Under the tie-breaking convention in \autoref{definitionRhoPositiveAssortative}, this advertising policy is a $\rho_\type$-positive assortative matching.
The initial segment in the $\alpha_\type$-ranking is an upper set in quality.
Thus, up to possible ties at the cutoff, it can be written as all items with $\quality(i)\geq\minquality(\type)$, with the convention that this set is empty if $\minquality(\type)>\quality(\itemnumber)$.
It remains to verify that $\minquality(\type)$ can be chosen decreasing.
Index items by decreasing quality.
Under positive assortative matching, the $k$-th highest-quality item has contribution $\virtual(\type)\quality^{[k]}-\kappa+\beta_k(\type)$,
where $\quality^{[k]}$ is its quality.
The order statistic $\beta_k(\type)$ is piecewise affine in $\virtual(\type)$, with slope either $-\disutil(j)$ for some ad $j$ or zero.
By the baseline net-quality assumption, the slope of this contribution with respect to $\virtual(\type)$ is either $\quality^{[k]}-\disutil(j)\geq0$ or $\quality^{[k]}\geq0$.
Because the contribution is continuous and piecewise affine, it is increasing in $\virtual(\type)$ and hence in $\type$.
Thus, once the $k$-th highest-quality item is allocated to some positive virtual type, it remains allocated to all higher types.
Therefore $\minquality(\type)$ can be chosen decreasing in $\type$ {on $(\zerotype,1]$}.

The allocation rule that solves the relaxed problem satisfies the monotonicity condition for incentive compatibility.
Suppose, toward a contradiction, that two types $\type_H>\type_L$ satisfy
\[
\sum_{i\in \bundle(\type_L)}\left[\quality(i)-\disutil(\matching(i|\type_L))\right]
>
\sum_{i\in \bundle(\type_H)}\left[\quality(i)-\disutil(\matching(i|\type_H))\right]
\]
under the relaxed solution.
Pointwise optimality for type $\type_L$ gives
\[
\begin{aligned}
&\virtual(\type_L)\sum_{i\in \bundle(\type_L)}
\left[\quality(i)-\disutil(\matching(i|\type_L))\right]
+\sum_{i\in \bundle(\type_L)}\rev(\matching(i|\type_L))
-\sum_{i\in \bundle(\type_L)}\kappa(i)\\
&\qquad\geq
\virtual(\type_L)\sum_{i\in \bundle(\type_H)}
\left[\quality(i)-\disutil(\matching(i|\type_H))\right]
+\sum_{i\in \bundle(\type_H)}\rev(\matching(i|\type_H))
-\sum_{i\in \bundle(\type_H)}\kappa(i).
\end{aligned}
\]
Replacing $\virtual(\type_L)$ by $\virtual(\type_H)$ raises the left-hand side by strictly more than the right-hand side, because $\virtual(\type_H)>\virtual(\type_L)$ and the net-quality sum for $\type_L$ is strictly larger than that for $\type_H$.
This contradicts pointwise optimality for type $\type_H$.
Thus the net-quality sum is nondecreasing in $\type$, and the standard envelope formula recovers the transfer rule.
\end{proof}

Part 1 gives a general description of the optimal allocation for any consumer type.
Once item costs are introduced, the relevant ordering of items is no longer the quality ordering itself but the cost-adjusted item score, $\alpha_\type(i)=\virtual(\type)\quality(i)-\kappa(i)$.
The platform positively assortatively matches this score with virtual ad profits and allocates items above a threshold in this score, up to possible ties at the threshold.
In general, positive assortative matching based on $\alpha_\type(i)$ and $\rho_\type(j)$ does not imply positive or negative assortative matching based on $q(i)$ and $\rho_\type(j)$.

Part 2 shows that, for negative virtual types, this general formula reduces to the pattern described in \autoref{theorem} when higher-quality items are more costly.
In that case, $\alpha_\type(i)$ is decreasing in quality, so the score cutoff is equivalent to an upper cutoff in quality, up to possible ties at the cutoff.

Part 3 shows that under the stronger assumption that the item cost is constant across items, the allocation to positive virtual types becomes the mirror image of that for negative virtual types.
If the common cost is strictly positive, consumers with positive virtual types may no longer receive all items.
Whenever they do not, the allocated set is written as all items whose quality levels are above a threshold.
The item-ad matching also becomes {positively} assortative in item quality and virtual ad profits.
The reason is that, for $\virtual(\type)>0$ and a common cost $\kappa>0$, the contribution $\max\{0,\virtual(\type)\quality-\kappa+\rho\}$ has nontrivial increasing differences in $(\quality,\rho)$; equivalently, the positive assortative matching follows from the supermodularity of this virtual-surplus function.

\paragraph{Investment Incentives.}
The item cost $\kappa(i)$ also changes how advertising affects the platform's investment incentives.
The following example illustrates this.

\begin{example}\label{exampleItemCostInvestment}
Consider a platform with one item and one ad.
The ad has revenue $\rev>0$ and no disutility, $\disutil=0$.
Allocating the item entails a total item cost $\kappa>0$.
Consumer types are uniformly distributed on $[0,1]$, so the virtual type is $\virtual(\type)=2\type-1$.
The platform chooses the item's quality $\quality\geq 0$ at investment cost $c\quality^2/2$, where $c>0$.

Fix $\quality$.
After optimizing transfers, the platform's payoff equals the maximized virtual surplus net of the investment cost:
\begin{equation}\label{equationItemCostInvestmentProfit}
\Pi(\quality,\rev)
=
\int_0^1
\max\left\{(2\type-1)\quality+\rev-\kappa,0\right\}\dint \type
-\frac{c\quality^2}{2}.
\end{equation}
Let $\typemeas^*(x)=\Pr(\virtual(\type)\leq x)$ denote the distribution of the virtual type $x=\virtual(\type)$.
For $\quality>0$, \eqref{equationItemCostInvestmentProfit} can be written as
\begin{equation}\label{equationItemCostInvestmentProfitVirtual}
\Pi(\quality,\rev)
=
\int^1_{(\kappa-\rev)/\quality}
\left(x\quality+\rev-\kappa\right)\dint \typemeas^*(x)
-\frac{c\quality^2}{2},
\end{equation}
where the lower limit is understood to be truncated to the support $[-1,1]$.
Thus, given $\quality$, the platform allocates the item to consumers whose virtual types exceed $(\kappa-\rev)/\quality$, or equivalently to types above
\begin{equation}\label{equationCutoff}
    \type^*(\quality,\rev)
    =
    \frac{1}{2}\left(1+\frac{\kappa-\rev}{\quality}\right),
\end{equation}
again with the cutoff truncated to $[0,1]$.

The marginal effect of quality on the platform's payoff is
\begin{equation}\label{equationItemCostInvestmentMarginal}
\frac{\partial \Pi}{\partial \quality}
=
\int^1_{(\kappa-\rev)/\quality} x\dint \typemeas^*(x)
-c\quality,
\end{equation}
with the same truncation convention.
For each fixed $\quality$, the integral term in \eqref{equationItemCostInvestmentMarginal} is non-monotone in $\rev$ and is single-peaked at $\rev=\kappa$.
{When $\rev<\kappa$, \eqref{equationCutoff} shows that an increase in ad revenue lowers the allocation cutoff, which expands allocation among consumers with positive virtual types and raises the marginal return to quality.}
This effect is absent in our baseline model.
{When $\rev>\kappa$, all consumers with positive virtual types are already served; a further increase in ad revenue expands allocation among negative virtual types, which lowers the marginal return to quality.}
Consequently, when the quality maximizer is unique, the platform's optimal quality choice is single-peaked in ad revenue, with the peak at $\rev=\kappa$ (see \autoref{figureItemCostInvestment}).
\begin{figure}[!htbp]
    \centering
    \begin{tikzpicture}
        \begin{axis}[
            clip=false,
            width=8cm,
            height=6cm,
            xlabel={Ad revenue $\rev$},
            ylabel={Optimal quality $\quality$},
            grid=none,
            axis lines=left,
            line width=0.08pt,
            tick label style={font=\footnotesize},
            label style={font=\footnotesize},
            enlargelimits=upper,
            xmin=0,
            xmax=1,
            ymin=0,
            ymax=1.35
        ]
            \addplot [line width=0.9pt, blue] coordinates {
                (0.000000,1.167459)
                (0.025000,1.182382)
                (0.050000,1.195321)
                (0.075000,1.206529)
                (0.100000,1.216196)
                (0.125000,1.224468)
                (0.150000,1.231454)
                (0.175000,1.237241)
                (0.200000,1.241895)
                (0.225000,1.245467)
                (0.250000,1.247994)
                (0.275000,1.249500)
                (0.300000,1.250000)
                (0.325000,1.249500)
                (0.350000,1.247994)
                (0.375000,1.245467)
                (0.400000,1.241895)
                (0.425000,1.237241)
                (0.450000,1.231454)
                (0.475000,1.224468)
                (0.500000,1.216196)
                (0.525000,1.206529)
                (0.550000,1.195321)
                (0.575000,1.182382)
                (0.600000,1.167459)
                (0.625000,1.150200)
                (0.650000,1.130102)
                (0.670370,1.111111)
            };
            \addplot [line width=0.9pt, blue] coordinates {
                (0.670370,0.000000)
                (1.000000,0.000000)
            };
            \node[fill=black, circle, inner sep=1.5pt, blue] at (axis cs:0.670370,1.111111) {};
            \node[fill=black, circle, inner sep=1.5pt, blue] at (axis cs:0.670370,0.000000) {};
        \end{axis}
    \end{tikzpicture}
    \caption{{The optimal quality choice as a function of ad revenue $\rev\in[0,1]$.
    The figure uses $c=0.2$ and $\kappa=0.3$.
    At the cutoff $\rev=\kappa+2/(27c)\simeq0.670$, both $\quality=0$ and $\quality=2/(9c)\simeq1.111$ are optimal.}}
    \label{figureItemCostInvestment}
\end{figure}
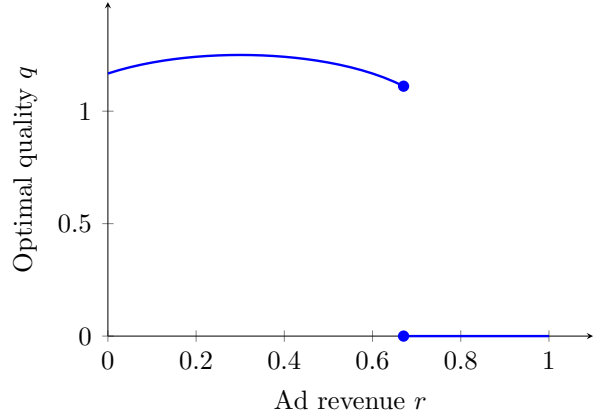
\end{example}

\section{Omitted Materials for \autoref{sectionNoncont}}\label{sectionAppendixFreeDisposal}\label{sectionAppendixD}

{We build on the item-cost environment of Supplemental Appendix~\ref{sectionAppendixItemCost}.
Specifically, we set $\kappa_{\mathrm{D}}(i)=0$ and $\kappa_{\mathrm{A}}(i)=\attention>0$ for every item $i$, so the total item cost is $\kappa(i)=\attention$.
Unlike in the preceding appendix, consumption is noncontractible: Consumers may freely dispose of any allocated item-ad pair.
A type-$\type$ consumer's gross utility from consuming item $i$ matched with ad $j$ is therefore}
\[
    \type\bigl[\quality(i) - \disutil(j)\bigr] - \attention,
\]
{and her payoff from consuming a set of items $\bundle$ under an advertising policy $\matching$ and paying $\transfer$ is}
\begin{equation}\label{equationUtil1}
    \type \sum_{i\in \bundle} \bigl[\quality(i) - \disutil(\matching(i))\bigr] - \attention\,|\bundle| - \transfer,
\end{equation}
{where $|\bundle|$ is the cardinality of $\bundle$.}

{Free disposal, rather than the attention cost alone, creates a new constraint for the platform.
A consumer ignores an item-ad pair whenever $\type[\quality(i)-\disutil(j)]-\attention<0$.
If consumption were contractible, the platform could induce consumption of such a pair by offering a negative price (i.e., $\transfer<0$).
With free disposal, however, the consumer would receive the monetary transfer and then ignore the pair.}

{Accordingly, we replace the incentive compatibility constraint \eqref{IC} with the following stronger constraint, which allows a consumer to report any type and then consume any subset of the resulting bundle:}
\begin{equation}
\begin{multlined}
    \type \sum_{i\in \bundle(\type)} \bigl[\quality(i) - \disutil(\matching(i|\type))\bigr]
    - \attention\,|\bundle(\type)| - \transfer(\type) \\
    \ge
    \max_{\type',\,\bundle'\subseteq \bundle(\type')}
    \Bigl\{\type \sum_{i\in \bundle'} \bigl[\quality(i) - \disutil(\matching(i|\type'))\bigr]
    - \attention\,|\bundle'| - \transfer(\type')\Bigr\},
\end{multlined}
\tag{ICO}\label{ICO}
\end{equation}
{Here, ICO stands for ``incentive compatibility and obedience.''
The left-hand side is the consumer's utility from truthfully reporting $\type$ and consuming every allocated pair.
The right-hand side covers both discarding items after a truthful report and a \emph{double deviation}, whereby the consumer reports $\type'\neq\type$ and consumes a subset $\bundle'\subseteq\bundle(\type')$.
The platform retains the baseline revenue objective and imposes individual rationality using the payoff in \eqref{equationUtil1}; only the IC constraint is replaced by \eqref{ICO}.
In the original baseline model, $\attention=0$ and condition \eqref{equationNON} makes the consumption payoff from every item-ad pair nonnegative, so \eqref{ICO} reduces to \eqref{IC}.
\autoref{sectionItemCost} permits positive attention costs but assumes contractible consumption; free disposal makes obedience an additional requirement here.}

{A general characterization of the optimal mechanism is beyond the scope of the paper, but we obtain characterizations under two sets of assumptions.
The first result rests on the following assumption.}
\begin{assumption}\label{assumptionSimple}
$\rev(j)\geq 0$ is constant and $\disutil(j)=0$ across all $j \in \adset$, and the number of ads is greater than the number of items, i.e., $\adnumber \ge \itemnumber$.
\end{assumption}

\begin{proposition}\label{propositionICO}
Under \autoref{assumptionSimple}, in the optimal mechanism, the platform allocates to each type $\type$ all items whose quality levels $\quality(i)$ belong to some interval $[\minquality(\type), \maxquality(\type)]$, where $\minquality(\type)$ is decreasing, $\maxquality(\type)$ is increasing, and $\maxquality(\type)$ equals the highest possible quality $\quality(\itemnumber)$ whenever $\type$ has a positive virtual value.
Moreover, every allocated item is matched with some ad.
\end{proposition}

{Relative to the common-cost case of \autoref{propositionItemCost}, free disposal adds an obedience constraint to each item-level allocation decision.
For negative virtual types, this constraint trims sufficiently low-quality items from the lower-quality set in that proposition, so the allocated qualities lie in an interval $[\minquality(\type),\maxquality(\type)]$.
For positive virtual types, the allocated set remains an upper set in quality.
The mechanism also preserves two features of \autoref{theorem}: The platform may serve negative virtual types, and, when it does, the upper threshold $\maxquality(\type)$ reflects its incentive to prioritize low-quality items to save information rents.
As $\type$ increases, $\minquality(\type)$ decreases and $\maxquality(\type)$ increases.}

{We next develop the part of the argument common to \autoref{propositionICO} and \autoref{propositionICO2}.
For both results, we solve the platform's problem in two steps.
First, we impose the report-only IC constraint from the item-cost environment together with obedience after a truthful report, but rule out double deviations; we call the resulting problem the \emph{weak problem}.
Second, we verify that its solution satisfies the full constraint \eqref{ICO}.}

The weak problem is written as follows:
\begin{align}
&\max_{\{(\matching(\cdot|\type), \bundle(\type), \transfer(\type))\}_{\type \in \Type}} \int^1_0  \transfer(\type) + \sum_{i \in \bundle(\type)} \rev(\matching(i|\type))
 \dint \typemeas(\type) \quad \tag{P1}\label{P1}     \\[12pt]
\text{subject to} \quad  \quad&\type \sum_{i \in \bundle(\type)} \left[ \quality (i) -  \disutil(\matching(i|\type)) \right] - \attention |\bundle(\type)|  - \transfer(\type) \notag\\
\ge&  \type \sum_{i \in \bundle(\type')} \left[ \quality (i) -  \disutil(\matching(i|\type')) \right] - \attention |\bundle(\type')|  - \transfer(\type'), \forall \type, \type' \in \Type \quad \tag{IC}\label{ICWeak} \\[12pt]
 \quad& \type[ \quality(i)  - \disutil(\matching(i | \type))] - \attention \ge 0, \forall \type, \forall i \in \bundle(\type)
 \quad \tag{O}\label{O}\\[12pt]
 \quad& \type \sum_{i \in \bundle(\type)} \left[ \quality (i) -  \disutil(\matching(i|\type)) \right] - \attention |\bundle(\type)|  - \transfer(\type) \ge 0 , \forall \type \in \Type \quad \tag{IR}\label{IRWeak}
\end{align}
{Constraint \eqref{ICWeak} rules out report-only deviations in the item-cost environment, the obedience constraint \eqref{O} requires consumers to optimally consume every allocated item after reporting truthfully, and \eqref{IRWeak} imposes individual rationality.}

{The weak problem has the same virtual-surplus objective as the problem in Supplemental Appendix~\ref{sectionAppendixItemCost}, specialized to $\kappa_{\mathrm{D}}(i)=0$ and $\kappa_{\mathrm{A}}(i)=\attention$ for every item.
It is equivalent to the following problem:}
\begin{align}
&\max_{\{(\matching(\cdot|\type), \bundle(\type))\}_{\type \in \Type}} \int^1_0  \sum_{i \in \bundle(\type)} \left\{\virtual(\type)[ \quality(i) - \disutil(\matching(i|\type)) ] + \rev(\matching(i|\type)) - \attention\right\}
 \dint \typemeas(\type), \quad \tag{P2}\label{P2}
 \end{align}
{The objective is maximized subject to the obedience constraint \eqref{O} and monotonicity of the allocation.
Using the notation of Supplemental Appendix~\ref{sectionAppendixItemCost}, each allocated item contributes $\alpha_\type(i)+\rho_\type(\matching(i|\type))$, where $\alpha_\type(i)=\virtual(\type)\quality(i)-\attention$ and $\rho_\type(j)=\rev(j)-\virtual(\type)\disutil(j)$.
Thus, relative to that appendix, the new restriction in the weak problem is obedience.
Hereafter, the relaxed problem refers to \eqref{P2} without the monotonicity constraint.}

\subsection{Proof of \autoref{propositionICO}}
We solve \eqref{P2}, verify the monotonicity of allocation, then verify the constraint \eqref{ICO}.
\paragraph{Allocation for Positive Virtual Types.}
Under \autoref{assumptionSimple}, write $\rev$ for the common value of $\rev(j)$.
Because $\adnumber \ge \itemnumber$, for each type we can match every item in $\qset$ with a distinct ad.
For allocated items, doing so increases virtual surplus and does not affect obedience or consumer utility, because $\rev(j)=\rev\ge0$ and $\disutil(j)=0$ for all $j\in\adset$.
We therefore select an optimal solution in which every item is matched with some ad and ties in the item-level allocation decision are resolved in favor of allocation.
The relaxed problem is then separable across items: item $i$ is allocated to type $\type$ if and only if
\begin{equation}\label{equationSimpleAllocation}
\type \quality(i)-\attention \ge 0
\quad\text{and}\quad
\virtual(\type)\quality(i)+\rev-\attention \ge 0.
\end{equation}
{The first inequality in \eqref{equationSimpleAllocation} is obedience, whereas the second is the nonnegative-virtual-surplus condition $\alpha_\type(i)+\rho_\type(j)\geq0$ from Supplemental Appendix~\ref{sectionAppendixItemCost}.}
Take any $\type > \zerotype$.
Because $\type>0$ and $\virtual(\type)>0$, both inequalities in \eqref{equationSimpleAllocation} are preserved as item quality increases.
Hence the allocated set is an upper set in quality: it can be written as all items whose quality levels lie in $[\minquality(\type),\quality(\itemnumber)]$, with the convention that this set is empty if $\minquality(\type)>\quality(\itemnumber)$.
As $\type$ increases, both inequalities in \eqref{equationSimpleAllocation} become easier to satisfy, so $\minquality(\type)$ is decreasing.

\paragraph{Allocation for Negative Virtual Types.}
For the lowest type $\type=0$, no item satisfies obedience, so the platform allocates no item.
Take any $\type \in (0,\zerotype)$.
By \eqref{equationSimpleAllocation}, the obedience constraint imposes a lower bound on item quality.
Because $\virtual(\type)<0$, the virtual-surplus constraint imposes an upper bound on item quality.
Thus, the allocated items have quality levels in a possibly empty interval $[\minquality(\type),\maxquality(\type)]$.
The lower bound $\minquality(\type)$ is decreasing in $\type$ because obedience becomes easier to satisfy.
The upper bound $\maxquality(\type)$ can be chosen increasing in $\type$: if $\type<\hat{\type}<\zerotype$ and item $i$ satisfies $\virtual(\type)\quality(i)+\rev-\attention\ge0$, then it also satisfies $\virtual(\hat{\type})\quality(i)+\rev-\attention\ge0$ because $\virtual(\hat{\type})\ge\virtual(\type)$ and $\quality(i)\ge0$.

\paragraph{Monotonicity.}
Monotonicity requires that the term $\sum_{i \in \bundle(\type)} \left[ \quality (i) -  \disutil(\matching(i|\type)) \right]$, which equals $\sum_{i \in \bundle(\type)} \quality (i)$,
is nondecreasing in $\type$.
Take any $\type'>\type$ and any item $i\in\bundle(\type)$.
By \eqref{equationSimpleAllocation},
\[
\type \quality(i)-\attention \ge 0
\quad\text{and}\quad
\virtual(\type)\quality(i)+\rev-\attention \ge 0.
\]
Because $\type'>\type$, $\quality(i)\ge0$, and $\virtual(\cdot)$ is increasing, both left-hand sides in these inequalities increase when $\type$ is replaced by $\type'$.
Hence $i\in\bundle(\type')$.
This argument applies both within each side of $\zerotype$ and across $\zerotype$: if $\type<\zerotype<\type'$, the same inequalities remain valid after replacing $\type$ by $\type'$.
Therefore, $\bundle(\type)\subseteq\bundle(\type')$ for all $\type'>\type$, and $\sum_{i\in\bundle(\type)}\quality(i)$ is nondecreasing in $\type$.
Thus, the monotonicity of allocation holds.

\paragraph{Solution to the Weak Problem Satisfies \eqref{ICO}.}
We show that the solution to the weak problem (i.e., the allocation rule constructed above and the transfer rule induced by the local IC constraint) satisfies \eqref{ICO}.
The proof of this step consists of two parts.
First, suppose that a consumer has type $\type$ but misreports to be $\hat\type< \type$.
Then, it is optimal for type $\type$ to consume all items in $\bundle(\hat\type)$.
Indeed, any item $i \in \bundle(\hat\type)$ satisfies $\hat\type \quality(i) - \attention \ge 0$, which also satisfies $\type \quality(i) - \attention \ge 0$.
Thus, there is no double deviation that involves downward misreporting.

Second, suppose to the contrary that there is some profitable double deviation for type $\type$ that involves upward misreporting, i.e., type $\type$ misreports to be type $\hat\type> \type$ and consumes some strict subset of allocated items.
{For $x>0$, let $\bundle'(x) \triangleq \{ i \in \bundle(\hat\type): \quality(i)  \ge \attention/x\}$ denote the optimal consumption that arises when type $x$ misreports to be $\hat\type$, and set $\bundle'(0)=\emptyset$.}
Also, define $\grossq^x(\hat\type) \triangleq \sum_{i \in \bundle'(x)} \quality(i)$ and $\size^x(\hat\type) \triangleq |\bundle'(x)|$.

The condition that type $\type$ benefits from this double deviation is written as 
\begin{equation}
\type \grossq^\type(\hat\type) - \attention \size^\type(\hat\type) - \transfer(\hat\type)>
\type \grossq(\type) - \attention \size(\type) - \transfer(\type),
\end{equation}
which we can write as
\begin{equation}\label{equationDouble1}
\type \grossq^\type(\hat\type) - \attention \size^\type(\hat\type) - U(\type) > \transfer(\hat\type),
\end{equation}
where
\begin{equation}
U(\type) = \type \grossq(\type) - \attention \size(\type) - \transfer(\type).
\end{equation}
Also denote
\begin{equation}
U(\hat\type) = \hat\type \grossq(\hat\type) - \attention \size(\hat\type) - \transfer(\hat\type).
\end{equation}
Solving the above equation with respect to $\transfer(\hat\type)$ and plugging it into \autoref{equationDouble1}, we obtain
\begin{equation}\label{equationDouble2}
U(\hat\type) - U(\type)> 
 \hat\type \grossq(\hat\type) - \attention \size(\hat\type)  -[\type \grossq^\type(\hat\type) - \attention \size^\type(\hat\type)].  
\end{equation}
By the local IC, we get
\begin{equation}\label{equationDouble11}
U(\hat\type) - U(\type) = \int^{\hat\type}_\type \grossq(x)\, \dint x.
\end{equation}
To arrange the RHS of \autoref{equationDouble2}, we apply the envelope theorem to
\begin{equation}\label{equationDouble3}
V(y)\triangleq \max_{x \in [\typemin, \typemax]}   y \grossq^x (\hat\type) - \attention \size^x(\hat\type).
\end{equation}
This is the problem of type $y$ that faces the content bundle for type $\hat\type$ and optimally chooses how to truncate the bundle.
Because $y \in \arg\max_{x \in [\typemin, \typemax]}   y \grossq^x (\hat\type) - \attention \size^x(\hat\type)$, the envelope theorem of \citet*{milgrom2002envelope} implies that
\begin{equation}\label{equationDouble44}
\hat\type \grossq(\hat\type) - \attention \size(\hat\type)  -[\type \grossq^\type(\hat\type) - \attention \size^\type(\hat\type)] = V(\hat\type) - V(\type) = \int^{\hat\type}_{\type}\grossq^y (\hat\type) \, \dint y.
\end{equation}
Combining \eqref{equationDouble2}, \eqref{equationDouble11}, and \eqref{equationDouble44}, we obtain 
\begin{equation}\label{equationDouble4}
          \int^{\hat\type}_\type \grossq(x)\, \dint x > \int^{\hat\type}_{\type}\grossq^x (\hat\type) \, \dint x.
\end{equation}

{This is a contradiction.
For every $x\in[\type,\hat\type]$, the nestedness established above gives $\bundle(x)\subseteq\bundle(\hat\type)$.
Moreover, each $i\in\bundle(x)$ satisfies $x\quality(i)-\attention\geq0$, so type $x$ consumes that item after reporting $\hat\type$.
Hence $\bundle(x)\subseteq\bundle'(x)$ and $\grossq(x)\leq\grossq^x(\hat\type)$ for every $x\in[\type,\hat\type]$, contradicting \eqref{equationDouble4}.} \hfill $\square$

We conclude this subsection by presenting an example in which (i) the assumptions for \autoref{propositionICO} (and those for \autoref{propositionICO2} shown below) do not hold and (ii) the approach based on the weak problem fails, because the solution to the relaxed problem (of the weak problem) violates the monotonicity constraint.

\begin{example}\label{example2}
Suppose there is only one item with quality $\quality$ and one ad with disutility level $\disutil$ and ad revenue $\rev$.
Types are binary but the low type has a positive virtual type.
Suppose that 
\begin{equation}
\frac{\attention}{\type_L}+\disutil> \quality >\max\left( \frac{\attention}{\type_H}+\disutil, \frac{\attention}{\virtual(\type_L)}\right).
\end{equation}
(To construct a more concrete example, assume that the fraction of $\type_H$ is low enough so that
$\virtual(\type_L)\approx \type_L$.
Then we can take $\disutil \approx \frac{\attention}{\type_L} - \frac{\attention}{\type_H}$, 
and then $\quality =  \frac{\attention}{\type_H}+\disutil+ \epsilon$.)
Suppose $\rev$ is very high.
Then, the platform allocates the item with the ad to $\type_H$ and without ad to $\type_L$.
{Indeed, $\frac{\attention}{\type_L}+\disutil>\quality> \frac{\attention}{\virtual(\type_L)}$ implies that allocation without the ad is profitable, whereas allocation with the ad violates obedience for type $\type_L$.
The total net quality assigned to $\type_H$ is then $\quality-\disutil$, which is strictly below $\quality$, the total net quality assigned to $\type_L$.
This allocation violates monotonicity.}
\end{example}

\subsection{The Statement and Proof of \autoref{propositionICO2}}

Despite the challenge described in \autoref{example2}, we can accommodate ads with strictly positive disutility levels when types are binary.
\begin{assumption}\label{assumptionICO2}
 Disutility levels are constant across all ads, i.e., there is some $\disutil \ge 0$ such that $\disutil(j) =\disutil \ge 0$ for all $j \in \adset$.
Types are binary (i.e., $\types  = \{\type_L, \type_H\}$ with $\type_L>0$) and the low type has a negative virtual value, i.e., $\type_L  - \frac{1-\beta}{\beta} (\type_H - \type_L) <0$ where $\beta = \Pr(\type = \type_L)$.
\end{assumption}
Unlike \autoref{assumptionSimple}, \autoref{assumptionICO2} does not restrict the relative number of items and ads.
It also allows ad revenues to be heterogeneous across ads.
\begin{proposition}\label{propositionICO2}
Under \autoref{assumptionICO2}, the optimal mechanism is as follows:
\begin{enumerate}
\item 
For type $\type_H$, the platform allocates a set of items whose quality levels exceed some threshold.

\item For type $\type_L$, the platform allocates all items whose quality levels belong to some interval $[\minquality(\type_L), \maxquality(\type_L)]$, and matches
every item with some ad.
\end{enumerate}
\end{proposition}
\begin{proof}

First, we derive the allocation for positive virtual types.
Take the high type, $\type = \type_H$.
Note that under binary types, we have $\virtual(\type_H) = \type_H$.
Recall that under \autoref{assumptionICO2}, all ads have the same disutility level, $\disutil$.
If $\type_H [\quality(i) - \disutil]  - \attention <0$, or equivalently, if $\quality(i) < \frac{\attention}{\type_H}+ \disutil$, then a user will never consume item $i$ when it is matched with an ad.
Also, item $i$ generates a nonnegative virtual surplus for type $\type$ without ads if and only if
$\type_H \quality(i) - \attention \ge 0$.
Combining these inequalities, we conclude that any item $i$ whose quality is in $[\frac{\attention}{\type_H}, \frac{\attention}{\type_H}+ \disutil)$ will be allocated to type $\type$ without ads.

A user will consume any item with $\quality(i) \ge \frac{\attention}{\type_H}+ \disutil$ whether or not it is matched with an ad.
Because $\virtual(\type_H) = \type_H$, the same inequality also ensures that the allocation increases virtual surplus.
Relative to allocating such an item without an ad, matching it with ad $j$ changes virtual surplus by the virtual ad profit $\rho_{\type_H}(j)=\rev(j)-\type_H\disutil$.
Thus, in an optimal solution, the platform matches eligible items only with ads that have nonnegative virtual ad profits, prioritizing ads with higher revenue.
{Thus, the allocation to type $\type_H$ is an upper set in quality with cutoff $\frac{\attention}{\type_H}$, and any item matched with an ad must have quality at least $\frac{\attention}{\type_H}+\disutil$.}

Second, we derive the allocation for negative virtual types.
Consider the low type, $\type_L$ with $\virtual(\type_L)<0$.
First, the platform never allocates items without ads.
Thus, if item $i$ is allocated, it must be matched with an ad and satisfy the obedience constraint
\[
\type_L[\quality(i)-\disutil]-\attention \ge 0,
\]
or equivalently, $\quality(i) \ge \minquality(\type_L)\triangleq \attention/\type_L+\disutil$.
Because disutility is constant across ads, the virtual ad profit is $\rho_{\type_L}(j)=\rev(j)-\virtual(\type_L)\disutil$, whose ranking coincides with the ranking of $\rev(j)$.
Conditional on matching item $i$ with ad $j$, the contribution of the item-ad pair to type $\type_L$'s virtual surplus is
\[
\virtual(\type_L)\quality(i)+\rho_{\type_L}(j)-\attention.
\]
Because $\virtual(\type_L)<0$, this contribution is decreasing in item quality for any fixed ad.
Hence, if an allocated item with quality $\quality(i')>\quality(i)\ge \minquality(\type_L)$ is paired with some ad while item $i$ is not allocated, replacing item $i'$ by item $i$ preserves obedience and one-to-one matching and strictly raises virtual surplus.
It follows that, among items satisfying the lower obedience bound, the allocated set is downward closed in quality.
Therefore, the platform allocates items only up to some upper quality threshold, $\maxquality(\type_L)$.
To sum up, the platform allocates to type $\type_L$ all items whose quality levels are in $[\minquality(\type_L), \maxquality(\type_L)]$.
By construction, every allocated item is matched with some ad.

We now verify the monotonicity of allocation.
Type $\type_H$ receives all items such that $\quality(i) \ge \frac{\attention}{\type_H}$.
In contrast, type $\type_L$ receives a subset of items such that $\quality(i) \ge \frac{\attention}{\type_L}+ \disutil$ and all of them are matched with ads.
Thus, the term $\sum_{i \in \bundle(\type)} \left[ \quality (i) -  \disutil(\matching(i|\type)) \right]$ is higher for type $\type_H$ than type $\type_L$.

Finally, we show that the solution to the weak problem satisfies \eqref{ICO}.
First, if type $\type_H$ misreports to be $\type_L$, then $\type_H$ optimally consumes all items allocated, 
because any item that satisfies $\type_L$'s obedience constraint also satisfies $\type_H$'s obedience constraint.
Thus, type $\type_H$ has no profitable double deviation.

Second, suppose that type $\type_L$ misreports to be $\type_H$.
Note that the IC constraint for $\type_H$ binds at the optimum, i.e., 
\begin{equation*}
\type_H[ \grossq(\type_H) -\grossd(\type_H)] -\attention \size(\type_H)-  \transfer (\type_H) =  \type_H [ \grossq(\type_L) -\grossd(\type_L)] -\attention \size(\type_L) - \transfer (\type_L),
\end{equation*}
where $\grossd(\type)$ is the total disutility level incurred by type $\type$.
Note that the gross payoff (excluding monetary transfer) of type $\type_H$ increases by
$\type_H [ \grossq(\type_H) -\grossd(\type_H)]   -\attention \size(\type_H) - [\type_H (\grossq(\type_L) -\grossd(\type_L))-\attention \size(\type_L) ]$ by reporting $\type_H$ instead of $\type_L$.
Let $\bundle^L(\type_H)$ denote the subset of $\bundle(\type_H)$ that type $\type_L$ optimally consumes after reporting $\type_H$.
Define $\grossq^L(\type_H)$, $\grossd^L(\type_H)$, and $\size^L(\type_H)$ as the total quality, total disutility, and cardinality of $\bundle^L(\type_H)$.
If type $\type_H$ reports truthfully but consumes as if type $\type_L$ would, then type $\type_H$ who adopts such behavior would ignore some item-ad pairs, so the gain decreases and becomes equal to $\type_H  [ \grossq^L(\type_H) -\grossd^L(\type_H)]  -\attention \size^L(\type_H) - [\type_H (\grossq(\type_L) -\grossd(\type_L))-\attention \size(\type_L) ]$.
We have $\grossq^L(\type_H) -\grossd^L(\type_H) \ge \grossq(\type_L) -\grossd(\type_L)$ because the allocation for type $\type_H$ contains more items, and some items do not come with ads.

If type $\type_L$ reports to be $\type_H$ and consumes item-ad pairs optimally, which is the optimal double deviation, the gross gain is 
\begin{align*}
&\type_L  [ \grossq^L(\type_H) -\grossd^L(\type_H)]  -\attention \size^L(\type_H) - [\type_L (\grossq(\type_L) -\grossd(\type_L))-\attention \size(\type_L) ]\\
\le&\type_H(
 [ \grossq^L(\type_H) -\grossd^L(\type_H)]
-(\grossq(\type_L) -\grossd(\type_L))
) -\attention \size^L(\type_H) + \attention \size(\type_L)\\
\le&\transfer(\type_H) - \transfer(\type_L).
\end{align*}
The first inequality comes from $\type_H> \type_L$ and  $\grossq^L(\type_H) -\grossd^L(\type_H) \ge \grossq(\type_L) -\grossd(\type_L)$.
Combining the first and last lines, we conclude that type $\type_L$ does not benefit from a double deviation.
\end{proof}

\section{Omitted Materials for \autoref{sectionItemAdRevenue}}\label{sectionAppendixItemAdRevenue}

\subsection{Arbitrary Item-Ad-Specific Advertising Revenue}
\label{sec:ext-arbitrary-ad-revenue}

We first allow advertising revenue to depend on the exact item-ad pair.
If allocated item $i$ is assigned option $j\in\adset\cup\{\noad\}$, the
platform's advertising revenue is $A(i,j)$, where
$A:\qset\times(\adset\cup\{\noad\})\to\R_+$ is arbitrary subject to
$A(i,\noad)=0$.
The baseline model is the special case in which
$A(i,j)=\rev(j)$ for every $j\in\adset$ and $A(i,\noad)=0$.
The consumer's utility is unchanged.
Under this specification, the platform's problem is
\begin{align}
&\max_{\{( \transfer(\type), \bundle(\type), \matching(\cdot|\type))\}_{\type \in \Type}}
\int^1_0
\transfer(\type)+
\sum_{i \in \bundle(\type)} A(i,\matching(i|\type))
\dint \typemeas(\type)
\tag{M-A}\label{equationMA}\\[12pt]
\text{subject to}\quad
&\type \sum_{i \in \bundle(\type)}
\left[ \quality(i)-\disutil(\matching(i|\type)) \right]
-\transfer(\type)
\nonumber\\
&\qquad\ge
\type \sum_{i \in \bundle(\type')}
\left[ \quality(i)-\disutil(\matching(i|\type')) \right]
-\transfer(\type'),
\quad \forall \type,\type'\in\Type,
\tag{IC-A}\label{equationICA}\\[12pt]
&
\type \sum_{i \in \bundle(\type)}
\left[ \quality(i)-\disutil(\matching(i|\type)) \right]
-\transfer(\type)
\ge 0,
\quad \forall \type\in\Type.
\tag{IR-A}\label{equationIRA}
\end{align}
Only the platform's advertising-revenue term differs from the baseline
problem; the IC and IR constraints are unchanged.

As in the baseline model, the relaxed problem maximizes virtual surplus
pointwise, ignoring the monotonicity constraint on net quality.
The
relaxed problem is
\begin{align}
\max_{\{( \bundle(\type), \matching(\cdot|\type))\}_{\type \in \Type}}
\int^1_0
\sum_{i \in \bundle(\type)}
\left\{
\virtual(\type)
\left[\quality(i)-\disutil(\matching(i|\type))\right]
+A(i,\matching(i|\type))
\right\}
\dint \typemeas(\type).
\tag{P-A}\label{equationPA}
\end{align}
For a fixed type $\type$, separability across types reduces
\eqref{equationPA} to the following assignment problem:
\begin{align}
\max_{\matching \in \matchingset}
\sum_{i \in \qset}
\max\left\{
0,\,
\virtual(\type)
\left[\quality(i)-\disutil(\matching(i))\right]
+A(i,\matching(i))
\right\}.
\tag{R$^A$-$\theta$}\label{equationAtheta}
\end{align}
Problem \eqref{equationAtheta} generally has no closed-form solution.
However, it is a standard maximum-weight bipartite matching problem
after separating item-ad matches from item-only assignments.

\begin{proposition}\label{prop:ext-incremental-assignment}
Fix a type $\type$ and define
\begin{equation*}
    s_i=\virtual(\type)\quality(i),
    \qquad
    b_i=\max\{s_i,0\},
\end{equation*}
and, for each $j\in\adset$,
\begin{equation*}
    W^A_{ij}
    =
    \virtual(\type)\{\quality(i)-\disutil(j)\}
    +A(i,j),
    \qquad
    \Delta^A_{ij}=W^A_{ij}-b_i.
\end{equation*}
Then the value of the pointwise assignment problem \eqref{equationAtheta} is
\begin{equation}
    \label{equationGeneralMatching}
    \begin{aligned}
    \sum_{i\in\qset} b_i
    +
    \max_{x_{ij}\in\{0,1\}}\quad&
    \sum_{i\in\qset}\sum_{j\in\adset}x_{ij}\Delta^A_{ij}\\
    \text{subject to}\quad&
    \sum_{j\in\adset} x_{ij}\le 1
    \quad\forall i\in\qset,\\
    &
    \sum_{i\in\qset} x_{ij}\le 1
    \quad\forall j\in\adset.
    \end{aligned}
\end{equation}
\end{proposition}

\begin{proof}
Let $x_{ij}=1$ if item $i$ is assigned ad $j$, and let
$y_i=1$ if item $i$ is assigned without an ad.
The exact integer program is
\begin{equation*}
    \max_{x,y}
    \sum_{i\in\qset}\sum_{j\in\adset}x_{ij}W^A_{ij}
    +\sum_{i\in\qset} y_i s_i
\end{equation*}
subject to
\begin{equation*}
    x_{ij},y_i\in\{0,1\},\qquad
    y_i+\sum_{j\in\adset}x_{ij}\le 1
    \quad\forall i\in\qset,
    \qquad
    \sum_{i\in\qset}x_{ij}\le 1
    \quad\forall j\in\adset.
\end{equation*}
Fix any feasible choice of item-ad pairs $x$ satisfying the one-to-one
constraints in \eqref{equationGeneralMatching}.
If item $i$ is matched to an ad, feasibility forces $y_i=0$.
If item $i$ is not matched to an ad, the optimal choice is to assign it
without an ad exactly when this contributes positively, yielding payoff
$b_i=\max\{s_i,0\}$.
Therefore the payoff from $x$ after optimizing over $y$ is
\begin{equation*}
    \sum_{i\in\qset} b_i
    +
    \sum_{i\in\qset}\sum_{j\in\adset}x_{ij}(W^A_{ij}-b_i),
\end{equation*}
which is the stated incremental assignment problem.
Conversely, any feasible choice of item-ad pairs $x$, combined with the
optimal no-ad decision for all items unmatched by ads, is feasible in
the pointwise assignment problem \eqref{equationAtheta} and achieves the
displayed value.
\end{proof}

The binary variable $x_{ij}$ indicates whether ad $j$ is matched with
item $i$.
The baseline term $\sum_i b_i$ is the value of assigning items without
ads whenever doing so is profitable; the matching problem chooses which
item-ad assignments to use relative to this baseline.
Because the empty item-ad assignment is feasible after the baseline
$\sum_i b_i$ is included, no item-ad match with
$\Delta^A_{ij}<0$ can appear in an optimum.
Thus, the arbitrary-revenue case can be solved by standard algorithms
for maximum-weight bipartite matching, such as the Hungarian algorithm
\citep*{kuhn1955hungarian,munkres1957algorithms}.

\begin{lemma}\label{lem:ext-net-quality-monotonicity}
For any type $\type$, let
$(\bundle(\type),\matching(\cdot|\type))$ be an optimal solution to
\eqref{equationAtheta}.
Define the total quality and total ad disutility induced by this
allocation by
$Q(\type)=\sum_{i\in\bundle(\type)}\quality(i)$ and
$D(\type)=\sum_{i\in\bundle(\type)}\disutil(\matching(i|\type))$.
If $\virtual(\type')>\virtual(\type)$, then
$Q(\type')-D(\type')\ge Q(\type)-D(\type)$.
\end{lemma}

\begin{proof}
The proof is the same as the proof of \autoref{lemma0} in
\autoref{sectionAppendixA1}, with the item-ad revenue
$\sum_{i\in\bundle(\type)}A(i,\matching(i|\type))$ in place of the
baseline ad revenue $\grossr(\type)$.
\end{proof}

\begin{proposition}[Properties inherited from the baseline model]
\label{prop:ext-arbitrary-inherited-properties}
For any type with \(\virtual(\type)<0\), there is an optimal pointwise allocation in which every allocated item is matched with an ad. 
For any type with \(\virtual(\type)>0\), there is an optimal pointwise allocation in which all items are allocated. Also, optimal allocations for positive virtual types can be selected so that the number of ad-matched items is decreasing in \(\type\).
\end{proposition}

\begin{proof}
Fix a type and write $v=\virtual(\type)$.
Recall from \eqref{equationGeneralMatching} that assigning item $i$
without an ad contributes $s_i=v\quality(i)$ before truncation at zero.
If $v<0$, then $s_i\le0$ for every item because $\quality(i)\ge0$.
Hence any item-only assignment contributes less than leaving the
item unassigned and can be removed without lowering virtual surplus.
Thus an optimal solution can be chosen so that every allocated item is
matched with an ad.

If $v>0$, then $s_i\ge0$ for every item.
After the optimal item-ad matches are chosen, assigning every remaining
item without an ad raises virtual surplus.
Resolving zero-payoff ties in favor of allocation, the type receives all
items.

It remains to compare positive virtual types.
For $v>0$, the preceding paragraph lets us take assignment without ads
as the baseline for every item.
The remaining problem is to choose item-ad matches solving
\begin{equation*}
    \max_{x_{ij}\in\{0,1\}}
    \sum_{i\in\qset}\sum_{j\in\adset}
    x_{ij}\{A(i,j)-v\disutil(j)\},
\end{equation*}
subject to the same one-to-one matching constraints as in
\eqref{equationGeneralMatching}.
Equivalently, for any set $S\subseteq\adset$, define the assignment
value
\begin{equation*}
    V(S)=
    \max_{\substack{x_{ij}\in\{0,1\}\\
    \sum_{j\in S}x_{ij}\le1\ \forall i,\ 
    \sum_{i\in\qset}x_{ij}\le1\ \forall j\in S}}
    \sum_{i\in\qset}\sum_{j\in S}x_{ij}A(i,j).
\end{equation*}
For each item $i$, define the unit-demand valuation
$u_i(S)=\max_{j\in S}A(i,j)$, with the convention $u_i(\emptyset)=0$.
{Allowing $S_i=\emptyset$, we can write}
\begin{equation*}
    V(S)=
    \max_{\substack{S_i\subseteq S\ \forall i\in\qset\\
    S_i\cap S_{i'}=\emptyset\ \forall i\ne i'}}
    \sum_{i\in\qset}u_i(S_i).
\end{equation*}
Indeed, any feasible item-ad matching induces such a disjoint collection
by assigning to each item either the singleton containing its matched ad
or the empty set.
Conversely, because $A(i,j)\ge0$, any disjoint collection can be
reduced, item by item, to the single ad that attains $u_i(S_i)$, or to
the empty set if $S_i=\emptyset$, without reducing its value; this
produces a feasible item-ad matching.
{For the purpose of applying the convolution closure result \citep*[Theorem~9.1]{paesleme2017gross} below, this expression can equivalently be written using partitions of $S$: leftover ads can be assigned to arbitrary items without lowering the value, and, conversely, any partition can be reduced, item by item, to the single ad that attains $u_i$ without changing the value.}
Unit-demand valuations satisfy gross substitutes
\citep*[Section~2]{paesleme2017gross}.
Because $V$ is the finite convolution of the unit-demand valuations
$(u_i)_{i\in\qset}$, repeated application of the binary closure theorem
for gross substitutes under convolution
\citep*[Theorem~9.1]{paesleme2017gross} implies that $V$ satisfies gross
substitutes.
Then the set of ads used at virtual type $v$ solves
\begin{equation*}
    \max_{S\subseteq\adset}
    \left\{V(S)-\sum_{j\in S}v\disutil(j)\right\}.
\end{equation*}
For $p\in\R_+^\adset$, define the demand correspondence by
\begin{equation*}
    \mathcal{D}(p)=
    \operatorname*{arg\,max}_{S\subseteq\adset}
    \left\{V(S)-\sum_{j\in S}p_j\right\}.
\end{equation*}
By the law of aggregate demand for gross-substitutes valuations
\citep*[Section~3.3, display~(LAD)]{murota2016discrete}, if $p'\ge p$
and $S\in\mathcal{D}(p)$, then there exists
$S'\in\mathcal{D}(p')$ such that $|S'|\le |S|$.
Apply this result to $p_j=v\disutil(j)$ and
$p'_j=v'\disutil(j)$.
For $v'>v>0$, we have $p'\ge p$.
Thus we can choose demanded sets of ads at $v$ and $v'$ such that the
latter has smaller cardinality.
Choosing associated maximizing item-ad matchings gives optimal
allocations in which higher positive virtual types have fewer
ad-matched items.
\end{proof}

\subsection{Increasing-Differences Advertising Revenue}
\label{sec:ext-generalized-ad-revenue}
\label{sec:ext-main-matching-pattern}

We now specialize the arbitrary mapping $A(i,j)$ from
\autoref{sec:ext-arbitrary-ad-revenue}.
Suppose that
\begin{equation}
    A(i,j)=R(\rev(j),\quality(i)),
    \qquad j\in\adset.
    \label{equationRspec}
\end{equation}
Throughout the remainder of this appendix, we assume that
$R:\R_+\times\R_+\to\R_+$ has strictly increasing differences,
$\quality(i)\ne\quality(i')$ whenever $i\ne i'$, and
$\rev(j)\ne\rev(j')$ whenever $j\ne j'$.\footnote{Strict increasing
differences means that
$R(r_H,q_H)+R(r_L,q_L)>R(r_H,q_L)+R(r_L,q_H)$ whenever
$r_H>r_L$ and $q_H>q_L$.}
Under this specialization,
\eqref{equationAtheta} is understood with
$A(i,j)=R(\rev(j),\quality(i))$ for $j\in\adset$ and
$A(i,\noad)=0$.
\begin{definition}\label{definitionPAM}
For a fixed type, an allocation is \emph{conditionally positively
assortative} if, whenever two allocated items $i$ and $i'$ are
matched with ads $j$ and $j'$, respectively,
\begin{equation*}
    \quality(i)<\quality(i')
    \quad\Longrightarrow\quad
    \rev(j)\le \rev(j').
\end{equation*}
\end{definition}
Conditional positive assortativity restricts only the pairing within the allocated sets of ad-matched items and ads; it does not restrict which items or ads are allocated.

\begin{proposition}
\label{prop:ext-matching-pattern}
For every type \(\theta\), every optimal solution to the pointwise allocation problem \eqref{equationAtheta} is conditionally positively assortative.
Also, suppose that $R(r,q)$ is increasing in $q$ for each $r$.
For a type with \(v(\theta)>0\), conditional on assigning \(m\) ads, there is an optimal solution in which those ads are assigned to the \(m\) highest-quality items.
\end{proposition}

\begin{proof}
Fix a type $\type$ and an optimal solution.
Suppose that conditional positive assortativity is violated.
Then, among the items that are matched with ads and allocated, there are items $i$ and $i'$ such that
$\quality(i)<\quality(i')$, item $i$ is matched with ad $j$, item $i'$
is matched with ad $j'$, and $\rev(j)>\rev(j')$.
Switching the two ads changes advertising revenue by
\begin{equation*}
    \begin{aligned}
    &R(\rev(j'),\quality(i))+R(\rev(j),\quality(i'))
    -R(\rev(j),\quality(i))-R(\rev(j'),\quality(i')) > 0,
    \end{aligned}
\end{equation*}
where the strict inequality follows from strict increasing differences.
Thus, switching the two ads strictly raises the platform's objective.
Doing so does not affect the consumer's utility, because both the total quality
of items and the total disutility from ads remain the same.
This contradicts optimality.
Therefore, every optimal pointwise allocation is conditionally positively assortative.

Second, fix a type with $\virtual(\type)>0$.
By \autoref{prop:ext-arbitrary-inherited-properties}, we can restrict
attention to optimal allocations in which every item is allocated.
Among such allocations that assign $m$ ads, choose one that maximizes
the total quality of ad-matched items.
Suppose that the ad-matched items are not the $m$ highest-quality items.
Then there are items $i$ and $i'$ such that item $i$ is matched with
some ad $j$, item $i'$ is assigned without an ad, and
$\quality(i')>\quality(i)$.
Move ad $j$ from item $i$ to item $i'$, and assign item $i$ without an
ad.
This change does not affect the consumer's utility or feasibility and
changes advertising revenue by
\begin{equation*}
    R(\rev(j),\quality(i'))-R(\rev(j),\quality(i))\ge0.
\end{equation*}
Thus, there is an optimal allocation in which the $m$ ads are assigned
to the $m$ highest-quality items.
\end{proof}

\autoref{prop:ext-matching-pattern} shows that every optimal pointwise
allocation satisfies a restricted notion of positive assortative matching.
However, the result does not show exactly which items and ads are allocated.
Indeed, the optimal allocation no longer has a simple form that uses a quality-level cutoff or virtual ad profits.
In the next subsection, we provide a simple algorithm that outputs the solution to \eqref{equationAtheta}.
Before turning to the algorithm, we present an example in which increasing-differences advertising revenue changes the comparative statics of \autoref{sectionInnovation}.

\begin{example}\label{exampleMultiplicativeRevenue}
The platform hosts one item and one ad.
Consumer types are uniformly distributed on $[0,1]$, and the cost of choosing quality $\quality$ is $C(\quality)=0.1\quality^{2}$.
The ad has disutility $\disutil=0.85$, and the advertising revenue $R(\rev,\quality)=z\rev\quality$ with $z=0.4$ has strictly increasing differences in $(\rev,\quality)$.
The platform chooses $\quality\ge\disutil$.
Given $(\quality,\rev)$, the value of the pointwise problem \eqref{equationAtheta}, integrated over types, is
\begin{equation*}
    \profit(\quality)
    =\int_{0}^{1}
    \max\left\{0,\; \virtual(\type)\quality,\; \virtual(\type)(\quality-\disutil)+z\rev\quality\right\}
    \dint\type.
\end{equation*}
A consumer of type $\type$ receives the item with the ad if $\virtual(\type)\in[-z\rev\quality/(\quality-\disutil),\,z\rev\quality/\disutil]$, the item without the ad if $\virtual(\type)>z\rev\quality/\disutil$, and nothing otherwise.\footnote{If $\quality=\disutil$, the lower cutoff is $-\infty$, i.e., every type receives the item.
The revenue admits a piecewise closed form: For $\rev>0$, letting $a=\min\{1,z\rev\quality/\disutil\}$ and $b=\max\{-1,-z\rev\quality/(\quality-\disutil)\}$, we have $\profit(\quality)=\frac{1}{2}\left[(\quality-\disutil)\frac{a^{2}-b^{2}}{2}+z\rev\quality(a-b)+\quality\frac{1-a^{2}}{2}\right]$.}
The optimal quality solves $\quality^{*}(\rev)\in\arg\max_{\quality\ge\disutil}\profit(\quality)-0.1\quality^{2}$.

\autoref{fig:ext-multiplicative-quality} plots $\quality^{*}(\rev)$ for $\rev\in[0,1]$; without advertising revenue, $\quality^{*}(0)=1.25$.
As $\rev$ increases, the optimal quality first decreases and, at $\rev\approx0.107$, drops to the minimum feasible level $\quality=\disutil$: The platform switches to full coverage, in which high types purchase the item without the ad and the remaining types receive it with the ad for free.
For larger $\rev$, the optimal quality increases, jumps from $1.02$ to $1.70$ at $\rev=0.75$, and reaches $\quality^{*}(1)\approx2.39$.

The non-monotonicity reflects two forces.
As in \autoref{sectionInnovation}, a more lucrative ad induces the platform to serve types with negative virtual values, which depresses quality: A type that sees the ad contributes $\virtual(\type)+z\rev$ to the marginal value of quality, which is negative for the lowest served types.
At the same time, every impression adds $z\rev$ to the marginal value of quality, and this force dominates for high $\rev$.
\end{example}

\begin{figure}[!htbp]
    \centering
    \begin{tikzpicture}
        \begin{axis}[
            clip=false,
            width=8cm,
            height=6cm,
            xlabel={Ad revenue $\rev$},
            ylabel={Optimal quality $\quality$},
            grid=none,
            axis lines=left,
            line width=0.08pt,
            tick label style={font=\footnotesize},
            label style={font=\footnotesize},
            enlargelimits=upper,
            xmin=0,
            xmax=1,
            ymin=0.7,
            ymax=2.5,
            xtick={0,0.25,0.5,0.75,1}
        ]
            \addplot [dashed, gray, line width=0.5pt] coordinates {(0,1.25) (1,1.25)};
            \addplot [line width=0.9pt, blue] coordinates {
                (0.000000,1.250000)
                (0.010000,1.249989)
                (0.020000,1.249954)
                (0.030000,1.249896)
                (0.040000,1.249815)
                (0.050000,1.249709)
                (0.060000,1.249579)
                (0.070000,1.249423)
                (0.080000,1.249240)
                (0.090000,1.249030)
                (0.100000,1.248791)
                (0.107404,1.248593)
                (0.107404,0.850000)
                (0.120000,0.850000)
                (0.130000,0.850000)
                (0.140000,0.850000)
                (0.150000,0.850000)
                (0.160000,0.850000)
                (0.170000,0.850000)
                (0.180000,0.850000)
                (0.190000,0.850000)
                (0.200000,0.850000)
                (0.210000,0.850000)
                (0.220000,0.850000)
                (0.230000,0.850000)
                (0.240000,0.850000)
                (0.250000,0.850000)
                (0.260000,0.850000)
                (0.270000,0.850000)
                (0.280000,0.850000)
                (0.290000,0.850000)
                (0.300000,0.850000)
                (0.310000,0.850000)
                (0.320000,0.850000)
                (0.330000,0.850000)
                (0.340000,0.850000)
                (0.350000,0.850000)
                (0.360000,0.850000)
                (0.370000,0.850000)
                (0.380000,0.850000)
                (0.390000,0.850000)
                (0.400000,0.850000)
                (0.410000,0.850000)
                (0.420000,0.850000)
                (0.430000,0.850000)
                (0.440000,0.850000)
                (0.450000,0.850000)
                (0.460000,0.850000)
                (0.470000,0.850000)
                (0.480000,0.850000)
                (0.490000,0.850000)
                (0.500000,0.850000)
                (0.510000,0.850000)
                (0.520000,0.850000)
                (0.530000,0.850000)
                (0.540000,0.850000)
                (0.550000,0.850000)
                (0.560000,0.850000)
                (0.570000,0.850000)
                (0.580000,0.850000)
                (0.590000,0.850000)
                (0.600000,0.850000)
                (0.610000,0.850000)
                (0.620000,0.850000)
                (0.630000,0.850000)
                (0.640000,0.850000)
                (0.650000,0.850000)
                (0.660000,0.850000)
                (0.670000,0.850000)
                (0.680000,0.869121)
                (0.690000,0.889229)
                (0.700000,0.909786)
                (0.710000,0.930810)
                (0.720000,0.952322)
                (0.730000,0.974342)
                (0.740000,0.996894)
                (0.750000,1.020000)
                (0.750000,1.700000)
                (0.760000,1.725797)
                (0.770000,1.752219)
                (0.780000,1.779309)
                (0.790000,1.807113)
                (0.800000,1.835674)
                (0.810000,1.865038)
                (0.820000,1.895251)
                (0.830000,1.926361)
                (0.840000,1.958418)
                (0.850000,1.991476)
                (0.860000,2.025590)
                (0.870000,2.060817)
                (0.880000,2.097220)
                (0.890000,2.134866)
                (0.900000,2.173823)
                (0.910000,2.214167)
                (0.920000,2.255979)
                (0.930000,2.292956)
                (0.940000,2.306536)
                (0.950000,2.320163)
                (0.960000,2.333836)
                (0.970000,2.347556)
                (0.980000,2.361322)
                (0.990000,2.375133)
                (1.000000,2.388991)
            };
            \node[fill=black, circle, inner sep=1.5pt, blue] at (axis cs:0.107404, 1.248593) {};
            \node[fill=black, circle, inner sep=1.5pt, blue] at (axis cs:0.107404, 0.85) {};
            \node[fill=black, circle, inner sep=1.5pt, blue] at (axis cs:0.75, 1.02) {};
            \node[fill=black, circle, inner sep=1.5pt, blue] at (axis cs:0.75, 1.70) {};
        \end{axis}
    \end{tikzpicture}
    \caption{The optimal quality $\quality^{*}(\rev)$ when $R(\rev,\quality)=z\rev\quality$ with $z=0.4$, $\disutil=0.85$, $C(\quality)=0.1\quality^{2}$, and $\type\sim U[0,1]$.
    The dashed line marks $\quality^{*}(0)=1.25$; at each jump, the dots mark the two optimal quality levels.}
    \label{fig:ext-multiplicative-quality}
\end{figure}

\subsection{Computing the Pointwise Allocation}
\label{sec:ext-pointwise-problem}

The algorithm operates on \eqref{equationAtheta} type by type.
Fix a type $\type$.
An assigned ad must be matched one-to-one with an assigned item, while
an item may also be assigned without any ad.

If item $i$ is assigned ad $j$, its contribution to
type-$\type$ virtual surplus is
\begin{equation*}
    W_{ij}
    =
    \virtual(\type)\{\quality(i)-\disutil(j)\}
    +R(\rev(j),\quality(i)).
\end{equation*}
If item $i$ is assigned without an ad, its contribution is
\begin{equation*}
    s_i=\virtual(\type)\quality(i).
\end{equation*}
Unassigned items and unused ads contribute zero.
The no-ad option is
therefore treated directly as the item-specific payoff $s_i$.
In the degenerate cases, if $\itemnumber=0$, the value is zero.
If $\adnumber=0$, the value is
$\sum_{i\in\qset}\max\{s_i,0\}$.

\subsubsection{Specialized Weights for the Algorithm}
\label{sec:ext-incremental-assignment}

Before writing the dynamic program, we translate
\autoref{prop:ext-incremental-assignment} into the specialized notation
with $A(i,j)=R(\rev(j),\quality(i))$.
Define
\begin{equation*}
    b_i=\max\{s_i,0\},
    \qquad
    \Delta_{ij}=W_{ij}-b_i.
\end{equation*}
Then \autoref{prop:ext-incremental-assignment} implies that the value of
the pointwise assignment problem \eqref{equationAtheta} is given by
\eqref{equationGeneralMatching}, with $\Delta^A_{ij}$ replaced by
$\Delta_{ij}$.
Thus, after adding the baseline term $\sum_i b_i$, it remains to choose
item-ad matches using the incremental weights $\Delta_{ij}$.

\subsubsection{Dynamic Program}\label{sec:ext-dynamic-program}
Hereafter, without loss of generality, we relabel items and ads so that
\begin{equation*}
    \quality(1)\le \quality(2)\le\cdots\le \quality(\itemnumber)
    \qquad\text{and}\qquad
    \rev(1)\le \rev(2)\le\cdots\le \rev(\adnumber).
\end{equation*}
Using the incremental weights from \autoref{sec:ext-incremental-assignment},
let $G(i,j)$ be the maximum additional payoff, relative to the baseline
$\sum_i b_i$, from any conditionally positively assortative set of item-ad
matches that use items $1,..., i$ and ads $1,..., j$.
We have 
\begin{equation*}
    G(0,j)=G(i,0)=0
    \qquad
    \text{for all }0\le i\le \itemnumber,\ 0\le j\le \adnumber,
\end{equation*}
and for $1\le i\le \itemnumber$ and $1\le j\le \adnumber$,
\begin{equation}
    G(i,j)=
    \max\bigl\{
    G(i-1,j),\,
    G(i,j-1),\,
    G(i-1,j-1)+\Delta_{ij}
    \bigr\}.
    \label{equationDP}
\end{equation}
The three terms correspond to (i) not attaching an ad to item $i$,
(ii) leaving ad $j$ unused, and (iii) matching item $i$ to ad $j$.
In the third case, conditional positive assortativity requires every
other item-ad match in this subproblem to use items $<i$ and ads $<j$.

\paragraph{Algorithm.}\label{sec:ext-algorithm}
We define the algorithm as follows.
\begin{enumerate}
    \item Compute
    \begin{equation*}
    \begin{aligned}
        s_i=\virtual(\type)\quality(i),\qquad
        b_i=\max\{s_i,0\},
        &\qquad i\in\qset,\\
        W_{ij}=\virtual(\type)\{\quality(i)-\disutil(j)\}
        +R(\rev(j),\quality(i)),
        \qquad
        \Delta_{ij}=W_{ij}-b_i,
        &\qquad i\in\qset,\ j\in\adset.
    \end{aligned}
    \end{equation*}

    \item Fill $G$ by the recurrence
    \eqref{equationDP}, with $G(0,j)=G(i,0)=0$, in any order
    consistent with the dependencies.

\item {To recover an allocation, backtrack from $(\itemnumber,\adnumber)$.
The table $G$ gives the optimal value; backtracking selects item-ad pairs that attain that value.
At $(i,j)$, move to any predecessor that attains the maximum.}
Any tie-breaking rule among maximizing predecessors returns an optimal allocation.
A convenient parsimonious convention is to first move to $(i-1,j)$ if
    $G(i,j)=G(i-1,j)$, otherwise move to $(i,j-1)$ if
    $G(i,j)=G(i,j-1)$, and match item $i$ to ad $j$ only if
    neither skip move attains the maximum.
Stop when $i=0$ or $j=0$.%

    \item Assign every item not matched to an ad without any ad if
    $s_i>0$; leave it unassigned if $s_i<0$.
If $s_i=0$, either
    choice is optimal.
\end{enumerate}

\subsubsection{Correctness of the Dynamic Program}
\label{sec:ext-dp-correctness}
The following result shows that the algorithm indeed outputs an optimal allocation for each type.
\begin{proposition}\label{prop:ext-dp-correctness}
The algorithm returns an optimal allocation for
the pointwise assignment problem \eqref{equationAtheta}.
The optimal value is
\begin{equation*}
    \sum_{i\in\qset} b_i+G(\itemnumber,\adnumber).
\end{equation*}
\end{proposition}

\begin{proof}
By \autoref{prop:ext-incremental-assignment}, applied with
$A(i,j)=R(\rev(j),\quality(i))$, solving
\eqref{equationAtheta} is equivalent to maximizing
\begin{equation*}
    \Delta(M)=\sum_{(k,\ell)\in M}\Delta_{k\ell}
\end{equation*}
over {feasible item-ad assignments
$M\subseteq\qset\times\adset$, where each item and each ad appears in at
most one pair}, with the baseline $\sum_i b_i$ added afterward.
By the exchange argument behind \autoref{prop:ext-matching-pattern}, it
is enough to maximize
$\Delta(M)$ over assignments that are conditionally positively
assortative in the sorted labeling.

For $0\le i\le \itemnumber$ and $0\le j\le \adnumber$, let
$\mathcal{C}(i,j)$ be the set of conditionally positively assortative
feasible item-ad assignments that use only items in $\{1,\ldots,i\}$
and ads in $\{1,\ldots,j\}$.
The empty assignment belongs to every
$\mathcal{C}(i,j)$.
Define
\begin{equation*}
    H(i,j)=\max_{M\in\mathcal{C}(i,j)}\Delta(M).
\end{equation*}
We prove by induction on $i+j$ that $G(i,j)=H(i,j)$ for all
$(i,j)$.

If $i=0$ or $j=0$, the only feasible assignment in
$\mathcal{C}(i,j)$ is the empty assignment.
Thus
\begin{equation*}
    H(0,j)=H(i,0)=0=G(0,j)=G(i,0),
\end{equation*}
which proves the base cases.

Now fix $i,j\ge 1$ and suppose the claim holds for all pairs
$(i',j')$ with $i'+j'<i+j$.
We first show that $H(i,j)$ satisfies
the same recurrence as $G(i,j)$.
The lower bound
\begin{equation*}
    H(i,j)\ge
    \max\{H(i-1,j),\,H(i,j-1),\,H(i-1,j-1)+\Delta_{ij}\}
\end{equation*}
is immediate: any assignment in $\mathcal{C}(i-1,j)$ or
$\mathcal{C}(i,j-1)$ is also in $\mathcal{C}(i,j)$, and any assignment
in $\mathcal{C}(i-1,j-1)$ can be augmented by the pair $(i,j)$, which
preserves conditional positive assortativity.

For the reverse inequality, take any $M\in\mathcal{C}(i,j)$.
If
$(i,j)\in M$, then $M\setminus\{(i,j)\}\in\mathcal{C}(i-1,j-1)$, so
\begin{equation*}
    \Delta(M)\le H(i-1,j-1)+\Delta_{ij}.
\end{equation*}
If $(i,j)\notin M$ and item $i$ is not matched to an ad, then
$M\in\mathcal{C}(i-1,j)$, so $\Delta(M)\le H(i-1,j)$.
Finally,
suppose $(i,j)\notin M$ and item $i$ is matched to some ad
$k<j$.
Conditional positive assortativity then implies that ad $j$
is not matched in $M$: if it were matched to some item $\ell<i$, the
lower item $\ell$ would be matched to the higher ad $j$, while the
higher item $i$ is matched to the lower ad $k$, contradicting
conditional positive assortativity.
Hence $M\in\mathcal{C}(i,j-1)$, and
$\Delta(M)\le H(i,j-1)$.
Because these cases exhaust
$\mathcal{C}(i,j)$,
\begin{equation*}
    H(i,j)\le
    \max\{H(i-1,j),\,H(i,j-1),\,H(i-1,j-1)+\Delta_{ij}\}.
\end{equation*}
Therefore $H$ obeys the same recurrence and boundary conditions as
$G$.
The induction hypothesis and \eqref{equationDP} imply
\begin{equation*}
    G(i,j)=H(i,j).
\end{equation*}
This completes the induction.

Taking $i=\itemnumber$ and $j=\adnumber$, $G(\itemnumber,\adnumber)$
is the maximum incremental value over conditionally positively
assortative feasible item-ad assignments.
By \autoref{prop:ext-matching-pattern}, this equals the maximum
incremental value over all feasible item-ad assignments.
Adding the baseline $\sum_i b_i$ gives the value of the pointwise
assignment problem \eqref{equationAtheta} by
\autoref{prop:ext-incremental-assignment}.

It remains only to justify the reconstruction step.
During backtracking,
each move is chosen to attain the maximum in
\eqref{equationDP}.
Skip moves add no item-ad pair; a diagonal
move from $(i,j)$ to $(i-1,j-1)$ adds the pair $(i,j)$.
Because
diagonal moves always move to smaller item and ad indices, the recovered
matching is feasible and conditionally positively assortative.
The sum
of the incremental weights on the recovered diagonal moves is exactly
$G(\itemnumber,\adnumber)$, because the chosen predecessor at each step
accounts for the full value of the current subproblem.
Assigning every
remaining item without an ad if $s_i>0$ and leaving it unassigned if
$s_i<0$ implements the baseline terms $b_i$; if $s_i=0$, either
choice has the same value.
Thus the backtracking rule returns an optimal
allocation.
\end{proof}

\begin{remark}[Pair-dependent advertising nuisance]\label{remarkPairDependentNuisance}
We can apply a similar analysis to the case in which the nuisance from showing ad $j$ next to item $i$ is $\type D(\disutil(j),\quality(i))$, while advertising revenue remains $\rev(j)$ as in the baseline model.
Suppose that $D(0,\quality(i))=0$ for the no-ad option and that the net quality of every feasible item-ad pair is nonnegative, i.e.,
\[
\quality(i)-D(\disutil(j),\quality(i))\geq0 .
\]
For a fixed type, write $v=\virtual(\type)$.
The pointwise relaxed problem is a maximum-weight assignment problem with item-only payoff $v\quality(i)$ and item-ad payoff
\[
v\{\quality(i)-D(\disutil(j),\quality(i))\}+\rev(j).
\]
Equivalently, after setting $b_i=\max\{v\quality(i),0\}$, the incremental value of matching item $i$ with ad $j$ is
\[
\Delta_{ij}=v\{\quality(i)-D(\disutil(j),\quality(i))\}+\rev(j)-b_i .
\]
Thus, we can use the standard maximum-weight assignment algorithms to compute the solution to the relaxed problem.

If $D$ has increasing differences in $(\disutil,\quality)$, the optimal allocation exhibits a conditional sorting property: 
If a consumer has a positive virtual type, the allocation is conditionally negatively assortative in $(\quality,\disutil)$.
If a consumer has a negative virtual type, the allocation is conditionally positively assortative in $(\quality,\disutil)$.\footnote{Formally, fix a type. In this extension, an allocation is conditionally positively assortative in $(\quality,\disutil)$ if, whenever allocated items $i$ and $i'$ are matched with ads $j$ and $j'$, respectively, $\quality(i)<\quality(i')$ implies $\disutil(j)\leq\disutil(j')$. It is conditionally negatively assortative in $(\quality,\disutil)$ if the same premise implies $\disutil(j)\geq\disutil(j')$. These definitions only restrict matching conditional on the set of ad-matched items and the set of ads used. For positive virtual types, fixing these sets, the platform wants to minimize total nuisance, which is attained by conditionally negative assortative matching when $D$ has increasing differences.}
\end{remark}

\section{Discrete Type Space}\label{sectionAppendixDiscreteType}
\subsection{Finite-Type Environment}

The type space is the only primitive changed from the baseline model.
A unit mass of consumers has types in
\[
\types=\{\type_1,\ldots,\type_N\}\subset \Type,
\qquad
0\leq \type_1<\cdots<\type_N\leq 1.
\]
Let $p_n=\Pr(\type=\type_n)>0$ denote the probability of type $\type_n$, with $\sum_{n=1}^N p_n=1$, and let
\[
\typemeas_n=\Pr(\type\leq \type_n)=\sum_{m=1}^n p_m
\]
denote the cumulative distribution.
For $n<N$, write $\Delta_n=\type_{n+1}-\type_n$ for the gap between adjacent types.
The platform knows the type distribution.

A contract is a tuple $(\transfer,\bundle,\matching)$, where $\transfer\in\R$ is the monetary transfer from a consumer to the platform, $\bundle\subseteq \qset$ is the set of allocated items, and $\matching$ is an advertising policy.
As in the baseline model, an advertising policy is a map $\matching:\qset\to \adset\cup\{\noad\}$ such that no two distinct items are assigned the same element of $\adset$.
Let $\matchingset$ denote the set of all advertising policies.
A consumer of type $\type_n$ who chooses contract $(\transfer,\bundle,\matching)$ obtains
\[
\type_n \sum_{i\in \bundle}\left[\quality(i)-\disutil(\matching(i))\right]-\transfer .
\]

\subsection{The Platform's Problem}

The platform chooses a direct mechanism, represented by $\{(\transfer_n,\bundle_n,\matching_n)\}_{n=1}^N$, which specifies for each type $\type_n$ a contract $(\transfer_n,\bundle_n,\matching_n)$.
The platform maximizes total revenue subject to incentive compatibility and individual rationality constraints.
The resulting problem is the finite-type counterpart of \hyperref[M]{Problem~\eqref{M}}:
\begin{align}
&\max_{\{(\transfer_n,\bundle_n,\matching_n)\}_{n=1}^N}
\sum_{n=1}^N p_n
\left[
\transfer_n+\sum_{i\in \bundle_n}\rev(\matching_n(i))
\right] \tag{M-D}\label{equationMD}\\[10pt]
\text{subject to}\quad
&\type_n \sum_{i\in \bundle_n}\left[\quality(i)-\disutil(\matching_n(i))\right]-\transfer_n
\geq
\type_n \sum_{i\in \bundle_m}\left[\quality(i)-\disutil(\matching_m(i))\right]-\transfer_m,
\quad \forall n,m, \tag{IC-D}\label{equationICD}\\[10pt]
&\type_n \sum_{i\in \bundle_n}\left[\quality(i)-\disutil(\matching_n(i))\right]-\transfer_n
\geq 0,\quad \forall n . \tag{IR-D}\label{equationIRD}
\end{align}

\subsection{Virtual Values and the Relaxed Problem}

For each type $\type_n$, define the net allocation and the advertising revenue by
\[
X_n
=
\sum_{i\in \bundle_n}\left[\quality(i)-\disutil(\matching_n(i))\right],
\qquad
\grossr_n
=
\sum_{i\in \bundle_n}\rev(\matching_n(i)).
\]
The finite-type virtual value is
\begin{equation}\label{equationDiscreteVirtual}
\virtual_n
=
\begin{cases}
\displaystyle
\type_n-\frac{1-\typemeas_n}{p_n}\Delta_n,
& n=1,\ldots,N-1,\\[10pt]
\type_N,
& n=N.
\end{cases}
\end{equation}
This is the discrete analogue of the virtual value $\type-\frac{1-\typemeas(\type)}{\density(\type)}$ in the baseline model.
The term $\frac{1-\typemeas_n}{p_n}\Delta_n$ captures the information-rent cost of serving type $\type_n$: Raising the net allocation $X_n$ by one unit raises the rent of every higher type by $\Delta_n$, and the mass of higher types relative to that of type $\type_n$ is $(1-\typemeas_n)/p_n$.

\begin{assumption}[Discrete regularity]\label{assumptionDiscreteRegular}
The discrete virtual values are strictly increasing:
\[
\virtual_1<\virtual_2<\cdots<\virtual_N.
\]
\end{assumption}
This is the finite-type counterpart of the assumption that $\virtual(\type)$ is strictly increasing.

The following lemma records the standard characterization of implementable allocations and revenue-maximizing transfers with finitely many types (see, e.g., \citealt*{maskin1984}).
Note that the nonnegative net quality of every item-ad pair implies $X_n\geq 0$ for every $n$.

\begin{lemma}[Revenue-maximizing transfers]\label{lemmaDiscreteTransfers}
An allocation with net allocations $(X_n)_{n=1}^N$ is incentive compatible under some transfers if and only if $X_1\leq \cdots\leq X_N$.
Moreover, for any allocation satisfying this monotonicity condition, among all transfers that implement the allocation and satisfy \eqref{equationIRD}, expected transfer revenue is maximized by
\begin{equation}\label{equationDiscreteTransfer}
\transfer_n
=
\type_n X_n-\sum_{m=1}^{n-1}\Delta_m X_m,
\qquad n=1,\ldots,N,
\end{equation}
where the empty sum is zero.
\end{lemma}

\begin{proof}
\emph{Necessity.}
Take any types $\type_m$ and $\type_n$ with $m<n$, and take any transfers under which the allocation is incentive compatible.
The IC constraints for types $\type_n$ and $\type_m$ imply
\[
\type_n X_n-\transfer_n
\geq
\type_n X_m-\transfer_m
\quad\text{and}\quad
\type_m X_m-\transfer_m
\geq
\type_m X_n-\transfer_n.
\]
Adding these inequalities gives
\[
(\type_n-\type_m)(X_n-X_m)\geq 0.
\]
Because $\type_n>\type_m$, we have $X_n\geq X_m$.

\emph{Sufficiency.}
Conversely, suppose that $X_1\leq \cdots\leq X_N$.
Define rents by
\[
\widehat U_1=0,
\qquad
\widehat U_n=\sum_{m=1}^{n-1}\Delta_m X_m
\quad\text{for } n=2,\ldots,N,
\]
and define transfers by $\transfer_n=\type_n X_n-\widehat U_n$.
Take two types $m<n$.
Type $\type_n$ does not choose type $\type_m$'s contract if and only if
$\widehat U_n\geq \widehat U_m+(\type_n-\type_m)X_m$, and
type $\type_m$ does not choose type $\type_n$'s contract if and only if
$\widehat U_m\geq \widehat U_n-(\type_n-\type_m)X_n$.
Thus the two global IC constraints for this pair are equivalent to
\begin{equation}\label{equationDiscreteICBounds}
    (\type_n-\type_m)X_m
\leq
\widehat U_n-\widehat U_m
\leq
(\type_n-\type_m)X_n .
\end{equation}
The proposed rents give
\[
\widehat U_n-\widehat U_m
=
\sum_{\ell=m}^{n-1}\Delta_\ell X_\ell.
\]
For each $\ell=m,\ldots,n-1$, monotonicity gives $X_m\leq X_\ell\leq X_n$.
Hence
\[
\sum_{\ell=m}^{n-1}\Delta_\ell X_m
\leq
\sum_{\ell=m}^{n-1}\Delta_\ell X_\ell
\leq
\sum_{\ell=m}^{n-1}\Delta_\ell X_n.
\]
Because $\sum_{\ell=m}^{n-1}\Delta_\ell=\type_n-\type_m$, these inequalities are exactly the pair of IC bounds in \eqref{equationDiscreteICBounds}.
Thus all global IC constraints hold, so the allocation is incentive compatible under the constructed transfers.

\emph{Revenue-maximizing transfers.}
Let $U_n=\type_n X_n-\transfer_n$ denote the rent of type $\type_n$.
With the allocation fixed, advertising revenue does not depend on transfers, and expected transfers are
\[
\sum_{n=1}^N p_n \transfer_n
=
\sum_{n=1}^N p_n \type_n X_n
-
\sum_{n=1}^N p_n U_n.
\]
Thus maximizing expected transfer revenue is equivalent to minimizing expected rents.
We derive lower bounds that every rent vector satisfying IC and IR must obey, and then verify that the constructed rents attain them.
The IR constraint for the lowest type gives $U_1\geq 0$.
The adjacent downward IC constraint---type $\type_{n+1}$ does not choose type $\type_n$'s contract---gives, for each $n=1,\ldots,N-1$,
\[
U_{n+1}
\geq
U_n+\Delta_n X_n.
\]
Iterating from $U_1\geq 0$ yields
\[
U_1\geq 0,
\qquad
U_n\geq \sum_{m=1}^{n-1}\Delta_m X_m
\quad\text{for } n=2,\ldots,N.
\]
The rents $(\widehat U_n)_{n=1}^N$ defined above attain these lower bounds and satisfy all IC constraints by the sufficiency step.
They also satisfy IR: $\widehat U_1=0$, and $\widehat U_n\geq 0$ for $n\geq 2$ because each $\Delta_m$ and each net allocation $X_m$ is nonnegative.
Because every rent vector satisfying IC and IR lies above the lower bounds and all probabilities are strictly positive, the constructed rents minimize expected rents.
Substituting them into $\transfer_n=\type_n X_n-\widehat U_n$ gives \eqref{equationDiscreteTransfer}.
\end{proof}

Using the transfers in \eqref{equationDiscreteTransfer}, we can write expected transfers in terms of virtual values:
\begin{align}
\sum_{n=1}^N p_n\transfer_n
&=
\sum_{n=1}^N p_n\type_n X_n
-
\sum_{n=1}^N p_n\sum_{m=1}^{n-1}\Delta_m X_m \notag\\
&=
\sum_{n=1}^N p_n\type_n X_n
-
\sum_{m=1}^{N-1}\Delta_m X_m
\sum_{n=m+1}^N p_n \notag\\
&=
\sum_{n=1}^N p_n\type_n X_n
-
\sum_{m=1}^{N-1}(1-\typemeas_m)\Delta_m X_m \notag\\
&=
\sum_{n=1}^{N-1}p_n
\left(
\type_n-\frac{1-\typemeas_n}{p_n}\Delta_n
\right)X_n
+
p_N\type_N X_N \notag\\
&=
\sum_{n=1}^N p_n \virtual_n X_n .
\label{equationDiscreteExpTransfers}
\end{align}
The second equality exchanges the order of summation, and the last equality uses the definition of $\virtual_n$.
Consequently, under the revenue-maximizing transfers, the platform's total revenue equals the expected virtual surplus, $\sum_{n=1}^N p_n \left(\virtual_n X_n+\grossr_n\right)$.
By \autoref{lemmaDiscreteTransfers}, the platform's problem \eqref{equationMD} is equivalent to maximizing the expected virtual surplus subject to the monotonicity of the net allocation.
The relaxed problem drops the monotonicity constraint:
\begin{align}
\max_{\{(\bundle_n,\matching_n)\}_{n=1}^N}
\sum_{n=1}^N p_n
\left\{
\virtual_n
\sum_{i\in \bundle_n}\left[\quality(i)-\disutil(\matching_n(i))\right]
+\sum_{i\in \bundle_n}\rev(\matching_n(i))
\right\}. \tag{P-D}\label{equationPD}
\end{align}
Finally, for type $\type_n$ and ad $j\in \adset$, define the $\type_n$-virtual ad profit as
\begin{equation}\label{equationDiscreteVirtualAd}
\rho_n(j)=\rev(j)-\virtual_n\disutil(j),
\qquad
\rho_n(\noad)=0.
\end{equation}
Relative to the baseline model, the only change is that the discrete virtual value \eqref{equationDiscreteVirtual} replaces the continuous virtual value.

\paragraph{Optimal Mechanism.}
The relaxed problem \eqref{equationPD} has the same form as the relaxed problem \eqref{P} in the baseline model: The discrete virtual value $\virtual_n$ replaces the continuous virtual value $\virtual(\type)$, and the expectation over the finite type distribution replaces the integral over types.
In particular, the objective is separable across types, and the problem for each type $\type_n$ coincides with the pointwise problem \eqref{equationVirtual3} with $\virtual(\type)$ replaced by $\virtual_n$.

As a result, the characterization of the optimal mechanism carries over.
For each type $\type_n$, an optimal mechanism adopts a $\rho_n$-negative assortative policy (\autoref{definitionAssortative}) and allocates all items that generate nonnegative virtual surplus, as in \autoref{lemma0}.
Under \autoref{assumptionDiscreteRegular}, which plays the role of the strictly increasing virtual value in the baseline model, the exchange argument in the proof of \autoref{lemma0} again verifies that the resulting net allocation is nondecreasing in $n$.
Thus, the solution to the relaxed problem, together with the transfers in \eqref{equationDiscreteTransfer}, solves the platform's problem \eqref{equationMD}.
The properties of the optimal mechanism in \autoref{theorem} also hold, with consumers partitioned by the sign of the discrete virtual value $\virtual_n$ rather than by the cutoff type $\zerotype$: Types with $\virtual_n<0$ receive lower-contour bundles in which every allocated item is matched with an ad, and types with $\virtual_n>0$ receive all items and are exposed to fewer ads as their types increase.

The same logic applies to the results that build on the virtual-surplus representation.
For example, \autoref{corollaryThreshold} compares virtual ad profits type by type, and the analysis of the platform's quality choice in \autoref{sectionInnovation} rests on the submodularity of the pointwise virtual surplus and monotone comparative statics.
Both arguments are pointwise in the consumer's type, with expectations over the finite type distribution replacing integrals over types.
Thus, these results extend to the discrete-type model as well.

\section{Example with Negative Net Quality}\label{sectionAppendixPAM}
This appendix illustrates the role of condition \eqref{equationNON} in the optimal advertising policy (see \autoref{sectionDiscussion}).
We assume that consumption is not contractible: 
A consumer would dispose of any pair with negative net quality, so the platform can, without loss, allocate only \ipairs with $\quality(i) - \disutil(j) \ge 0$; call such pairs feasible.
Every allocated pair then yields nonnegative net quality, so the derivation of the optimal mechanism in \autoref{sectionAppendixA1} applies: The platform solves the pointwise problem \eqref{equationVirtual3} for each type $\type$, with the advertising policy restricted to feasible pairs.
We show with an example that the solution can match items and ads positively---rather than negatively---assortatively.

Let $\itemnumber = \adnumber = 3$.
Item qualities are $\quality(i) = i$, and ad disutility levels are $\disutil(1) = 0.5$, $\disutil(2) = 1.6$, and $\disutil(3) = 2.7$.
Every ad generates the same revenue $\rev(j) = \rev > 0$.
The condition \eqref{equationNON} fails because, e.g., $\quality(1) - \disutil(2) < 0$.
The feasible pairs are as follows: Item 1 can be matched only with ad 1, item 2 only with ad 1 or ad 2, and item 3 with any ad.
Thus, the set of ads that an item can carry expands with its quality.

Take any type $\type$ with $\virtual(\type) < 0$.
The virtual ad profit $\rho_\type(j) = \rev - \virtual(\type)\disutil(j)$ is positive and increasing in $\disutil(j)$.
Thus, the $\rho_\type$-negative assortative policy of \autoref{lemma0} matches item 1 with ad 3, item 2 with ad 2, and item 3 with ad 1.
This policy is infeasible: $\quality(1) - \disutil(3) < 0$, so a consumer would dispose of the pair $(1, 3)$.

We claim that if $-\virtual(\type) < 2\rev$---which holds at least for types slightly below $\zerotype$---the optimal advertising policy for type $\type$ instead matches item $i$ with ad $i$ for every $i$, and every matched pair is allocated.
This policy is positively assortative: Higher-quality items carry ads with higher virtual ad profits.

To verify the claim, recall that the platform allocates no item without an ad to type $\type$, because the contribution $\virtual(\type)\quality(i)$ of such an item is negative.
Type $\type$'s virtual surplus from allocating a set of feasible pairs equals
\begin{align*}
\rev \times (\text{number of allocated pairs}) + \virtual(\type) \times (\text{total net quality of allocated pairs}).
\end{align*}
The positively assortative policy is the only feasible policy that displays all three ads, because ad 3 can be carried only by item 3, and then ad 2 only by item 2.
Its total net quality is $0.5 + 0.4 + 0.3 = 1.2$, so it attains the virtual surplus of $3\rev + 1.2\virtual(\type)$; allocating all three pairs is optimal because each pair contributes $\virtual(\type)[\quality(i) - \disutil(i)] + \rev \ge 0.5\virtual(\type) + \rev > 0$.
Any policy that allocates two pairs attains at most $2\rev + 0.7\virtual(\type)$, because the total net quality of any two feasible pairs is at least $[\quality(2) - \disutil(2)] + [\quality(3) - \disutil(3)] = 0.7$.
Similarly, any policy that allocates one pair attains at most $\rev + 0.3\virtual(\type)$, and allocating no pairs attains zero.
The condition $-\virtual(\type) < 2\rev$ implies that $3\rev + 1.2\virtual(\type)$ exceeds all of these bounds, which proves the claim.

The example conveys the following lesson.
For a negative virtual type, the platform would like to attach the most annoying ads to the allocated items, but free disposal implies that such ads can be carried only by high-quality items.
Without condition \eqref{equationNON}, this feasibility constraint can reverse the direction of the optimal matching from negatively to positively assortative.

\begingroup
\catcode`\&=12
\putbib[freemium]
\endgroup
\end{bibunit}

\end{document}